\documentclass[]{article}
\usepackage{authblk}
\let\chapter\section
\usepackage[noend]{algorithmic}
\usepackage[vlined,ruled]{algorithm2e}
\usepackage{amsmath,amsthm,amssymb}
\usepackage{Definitions}
\usepackage{enumitem}
\usepackage{mathtools}
\usepackage{fullpage}

\def\maximize{\mathop{\rm maximize}}

\begin{document}

\title{Multistage Campaigning in Social Networks}

\author[1]{Mehrdad Farajtabar}
\author[2]{Xiaojing Ye}
\author[3]{Sahar Harati}
\author[1]{Le Song}
\author[1]{Hongyuan Zha}

\affil[1]{Georgia Institute of Technology, mehrdad@gatech.edu, \{lsong,zha\}@cc.gatech.edu}
\affil[2]{Georgia State University,  xye@gsu.edu}
\affil[3]{Emory University, sahar.harati@emory.edu}

\date{}

\maketitle

\begin{abstract}
We consider the problem of how to optimize multi-stage campaigning over social networks.
The dynamic programming framework is employed to balance the high present reward
and large penalty on low future outcome in the presence of extensive uncertainties.
In particular, we establish theoretical foundations of optimal campaigning over 
social networks where the user activities are modeled as a multivariate Hawkes 
process, and we derive a time dependent linear relation between the 
intensity of exogenous events and several commonly used objective functions of campaigning. 
We further develop a convex dynamic programming framework for 
determining the optimal intervention policy that prescribes the required level of external 
drive at each stage for the desired campaigning result.
Experiments on both synthetic data and the real-world 
MemeTracker dataset show that our algorithm can steer the user
activities for optimal campaigning much more accurately than baselines.
\end{abstract}

%%%%%%%%%%%%%%%     SECTION     %%%%%%%%%%%%%%%% 
\section{Introduction}

Obama was the first US president in history who successfully leveraged
online social media in presidential campaigning, which has been popularized
and become a ubiquitous approach to electoral politics (such as 
in the on-going 2016 US presidential election) in contrast to
the decreasing relevance of traditional media such as TV and newspapers \cite{west2013air,vergeer2013online}.
The power of campaigning via social media in modern politics is a consequence of 
online social networking being an important part of people's regular daily social lives.
It has been quite common that individuals use social network sites to 
share their ideas and comment on other people's opinions.
In recent years, large organizations, such as governments, public media, 
and business corporations, also start to announce news, spread ideas, and/or post advertisements
in order to steer the public opinion through social media platform. 
There has been extensive interest for these entities to influence the public's 
view and manipulate the trend by incentivizing influential users to
endorse their ideas/merits/opinions at certain monetary expenses or credits.
To obtain most cost-effective trend manipulations, one needs to design an optimal campaigning 
strategy or policy such that quantities of interests, such as influence of opinions, 
exposure of a campaign, adoption of new products, can be maximized or steered 
towards the target amount given realistic budget constraints.

The key factor differentiating social networks from traditional media is \emph{peer influence}.
In fact, events in an online social network can be categorized roughly into
two types: endogenous events where users just respond to the actions of their neighbors 
within the network, and exogenous events where users take actions due to drives
external to the network. Then it is natural to raise the following fundamental questions regarding
optimal campaigning over social networks:
can we model and exploit those event data to steer the online community to a desired exposure level?  
More specifically, can we drive the overall exposure to a campaign to a certain level 
(e.g., at least twice per week per user) by incentivizing a small number of users to take more initiatives? 
What about maximizing the overall exposure for a target group of people? 

More importantly, those exposure shaping tasks are more effective when the interventions are 
implemented in multiple stages. Due to the inherent uncertainty in 
social behavior, the outcome of each intervention may not be fully predictable but can 
be anticipated to some extent before the next intervention happens. A key aspect of 
such situations is that interventions can't be viewed in isolation since one must balance 
the desire for high present reward with the penalty of low future outcome. 

In this paper, the \emph{dynamic programming} framework~\cite{bertsekas1995dynamic} 
is employed to tackle the aforementioned issues.
In particular, we first establish the fundamental theory of optimal campaigning over 
social networks where the user activities are modeled as a multivariate Hawkes 
process (MHP)~\cite{DalVer2007, ZhoZhaSon13} since MHP can capture both endogenous and 
exogenous event intensities. We also derive a time dependent linear relation between the 
intensity of exogenous events and the overall exposure to the campaign. 
Exploiting this connection, we develop a convex dynamic programming framework for 
determining the optimal intervention policy that prescribes the required level of external 
drive at each stage in order for the campaign to reach a desired exposure profile.  
We propose several objective functions that are commonly considered as campaigning 
criteria in social networks. Experiments on both synthetic data and real world network 
of news websites in the MemeTracker dataset show that our algorithms can shape the 
exposure of campaigns much more accurately than baselines.

%%%%%%%%%%%%%%%%%%%%%%%%%%%%%%%%%%%%%%%%%%%%%%%%%%%%

\begin{figure*}[!t]
  % \vspace{-3mm}
  \centering
  \setlength{\tabcolsep}{6pt}
  \begin{tabular}{c}
          %\hspace{-5mm}
          \includegraphics[width=0.75\textwidth]{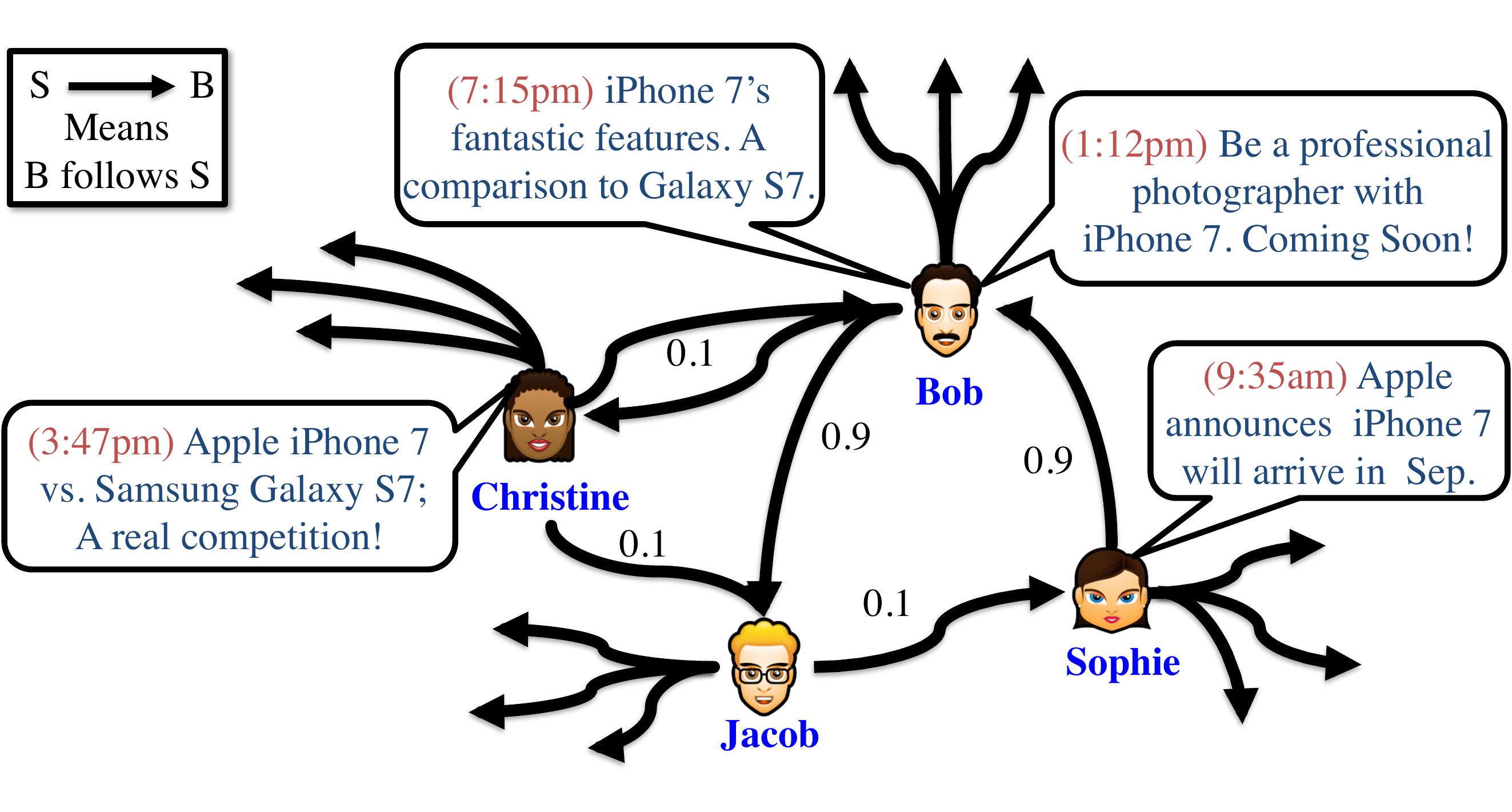} \\ %&
                    a) The social network  and events\\
           \begin{tabular}{cc}
          %\hspace{-4mm}
          \includegraphics[width=0.53\textwidth]{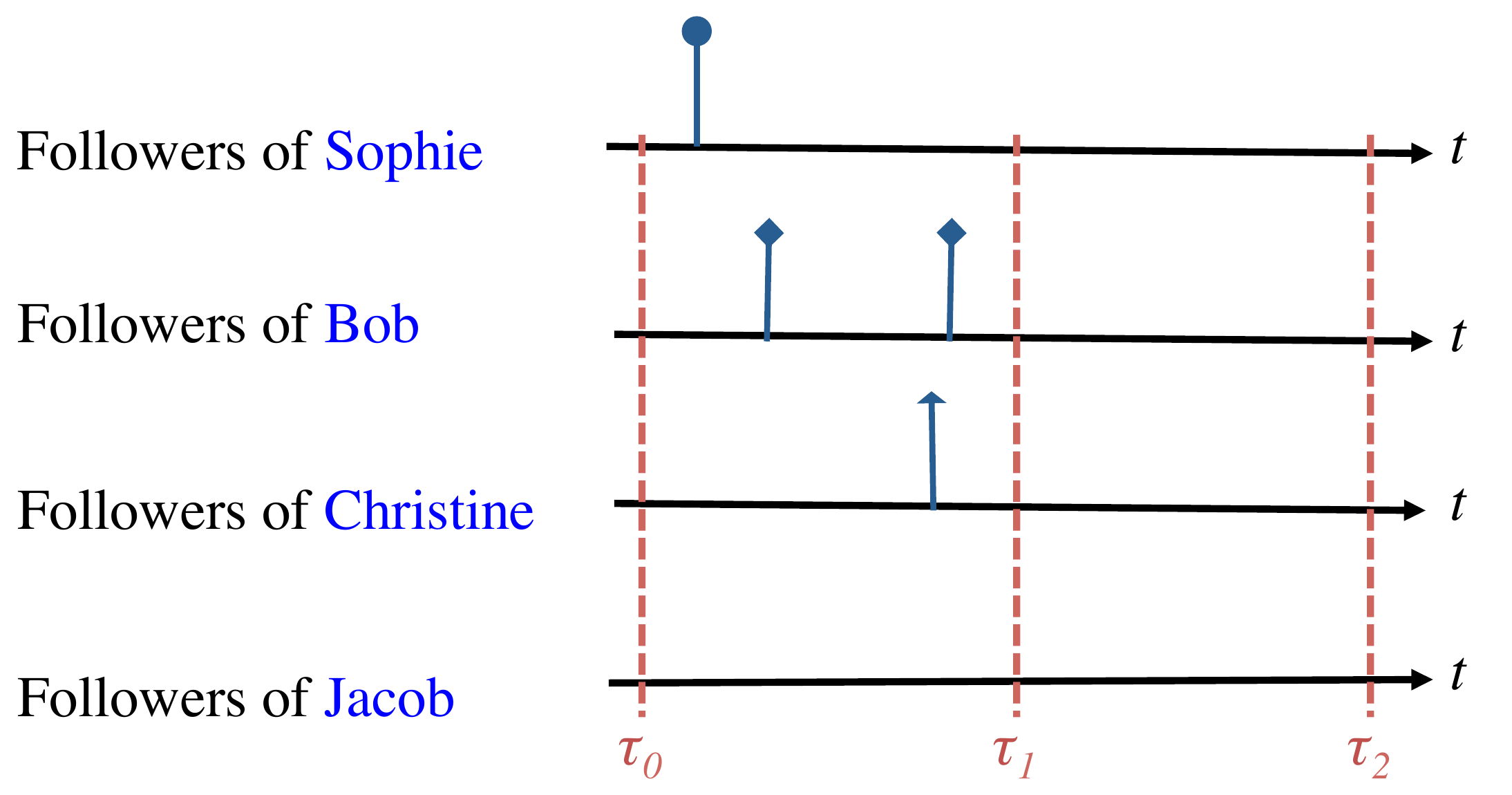} &
          %\hspace{-6mm}
          \includegraphics[width=0.35\textwidth]{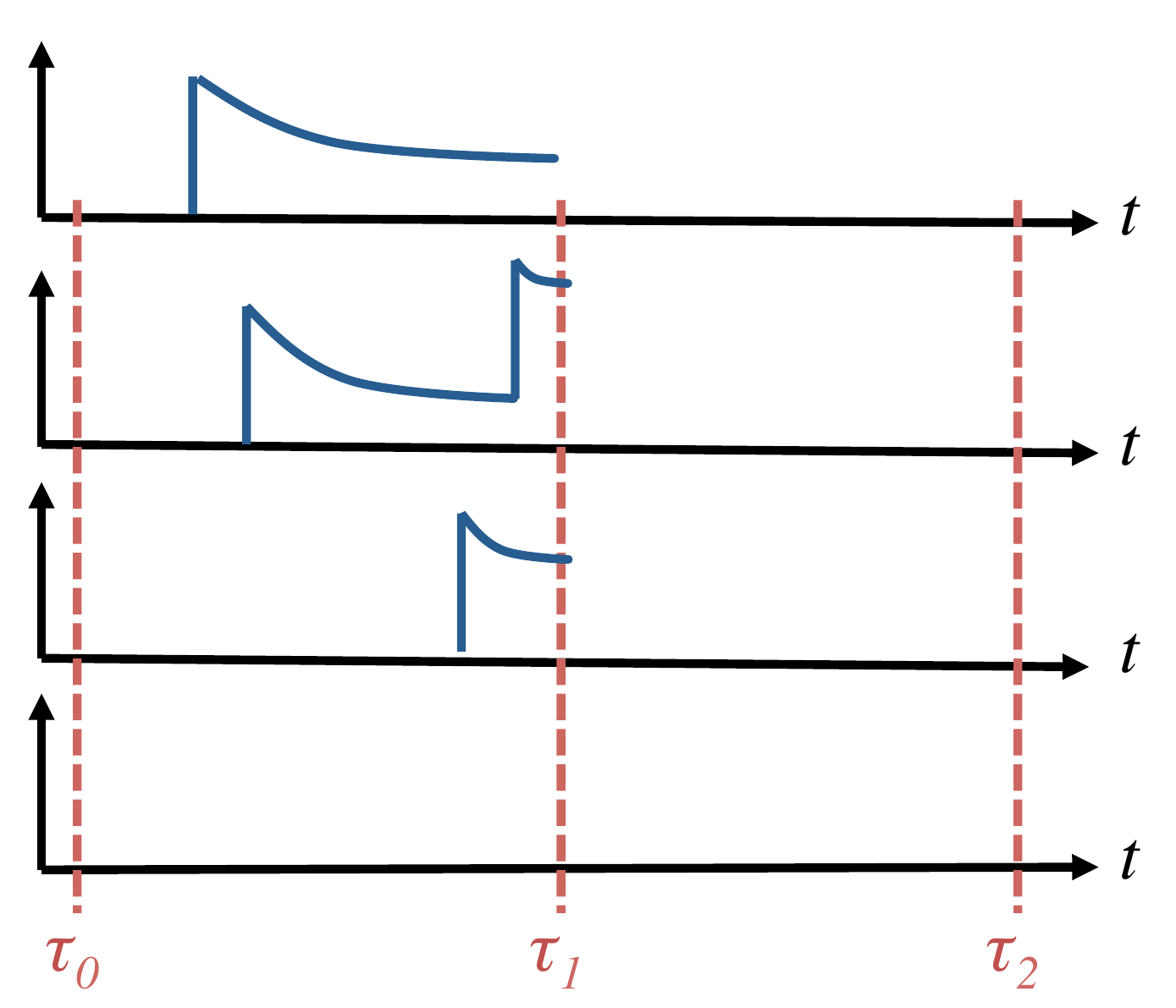}  \\
	   b) exposure events to followers & c) intensity of exposure of followers
           \end{tabular}
  \end{tabular}
  \caption{A social network and multi-stage campaigning to maximize the minimum exposure.}
  \label{fig:illustration}
\end{figure*}

\section{Illustrative Example}
\label{appen-illustration}

Next sections will define the campaigning problem formally. However, we found it beneficial to formulate it more intuitively. Next section will define all the concepts appeared here rigorously.

Fig.~\ref{fig:illustration}-a shows a hypothetical social network highlighting 4 users and their influence on each other. They have their own set of followers which for simplicity are considered disjoint and being influenced equally. Our objective in this toy example is to maximize the minimum exposure on the followers. For the ease of exposition, we only consider 2 stages which in the beginning we can intervene.  Hypothetically, we are going to advertise iPhone 7 on behalf of Apple Inc.

Now, we intuitively proceed to find the optimum intervention. At the first stage, where the state is zero (no exposure from the past), intervention through Jacob and Christine would not be part of the solution (or their share would be negligible), because the outcome would be small due to the low number of total followers and the small amount of influence. Between Bob and Sophie, we would select her. Note that Sophie has a high influence on Bob and if she is incentivized to make a blog post for the campaign, she can make Bob interested too. Therefore, the procedure will assign most budget to Sophie. Furthermore, Bob has a high influence on Jacob and will probably inspire him to do an activity. Let's follow a hypothetical scenario. Sophie makes her post at 9:35am about apple's new phone. Bob is notified about her post and will make an inspiring post about the phone's camera. However, Jacob will not respond to this campaign well, e.g., he may be traveling that day. This phenomenon is perfectly captured via the 
probabilistic framework for event analysis~\cite{AalBorGje08, DalVer2007}.
On the other hand, Christine who is not a big fan of Bob's post compared to Jacob will respond to Bob's post by writing a comparison between Apple and Samsung's product at 3:47pm. This will reinforce Bob to make a separate post explaining the features of Apple's new product. These events are illustrated in Fig.~\ref{fig:illustration}-b. 

The cascade of events at the first stage will expose followers of Sophia, Bob, and Christine to the new product. However, all things do not go according to our expectation. Jacob happens to be out of his office that day and his followers are not aware of the news. Considering the MEM as the objective function, this is quite disappointing that there are people who are exposed to the campaign 0 times. Exposures of followers are demonstrated in Fig.~\ref{fig:illustration}-b.
This is the point where multi-stage dynamic campaigning helps. Considering the exposure intensity of followers in Fig.~\ref{fig:illustration}-c, Christine's, Bob's, and even Sophia's followers are still under the influence. Furthermore, they may follow their initial posts by new ones. Then, in the next stage, it looks reasonable to invest on Jacob. At least his followers will notify about the campaign and also as Jacob has an influence on Sophie, his posts may inspire Sophie again and this may increase the exposure intensity of followers of Sophie as well. Therefore, thanks to a second chance in intervention, Jacob will get the most budget at the second stage to advertise the product.

%%%%%%%%%%%%%%%     SECTION     %%%%%%%%%%%%%%%% 
\section{Basics and Background}
An $n$-dimensional temporal point process is a random process whose realization consists of a list of discrete events in time and their associated dimension, $\cbr{(t_k, d_k)}$ with $t_k \in \RR^+$ and $d_k \in \cbr{1, \ldots, n}$. 
Many different types of data produced in online social networks can be represented as temporal point processes, such as likes and tweets.
A temporal point process can be equivalently represented as a counting process, $\Ncal(t) = (\Ncal^1(t), \ldots, \Ncal^n(t))^{\top}$ associated to $n$ users in the social network. 
Here,  $\Ncal^i(t)$ records the number of events user $i$ performs before time $t$ for $1 \le i \le n$.
Let the history $\Hcal^i(t)$ be the list of times of events $\cbr{t_1, t_2, \ldots, t_k}$ of the $i$-th user up to time $t$. Then, the number of observed events in a small time window $[t, t+dt)$ of length $dt$ is 
\begin{align}
d \Ncal^i(t) = \sum_{t_k \in \Hcal^i(t)} \delta(t-t_k) \, dt,
\end{align} 
and hence $\Ncal^i(t) = \int_0^t d\Ncal^i(s)$, where $\delta(t)$ is a Dirac delta function. %
The point process representation of temporal data is fundamentally different from the discrete time representation typically used in social network analysis. 
It directly models the time interval between events as random variables, avoids the need to pick a time window to aggregate events, and allows temporal events to be modeled in a fine grained fashion. 
Moreover, it has a remarkably rich theoretical support~\cite{AalBorGje08}.  

An important way to characterize temporal point processes is via the conditional intensity function --- a stochastic model for the time of the next event given all the times of previous events. Formally, the conditional intensity function $\lambda^i(t)$ (intensity, for short) of user $i$ is the conditional probability of observing an event in a small window $[t, t+dt)$ given the history $\Hcal(t)= \cbr{\Hcal^1(t), \ldots, \Hcal^n(t)}$:
\begin{align}
  \label{eq:intensity}
  \lambda^i(t)dt := \PP\cbr{\text{user $i$ performs event in $[t, t+dt)$} \, | \, \Hcal(t)} = \EE[d\Ncal^i(t) \, | \, \Hcal(t)], 
\end{align}
% where $^*$ means that the function $\lambda^*(t)$ may depend on the history $\Hcal(t)$. 
where one typically assumes that only one event can happen in a small window of size $dt$. 
Fig.~\ref{fig:illustration}-b and Fig.~\ref{fig:illustration}-c show the events and intensity function of the activities of the 4 users in the social network respectively. Furthermore, these two can be seen as the exposure events and the exposure intensity function to their followers since for simplicity we are just considering the activities of these 4 users in the network.
The functional form of the intensity $\lambda^i(t)$ is often designed to capture the phenomena of interests.

The Hawkes process~\cite{Hawkes71} is a class of self and mutually exciting point process models,
\begin{align}
\label{eq:intensity-hawkes}
 \lambda^i(t) = \mu^i(t) + \sum_{k: t_k < t} \phi^{i d_k} (t, t_k) = 
 \mu^i(t) + \sum_{j=1}^n \int_{0}^t\phi^{ij}(t,s) d\Ncal^j(s),
 \end{align}
where the intensity is history dependent. $\phi^{ij}(t,s)$ is the impact function capturing the temporal influence of an event by user $j$ at time $s$ to the future events of user $j$ at time $t\geqslant s$.  Here, the first term $\mu^i(t)$ is the exogenous event intensity modeling drive outside the network and indecent of the history, and  the second term $\sum_{k: t_k < t} \phi^{i d_k} (t, t_k)$ is the endogenous event intensity modeling interactions within the network~\cite{farajtabar2014activity}.
Defining $\Phi(t,s)=[\phi^{ij}(t,s)]_{i,j = 1\ldots n}$, and $\lambda(t) = (\lambda^1(t) , \ldots, \lambda^n(t))^{\top}$,  and $\mu(t) = (\mu^1(t) , \ldots, \mu^n(t))^{\top}$ we can compactly rewrite Eq~\ref{eq:intensity-hawkes} in matrix form:
\begin{align}
\label{eq:intensity-hawkes-compact}
 \lambda(t) = \mu(t) + \int_{0}^t \Phi(t,s) d\Ncal(s).
 \end{align}
In practice it is standard to employ shift-invariant impact function, 
\ie, $\Phi(t,s) = \Phi(t-s)$. Then, by using notation of convolution 
$f(t) * g(t) = \int_0^t f(t-s) g(s) ds$ we have
\begin{align}
\label{eq:intensity-hawkes-conv}
 \lambda(t) = \mu(t) + \Phi(t) * d\Ncal(t).
\end{align}
 
 %%%%%%%%%%%%%%%     SECTION     %%%%%%%%%%%%%%%% 
\section{From Intensity to Average Activity}
\label{sec:mean}
In this section we  will  develop  a  closed  form  relation  between  the  expected total intensity $\EE[ \lambda(t)]$ and the
intensity $\mu(t)$ of exogenous events.   This relation establish the basis of our campaigning framework. 
First, define the {\it mean function} as
$
 \Mcal(t) :=\EE[ \Ncal(t)]= \EE_{\Hcal(t)}[ \EE(\Ncal(t)|\Hcal(t))].
$ 
Note that $\Mcal(t)$ is history independent,
and it gives the average number of events up to time $t$ for each of the dimension. Similarly, the {\it rate} function $\eta(t)$ is given by
$\eta(t) dt := d \Mcal(t)$. On the other hand, 
\begin{align} 
d \Mcal(t) = 
d \EE[ \Ncal(t) ] 
= \EE_{\Hcal(t)}[ \EE(d\Ncal(t) | \Hcal(t))] = \EE_{\Hcal(t)}[ \lambda(t)|\Hcal(t)]dt
=\EE[\lambda(t)]dt.
\end{align}
Therefore $\eta(t) = \EE[\lambda(t)]$ which serves as a measure of activity in the network. In what follows we will find an analytical form for the average activity. Proofs are presented in Appendix~\ref{appen-proofs}.
%When the exogenous intensity is constant, $\mu(t) = \mu$, according to~\cite{farajtabar2014activity} we have $\eta(t) = \Psi(t) \mu$
% \begin{equation}\label{eqn:defPsi}
%     \Psi(t) =  e^{(A-I)\omega t} + (A-I )^{-1} ( e^{(A-I)\omega t} - I)
%\end{equation}
%Now, we will extend this result to the case where $\mu(t)$ is not constant any more. We start with piecewise constant $\mu(t)$ of the form:
%\begin{align}
%\label{eq:piecewise-constant-intensity} 
%\mu(t) &= \sum_{m=1}^M \mu_m \II(\tau_{m-1} \le t < \tau_m).
%\end{align}
%%
%where $\mu_m > 0$ are constant and $\tau_0 < \tau_1 < \ldots < \tau_M$ are fixed points in time and $\II$ is the indicator function.

\begin{lemma}
Suppose $\Psi:[0,T] \to \mathbb{R}^{n\times n}$ is a non-increasing matrix function, then
for every fixed constant intensity $\mu(t) = c\in\mathbb{R}_{+}^n$, $\eta_c(t):=\Psi(t)c$
solves the semi-infinite integral equation
\begin{equation}\label{eqn:WienerHopf_constant}
\eta(t) = c + \int_{0}^{t} \Phi(t-s)\eta(s)ds,\quad \forall t\in[0,T],
\end{equation}
if and only if $\Psi(t)$ satisfies 
\begin{equation}\label{eqn:Psi_gen}
\Psi(t) = I + \int_{0}^{t} \Phi(t-s)\Psi(s)ds,\quad \forall t\in[0,T].
\end{equation}
In particular, if $\Phi(t)=Ae^{-\omega t}\one_{\geq0}(t)=[a_{ij}e^{-\omega t}\one_{\geq0}(t)]_{ij}$
where $0\leq \omega\notin \text{Spectrum}(A)$, then 
 \begin{equation}\label{eqn:defPsi}
     \Psi(t) =  e^{(A-\omega I) t} + \omega(A-\omega I )^{-1} ( e^{(A-\omega I) t} - I)
\end{equation}
for $t\in[0,T]$, where, $\one_{\geq0}(t)$ is an indicator function for $t \geq 0$.
\end{lemma}

Let $\mu:[0,T] \to\mathbb{R}_{+}^n$ be a right-continuous piecewise constant function
\begin{align}
\label{eq:piecewise-constant-intensity} 
\mu(t) &= \sum_{m=1}^M c_m \one_{[\tau_{m-1},\tau_m)}(t),
\end{align}
where $0=\tau_0<\tau_1<\cdots<\tau_M=T$ is a finite partition of time interval $[0,T]$ and function $\one_{[\tau_{m-1},\tau_m)}(t)$ indicates $\tau_{m-1} \leq t < \tau_m$.
The next theorem shows that if $\Psi(t)$ satisfies \eqref{eqn:Psi_gen},
then one can calculate $\eta(t)$ for piecewise constant intensity $\mu:[0,T]$ of form
\eqref{eq:piecewise-constant-intensity}.

\begin{theorem}
\label{theo:piecewise_constant_average}
Let $\Psi(t)$ satisfy \eqref{eqn:Psi_gen} and $\mu(t)$ be a right-continuous piecewise 
constant intensity function of form \eqref{eq:piecewise-constant-intensity}, then
the rate function $\eta(t)$ is given by
\begin{align}\label{eqn:eta_pwconst}
\eta(t) = \sum_{k=0}^m \Psi(t-\tau_k) (c_{k}-c_{k-1}),
\end{align}
for all $t\in (\tau_{m-1},\tau_m]$ and $m=1,\dots,M$, where $c_{-1}:=0$ by convention.
\end{theorem}

%Let $\mu_{-1} \triangleq \mu_{0^-}=0$ then we can rewrite the above as
%\begin{align}
%\eta(t) = \sum_{k=0}^m \Psi(t-\tau_k) (\mu_{k}-\mu_{k-1}).
%\end{align} 
%
%To generalize the above results to the general case $\mu(t)$ we consider the Rimmanian approximation. We divide the interval $[0, T]$ into $M$ interval of length $\Delta_M = T/M$ and use the value of the function in the middle of the interval to approximate the function in that interval:
%\begin{align}
%\mu(t) \approx \sum_{m=1}^M \mubar_m \II((m-1)\Delta_M < t < m \Delta_M).
%\end{align}
%wehre $\mubar_m = \mu(\tau_m)$.
%It's easy to see that by tending $\Delta_M \to 0$ or equivalently $M \to \infty$ we can be as close as we want to the real function.
%Then, we define $\etabar(t)$ as the average intensity when we apply the piece-wise constant approximated intensity:
%\begin{align}
%\etabar(t) =   \sum_{k=0}^{\floor{t/\Delta_M}} \Psi(t-k \Delta_M) (\mubar_{k}-\mubar_{k-1}).
%\end{align}

%\begin{lemma}
%If $\eta(t)$ is the average intensity when we apply $\mu(t)$ and $\etabar(t)$ is the average intensity when we apply the piecewise constant approximation $\mubar(t)$ then $\etabar(t) \to \eta(t)$ as $M \to \infty$.
%\begin{proof}
%\mehrdad{We need a convergence theorem to exchange the order of internal and limit}, then:
%\begin{align}
%\eta(t) = \EE[\mu(t)] = \EE[ \lim_{\Delta_M \to 0} \mubar(t) ] = \lim_{\Delta_M \to 0} \EE[\mubar(t)] = 
%\lim_{\Delta_M \to 0}  \etabar(t)
%\end{align}
%\end{proof}
%\end{lemma}

Using the above lemma, for the first time, we derive the average intensity for a general exogenous intensity. Section~\ref{appen-temporal} includes a few experiments to investigate these results empirically.
\begin{theorem}
\label{theo:average_general}
If $\Psi\in C^{1}([0,T])$ and satisfies \eqref{eqn:Psi_gen}, and exogenous intensity $\mu$ 
is bounded and piecewise absolutely continuous
on $[0,T]$ where $\mu(t+)=\mu(t)$ at all discontinuous points $t$, then $\mu$ is differentiable almost everywhere,
and the semi-indefinite integral
\begin{equation}\label{eqn:WienerHopf_gen}
\eta(t) = \mu(t) + \int_{0}^{t} \Phi(t-s)\eta(s)ds,\quad \forall t\in[0,T],
\end{equation}
yields a rate function $\eta:[0,T]\to \mathbb{R}_{+}^n$ given by
\begin{align}\label{eqn:eta_gen}
\eta(t) = \int_{0}^{t}\Psi(t-s) d\mu(s).
\end{align}
\end{theorem}

\begin{corollary}
Suppose $\Psi$ and $\mu$ satisfy the same conditions as in Thm.  \ref{theo:average_general},
and define $\psi=\Psi'$, then the rate function is $\eta(t)=(\psi * \mu)(t)$.
In particular, if $\Phi(t)=Ae^{-\omega t}\one_{\geq0}(t)=[a_{ij}e^{-\omega t}\one_{\geq0}(t)]_{ij}$
then the rate function 
$\eta(t)=\mu(t)+A \int_{0}^{t} e^{(A-wI)(t-s)}\mu(s)ds$.
\end{corollary}

%\begin{corollary}
%The result reported in \cite{farajtabar2014activity} is a special case of Thm.  \ref{theo:average_general}. Let $\mu(t) = \mu \sigma(t)$ where $\sigma(t)$ is the unit step function. If $\delta(t)$ is the Dirac delta function, then,
%\begin{align}
%\eta(t) =  \Psi(t) \star \mu'(t) =  \Psi(t) \star (\mu \eta(t))' =  \Psi(t) \star (\mu \delta(t)) =  (\Psi(t) \star \delta(t)) = \Psi(t) \mu  
%\end{align}
%\end{corollary}

%\begin{corollary}
%Let the distance between points are all equal to $\Delta_M = T/M$ and $\mu_{-1} \triangleq \mu_{0^-}=0$, then
%\begin{equation}
%\eta(t) = \sum_{m=0}^{k} \Psi(t-m \Delta_M) (\mu_{m}-\mu_{m-1})
%\end{equation}
%\end{corollary}

 %%%%%%%%%%%%%%%     SECTION     %%%%%%%%%%%%%%%% 
\section{Multi-stage Closed-loop Control Problem}
Given the analytical relation between exogenous intensity and expected overall intensity
(rate function), 
one can solve a single one-stage campaigning problem to find the optimal 
constant intervention intensity~\cite{farajtabar2014activity}.
Alternatively, the time window can be partitioned into multiple stages and 
one can impose different levels of interventions in these stages. 
This yields an open-loop optimization of the cost function where
one selects all the intervention actions at initial time 0.
More effectively, we tackle the campaigning problem in a 
dynamic and adaptive manner where we can postpone deciding 
the intervention by observing the process until the next stage begins. 
This is called the \emph{closed-loop} optimization 
of the objective function.

In this section, we establish the foundation to formulate the problem as 
a multi-stage closed-loop optimal control problem.
We assume that $n$ users are generating events according to multi-dimensional 
Hawkes process with exogenous intensity $\mu(t) \in \RR^n$ and 
impact function $\Phi(t,s) \in \RR^{n \times n}$.

\subsection{Event exposure} Event exposure is the quantity of major interests in campaigning.
The exposure process is mathematically represented as a counting process, 
$\Ecal(t) = (\Ecal^1(t), \ldots, \Ecal^n(t))^{\top}$:
Here,  $\Ecal^i(t)$ records the number of times user $i$ is
exposed (she or one of her neighbors performs an activity) 
to the campaign by time $t$.
Let $B$ be the adjacency matrix of the user network, 
\ie, $b_{ij}=1$ if user $i$ follows user $j$ or equivalently user $j$ influences user $i$. 
We assume $b_{ii}=1$ for all $i$. 
Then the exposure process is  given by
$
\Ecal(t) = B \, \Ncal(t).
$

\subsection{Stages and interventions}
Let $[0,T]$ be the time horizon and 
$0=\tau_0 < \tau_1 < \ldots < \tau_{M-1} < \tau_M = T$ be a partition
into the $M$ stages.
In order to steer the activities of network towards a desired level
(criteria given below) at these stages,
we impose a constant intervention $u_m \in \RR^n$ 
to the existing exogenous intensity $\mu$ 
during time $[\tau_{m},\tau_{m+1})$ for each stage $m=0,1,\dots,M-1$.
The activity intensity at the $m$-th stage is
$
\label{eq:hawkes-matrix-with-control}
\lambda_m(t) = \mu + u_m +  \int_{0}^t \Phi(t,s) \, d\Ncal(s)
$
for $\tau_m \leq t < \tau_{m+1}$
where $\Ncal(t)$ tracks the counting process of activities
since $t=0$.
Note that the intervention itself exhibits a stochastic nature:
%For example, consider the case where campaign members are asked to share a status update about the candidate. They may not comply with this and even if they do so it's uncertain when they make that post. 
%Interestingly, Our intervention type captures this uncertainty in a very simple and effective form. 
adding $u_m^i$ to $\mu^i$ is equivalent to incentivizing
user $i$ to increase her activity rate but it is still uncertain
when she will perform an activity, which appropriately mimics
the randomness in real-world campaigning. 
 
\subsection{States and state evolution}
Note that the Hawkes process is non-Markov and one needs complete
knowledge of the history to characterize the entire process. 
However, the conditional intensity $\lambda(t)$ only depends
on the state of process at time $t$ when the standard 
exponential kernel $\Phi(t,s)=Ae^{-\omega (t-s)}\one_{\geq0}(t-s)$
is employed. In this case, the activity rate at stage $m$ is
\begin{align}
\lambda_m(t) = \mu + u_m +  \underbrace{\int_{0}^{\tau_m} A   e^{-\omega(t-s)} \, d\Ncal(s)}_{\text{from previous stages}}  +
  \underbrace{\int_{\tau_m}^{t} Ae^{-\omega(t-s)} \, d\Ncal(s)}_{\text{current stage}}  
\end{align}
Define $x_m := \lambda_{m-1}(\tau_{m}) - u_{m-1} - \mu$ 
(and $x_0 = 0$ by convention)
then the intensity due to events of all previous $m$ stages can be written as
$
\int_{0}^{\tau_m} A  e^{-\omega(t-s)} \, d\Ncal(s) = x_m e^{-\omega(t-\tau_m)}
$.
In other words, $x_m$ is sufficient to encode the information of activity 
in the past $m$ stages that is relevant to future. 
This is in sharp contrast to the general case where the state space grows 
with the number of events.
%%
%By convention we take $x_0 = 0$.
%%
%Given  $x_m$ and $u_m$ the counting exposure process is characterized using the counting activity process. We write $x_{m+1} = f(x_m, u_m, r_m)$ where $r_m$ captures the randomness in the process during the stage $m$. 

\subsection{Objective function}
For a sequence of controls $u(t)=\sum_{m=0}^{M-1} u_m\one_{[\tau_m,\tau_{m+1})}(t)$,
the activity counting process $\Ncal(t)$ is generated by intensity
$\lambda(t) = \mu + u(t) +  \int_{0}^t A e^{-\omega (t-s)} \, d\Ncal(s)$.
For each stage $m$ from $0$ to $M-1$, $x_m$ encodes the effects from
previous $m$ stages as above and $u_m$ is the current control
imposed at this stage. Let $\Ecal_m^i(t;x_m,u_m):=B\int_{\tau_m}^td\Ncal^i(s)$ 
be the number of times user $i$
is exposed to the campaign by time $t\in[\tau_m,\tau_{m+1})$ in stage $m$,
then the goal is to steer the expected total number of exposure
$\bar{\Ecal}_m^i(x_m,u_m):=\EE[\Ecal_m^i(\tau_{m+1};x_m,u_m)] $ to a desired level.
In what follows, we introduce several instances of the objective function $g(x_m,u_m)$
in terms of $\{\bar{\Ecal}_m^i(x_m,u_m)\}_{i=1}^{n}$ in each stage $m$
that characterize different \emph{exposure shaping} tasks. Then the overall control problem
is to find $u(t)$ that optimizes 
the total objective $\sum_{m=0}^{M-1} g_m(x_m, u_m)$.
%where $g_m(x_m, u_m, r_m) $ is the cost function associated to stage $m$. 
%Let $\Ecal_m(t)$ be the counting exposure process for the $m$-th interval, \ie, $\Ecal(t) = \sum_{m=0}^{\floor{t/\Delta_M}} \Ecal_m(t)$ and  $\Ecal_m^i(t)$ counts the number of times user $i$ is exposed to campaign before time $t$ in the $m$-th time interval. In other words,  $\Ecal_m^i(\tau_m) = 0$ and  $\Ecal_m^i(\tau_{m+1}) $ is the total number of times $i$ is exposed to the campaign at the $m$-th stage. We introduce the following \emph{exposure shaping}  tasks: 
\begin {itemize}[leftmargin=*]
\item
\emph{Capped Exposure Maximization (CEM)}:
In real networks, there is a cap on the exposure each user can tolerate due to the limited attention of a user. Suppose we know the upper bound $\beta_m^i$ , on user $i$'s exposure tolerance over which the extra exposure is not counted towards the objective. Then, we can form the following \emph{capped exposure maximization}
\begin{align}
g_m(x_m, u_m) = \frac{1}{n} \sum_{i=1}^n \min \cbr{ \bar{\Ecal}_m^i( x_m, u_m ), \beta_m^i}
\end{align}
\item 
\emph{Minimum Exposure Maximization (MEM)}:
Suppose our goal is instead to maintain the exposure of campaign on each user above a certain minimum level, at each stage or, alternatively to make the user with the minimum exposure as exposed as possible, we can consider the following cost function:
\begin{align}
g_m(x_m, u_m) =  \min_i \bar{\Ecal}_m^i(x_m, u_m)
\end{align}
\item
\emph{Least-squares Exposure Shaping (LES)}:
Sometimes we want to achieve a pre-specified target exposure levels, $\gamma_m \in \RR^n$, for the users.
For example, we may like to divide users into groups and desire a different level of exposure in each group. To this end, we can perform 
least-squares campaigning task with the following cost function where $D$ encodes potentially additional constraints (e.g., group partitions):
\begin{align}
g_m(x_m, u_m) =  -\frac{1}{n} \| D \bar{\Ecal}_m(x_m, u_m) - \gamma_m \|^2
\end{align}
\end{itemize}

\subsection{Policy and actions}
By observing the counting process in previous stages (summarized in
a sequence of $x_m$) and
taking the future uncertainty into account, the control problem
is to design a policy $\pi = \{\pi_m:\RR^n \to \RR^n:m=0,\dots,M-1\}$ 
such that the controls $u_m=\pi_m(x_m)$ can maximize the total objective
$\sum_{m=0}^{M-1} g_m(x_m, u_m)$.
%$
%x_{m+1} = f(x_m, \pi_m(x_m), r_m).
%$
In addition, we may have constraints on the amount of control.
For example, a budget constraint on the sum of all interventions to 
users at each stage, or, a cap over the amount of intensity a user 
can handle.  
A feasible set or an action space over which we find the best intervention 
is represented as
$
\Ucal_m := \cbr{u_m \in \RR^n | c_m^{\top} u_m \le C_m, 0 \leqslant u_m \leqslant \alpha_m}
$.
Here, $c_m \in \RR_+^n$ contains the price of each person per unit increase of exogenous intensity and $C_m \in \RR_+ $ is the total budget at stage $m$. Also, $\alpha_m \in \RR_+^n$ is the cap on the amount of activities of the users. 

To summarize, the following problem is formulated to find the optimal control policy $\pi$:
\begin{equation}
	\label{eq:optimization-main}
	\begin{array}{l}
		\maximize_{\pi }  \sum_{m=0}^{M-1} g_m(x_m, \pi_m(x_m)), \ \, \, \\
		\mbox{subject to }    \pi_m(x_m) \in \Ucal_m, \mbox{ for } m=0,\dots,M-1.
	\end{array}
\end{equation}
%%%%%%%%%%%%%%%     SECTION     %%%%%%%%%%%%%%%% 
\section {Closed-loop Dynamic Programming Solution}

We have formulated the control problem as an optimization in \eqref{eq:optimization-main}.
However, when control policy $\pi_m$ is to be implemented, only
$x_m$ is observed and there are still uncertainties in future 
$\{x_{m+1},\dots,x_{M-1}\}$. 
For instance, when $\pi_m$ is implemented according to $x_m$ starting
from time $\tau_m$, the
intensity $x_{m+1}:=f(x_m,\pi_m(x_m))$ at time $\tau_{m+1}$ 
depends on $x_m$ and the control $\pi_m(x_m)$,
but is also random due to the stochasticity of 
the process during time $[\tau_m,\tau_{m+1})$.
Therefore, the design of $\pi$ needs
to take future uncertainties into considerations.

Suppose we have arrived at stage $M$ at time $\tau_{M-1}$
with observation $x_{M-1}$, then the optimal policy $\pi_{M-1}$
satisfies $g_{M-1}(x_{M-1},\pi_{M-1}(x_{M-1}))
=\max_{u\in\Ucal_{M-1}}g_{M-1}(x_{M-1},u)=:J_{M-1}(x_{M-1})$.
We then repeat this procedure for $m$ from $M-1$ to $0$ backward 
to find the sequence of controls via dynamic programming such that
the control $\pi_m(x_m)\in \Ucal_m$ yields optimal objective value
\begin{align}
\label{eq:dynamic-programming}
J_{m}(x_m) = \max_{u_m\in \Ucal_m}  \,  \EE [g_m(x_m, u_m)  +  J_{m+1}(f(x_m, u_m))]
\end{align}

%We therefore formulate the problem as an optimization of the expected cost:
%\begin{align}
%J_\pi(x_0) = \EE \left[ \sum_{m=0}^{M-1} g_m(x_m, \pi(x_m), r_m) \right],
%\end{align}
%such that $\pi(x_m) \in \Ucal_m$ and the expectation is given with respect to the joint distribution of all the random variables involved.
%%
%It's easy to show that this optimization problem exhibits the \emph{principle of optimality}~\cite{bertsekas1995dynamic} and a dynamic programming solution can be provided to find the optimal policy.
%%
%For every initial state $x_0$, the optimal cost $J(x_0)$ of the above problem is equal to $J_0(x_0)$ which is given by the last step of the following algorithm which proceeds backward in $m$ from $M-1$ to $0$:
%such that $u_m \in U_m$.
%The expectation is taken with respect to the probability distribution of $r_m$, which depends on $x_m$ and $u_m$.

\begin{algorithm}[t]
\caption{Closed-loop Multi-stage Dynamic Programming}
\label{alg:dynamic} %  \cite{Ogata81}} \label{alg:sampling}
\begin{algorithmic}
\STATE \textbf{Input:} Intervention constraints: $c_0 \ldots c_{M-1}$, $C_0 \ldots C_{M-1}$, $\alpha_0\ldots \alpha_{M-1}$, 
\STATE \textbf{Input:} Objective-specific constraints: $\beta_0 \ldots \beta_{M-1}$ for CEM and
$\gamma_0 \ldots \gamma_{M-1}$ for LES
\STATE \textbf{Input:}  Time: $T$,  Hawkes parameters: $A$, $\omega$ 
\STATE \textbf{Output:} Optimal intervention $u_0 \dots u_{M-1}$, Optimal cost: $Cost$
\STATE Set $x_0 \gets 0$
\STATE Set $Cost \gets 0$ 
\FOR {$l \gets 0: M-1$ }
\STATE $(v_l \ldots v_{M-1})$ = $open\_loop(x_l)$
(Problems \eqref{eq:optimization-CEM}, \eqref{eq:optimization-MEM}, \eqref{eq:optimization-LES} for CEM, MEM, LES respectively)
\STATE  Set $u_l \gets v_l$  
\STATE Drop $v_{l+1} \ldots v_{M-1}$
\STATE Update next state $x_{l+1} \gets f_l(x_l, u_l)$
 \STATE  Update $Cost  \gets Cost + g_l(x_l, u_l)$
\ENDFOR
\end{algorithmic}
\end{algorithm}

\subsection{Approximate dynamic programming} Solving $\eqref{eq:dynamic-programming}$ for finding $J_m(x_m)$ analytically is intractable.
Therefore, we will adopt an approximate dynamic programming scheme.
In fact approximate control is as essential part of dynamic programming as the optimization is usually intractable due to curse of dimensionality except a few especial cases~\cite{bertsekas1995dynamic}.
Here we adopt a suboptimal control scheme, \emph{certainty equivalent control} (CEC), which applies at each stage the control that would be optimal if the uncertain quantities were fixed at some typical values like the average behavior. It results in an optimal control sequence, the first component of which is used at the current stage, while the remaining components are discarded. The procedure is repeated for the remaining stages. Algorithm~\ref{alg:dynamic} summarizes the dynamic programing steps. 
This algorithm has two parts: (i) certainty equivalence which the random behavior is replaced by its average; and (ii) the open-loop optimization.
Let's assume we are at the beginning of stage $l$ of the Alg.~\ref{alg:dynamic} with state vector $x_l$ at $\tau_l$.

\subsection{Certainty equivalence} 
We use the machinery developed in Sec.\  \ref{sec:mean} to compute the average of exposure at any stage $m=l, l+1, \ldots, M-1$.
\begin{align}
\label{eq:open-loop-exposure}
\bar{\Ecal}_m(x_{m},u_{m})  &
= B  \EE [\Ncal(\tau_{m+1})-\Ncal(\tau_{m}) ] 
= B \EE \left[ \int_{\tau_{m}}^{\tau_{m+1}}  d \Ncal(s) \right] 
= B \int_{\tau_{m}}^{\tau_{m+1}} \eta_m(s) \, ds
\end{align}
where $\eta_m(t) = \EE[ \lambda_m(t)]$ and
$\lambda_m(t) = \mu + u_m + x_{l} e^{-\omega (t-\tau_l)} + \int_{\tau_l}^{t} A e^{-\omega (t-s)} d \Ncal(s)$ for $t\in[\tau_m , \tau_{m+1})$.
Now,  we use the superposition property of point processes~\cite{DalVer2007} to decompose the  process as $\Ncal(t) = \Ncal^c(t) + \Ncal^v(t)$ 
corresponding to $ \lambda_m(t) = \lambda_m^c(t) + \lambda_m^v(t)$ 
where the first 
$ \lambda_m^c(t) = \mu + u_m + \int_{\tau_l}^{t} A  e^{-\omega (t-s)} d \Ncal^c(s)$
consists of events caused by exogenous intensity at current stage $m$
and the second
$\lambda_m^v(t) = x_{l} e^{-\omega (t-\tau_l)} + \int_{\tau_l}^{t} A  e^{-\omega (t-s)} d \Ncal^v(s)$
is due to activities in previous stages.
%\begin{align}
%&  \lambda_m^c(t) = \mu + u_m + \int_{\tau_l}^{t} A  e^{-\omega (t-s)} d \Ncal^c(s) \\  
%&  \lambda_m^v(t) = x_{l} e^{-\omega (t-\tau_l)} + \int_{\tau_l}^{t} A  e^{-\omega (t-s)} d \Ncal^v(s)
%\end{align}
According to Thm.   \ref{theo:piecewise_constant_average} we have
\begin{align}
\eta_m^c(t):=\EE[\lambda_m^c(t)] = \Psi(t-\tau_l)   \mu +  
\Psi(t-\tau_l) u_l + \sum_{k=l+1}^{m-1}  \Psi(t-\tau_{k}) (u_{k}-u_{k-1}),
\end{align}
and according to Thm.  \ref{theo:average_general} we have
\begin{align}
\eta_m^v(t):=\EE[\lambda_m^v(t)] = \int_{\tau_l}^t \Psi(t-s) 
\, d(x_l e^{-\omega(s-\tau_l)} \one_{[\tau_l,\infty)}(s)).
\end{align}
From now on, for simplicity, we assume stages are based on equal partition of $[0,T]$ to $M$ segments where each has length $\Delta_M$. Combining Eq.~\eqref{eq:open-loop-exposure} and
$ \eta_m(t) = \eta_m^c(t)+\eta_m^v(t) $ yields:
\begin{equation}
\label{eq:epose-linear-first}
\begin{split}
\bar{\Ecal}_m(x_{m},u_{m})   =&  \Gamma((m-l+1) \Delta_M) u_l + 
 \Gamma((m-l) \Delta_M) (u_{l+1}-u_l)   + \ldots  \\ &
+ \Gamma(\Delta_M) (u_{m}-u_{m-1}) + \Gamma((m-l+1) \Delta_M) \mu +
 \Upsilon((m-l+1) \Delta_M) x_l
 \end{split}
\end{equation}
where $\Gamma(t)$ and $\Upsilon(t)$ are matrices independent of $u_m$'s and are defined as:
\begin{align}
& \Upsilon(t) = B \int_0^t e^{ (A-\omega I) s}\, ds   = B (A-\omega I)^{-1} (e^{ (A-\omega I) t}- I) \\
& \Gamma(t) =  B \int_0^t \Psi(s) \, ds = B \Upsilon(t) + B(A-\omega I)^{-1} (\Upsilon(t)- It)/\omega;
\end{align}
Note the linear relation between average exposure $\bar{\Ecal}_m(x_{m},u_{m}) $ and intervention values $u_l, \ldots, u_{m-1}$.% which will be employed in the next part.
Let $\Gamma(k \Delta_M) = \Gamma_k$ and $\Upsilon(k \Delta_M) = \Upsilon_k$. Then  for every $m \geq l$, Eq.~\eqref{eq:epose-linear-first} is rewritten as:
\begin{align}
\bar{\Ecal}_m(x_{m},u_{m})  = \sum_{k=l}^{m-1} (\Gamma_{m-k+1}-\Gamma_{m-k}) u_k +  \Gamma_1 u_m + \Gamma_{m-l+1} \mu 
+ \Upsilon_{m-l+1} x_l.
\end{align}

\subsection{Open-loop optimization}
Having found the average exposure at stages  $m=l,  \ldots, M-1$ we formulate an open-loop optimization to find optimal $u_l, u_{l+1}, \ldots, u_{M-1}$.
Defining $\uhat_l = (u_l; \ldots; u_{M-1})$ and $\hat
{\Ecal}_l = (\bar{\Ecal}_l(x_{l},u_{l}); \ldots; \bar{\Ecal}_{M-1}(x_{M-1},u_{M-1})) $
we can aggregate
Aggregating these linear forms for all $l \geq m$ yields to the following matrix equation for finding $\hat{\Ecal}_l$:
\begin{equation}
\begin{split}
\underbrace{
\begin{bmatrix}
\bar{\Ecal}_l(x_{l},u_{l})   \\
\bar{\Ecal}_{l+1}(x_{l+1},u_{l+1})   \\
\bar{\Ecal}_{l+2}(x_{l+2},u_{l+2})  \\
\vdots \\
\bar{\Ecal}_{M-1}(x_{M-1},u_{M-1}) 
\end{bmatrix}
}_{\hat{\Ecal}_l}
= &
\underbrace{
\begin{bmatrix}
    \Gamma_1 & 0 & 0 & \dots  & 0 \\
    \Gamma_{2} -   \Gamma_{1} & \Gamma_{1}  & 0 & \dots  & 0 \\
    \Gamma_{3} - \Gamma_{2}  & \Gamma_{2} -  \Gamma_{1}  & \Gamma_{1} & \dots  & 0 \\
    \vdots & \vdots & \vdots & \ddots & \vdots \\
    \Gamma_{M-l}-\Gamma_{M-l-1} & \Gamma_{M-l-1}-\Gamma_{M-l-2}   & \ldots & 
    \Gamma_{2}-\Gamma_{1}  & \Gamma_{1}
\end{bmatrix}
}_{X_l}
\underbrace{
\begin{bmatrix}
u_l   \\
u_{l+1} \\
u_{l+2} \\
\vdots \\
u_{M-1}
\end{bmatrix}
}_{\uhat_l}  \\ 
+ &
\underbrace{
\begin{bmatrix}
\Gamma_1 \\
\Gamma_2 \\
\Gamma_3 \\
\vdots \\
\Gamma_{M-l}
\end{bmatrix}
}_{Y_l}
\mu
+ 
\underbrace{
\begin{bmatrix}
\Upsilon_1 \\
\Upsilon_2 \\
\Upsilon_3 \\
\vdots \\
\Upsilon_{M-l}
\end{bmatrix}
}_{W_l}
x_l
\end{split}
\end{equation}

%and $X_l$, $Y_l$, $W_l$,  $Z_l$, and $z_l$ are  independent of $\uhat_l$, $\mu$, and $x_l$ as defined in Appendix~\ref{appen-opt}.

Defining  the expanded form of constraint variables as $\chat_l = (c_l; \ldots ; c_{M-1})$, $\CChat_l = (C_l; \ldots ; C_{M-1})$, and $\alphahat_l = (\alpha_l ; \ldots ; \alpha_{M-1})$  we have
\begin{align}
\underbrace{
\begin{bmatrix}
    c_l^{\top}   & \dots  & 0^{\top}  \\
    \vdots & \ddots & \vdots \\
   0^{\top}  & \dots  & c_{M-1}^{\top}  \\
       I  & \dots  & 0 \\
    \vdots & \ddots & \vdots \\
   0& \dots  & I \\
       -I  & \dots  & 0 \\
    \vdots & \ddots & \vdots \\
   0& \dots  & -I
\end{bmatrix}
}_{Z_l}
\underbrace{
\begin{bmatrix}
u_l   \\
u_{l+1} \\
u_{l+2} \\
\vdots \\
u_{M-1}
\end{bmatrix}
}_{\uhat_l}
\leq
\underbrace{
\begin{bmatrix}
C_l   \\
\vdots \\
C_{M-1} \\
\alpha_l \\
\vdots \\
\alpha_{M-1}  \\
0 \\
\vdots \\
0
\end{bmatrix}
}_{z}.
\end{align}

In summary the following equation states the relation between intervention intensities given the constrains:
\begin{align} 
X_l \uhat_l + Y_l \mu + W_l x_l = \hat{\Ecal}_l  \quad \mbox{where} \quad  Z_l \uhat_l \leq z_l  
\end{align}
Then, we provide the optimization from of the above exposure shaping tasks.

For CEM consider $\betahat_l= (\beta_l; \ldots, \beta_{M-1})$. Then the problem
%\begin{align}
%\sum_{m=l}^{M-1} \frac{1}{n} \sum_{i=1}^n \min \cbr{ \Ecal_m^i( \tau_{m+1} ), \beta_m^i} 
%=  \frac{1}{n} \one^{\top} \min \cbr{X_l \uhat_l + Y_l \mu + W_l x_l , \betahat_l}
%\end{align}
%Therefore, the linear program bellow find the optimum value:
\begin{equation}
	\label{eq:optimization-CEM}
	\begin{array}{l}
		\mbox{maximize}_{\hhat, \uhat_l}   \frac{1}{n} \one^{\top} \hhat   \\
		\mbox{subject to} \, \,   X_l \uhat_l + Y_l \mu + W_l x_l  \geq \hhat, \,\,\,  \betahat_l \geq \hhat, \,\,\, Z_l \uhat_l \leq z_l ,\, \,
	\end{array}
\end{equation}
solves CEM where $\hhat$ is an auxiliary vector of size $n(M-l)$.

For MEM consider the auxiliary $h$ as a vector of size $M-l$ and $\hhat$ a vector of size $n(M-1)$.  $\hhat = (h(1); \ldots; h(1); h(2); \ldots, h(2); \ldots, h(M-l); \ldots; h(M-l))$ where each $h(k)$ is repeated $n$ times.  Then MEM is equivalent to
\begin{equation}
	\label{eq:optimization-MEM}
	\begin{array}{l}
		\mbox{maximize}_{\hhat, \uhat_l}  \one^{\top} \hhat   \\
		\mbox{subject to} \,\, X_l \uhat_l + Y_l \mu + W_l x_l  \geq \hhat, \,\,\,  
		\betahat_l \geq \hhat, \,\,\, Z_l \uhat_l \leq z_l 
	\end{array}
\end{equation}

For LES let $\gammahat_l = (\gamma_l; \ldots; \gamma_{M-1})$ and $\hat{D}_l = diag(D, \ldots, D)$, then
\begin{equation}
	\label{eq:optimization-LES}
	\begin{array}{ll}
		\mbox{minimize}_{\uhat_l}    \frac{1}{n} \|\hat{D}_l ( X_l \uhat_l + Y_l \mu + W_l x_l) - \gammahat_l \|^2 \\
		\mbox{subject to} \,\,  Z_l \uhat_l \leq z_l 
	\end{array}
\end{equation}
All the three tasks involve convex (and linear) objective function with linear constraints which impose a convex feasible set. Therefore, one can use the rich and well-developed literature on convex optimization and linear programming to find the optimum intervention.

%matrices and vectors constructed using $\Gamma_1, \ldots, \Gamma_{M}$  and $\Upsilon_1, \ldots, \Upsilon_{M}$ as building blocks. 

\subsection{Scalable optimization}
All the exposure shaping problems defined above require an efficient evaluation of average intensity $\eta(t)$ at all stages, which entails computing  matrices $X_l$, $Y_l$, $W_l$, and $Z_l$. This leads to work with matrix exponentials and inverse matrices to obtain $\Upsilon_m$, and $\Gamma_m$ for $m=1, \ldots, M-1$.
In small or medium networks, we can rely on well-known numerical methods to compute matrix exponentials and inverse.
However, in large networks, % with sparse graph structure $\Ab$, 
the explicit computation of $X_l$, $Y_l$, $W_l$, and $Z_l$ becomes intractable.
Fortunately, we can exploit the following key property of our convex campaigning framework: the  average intensity itself and the gradient of the objective functions only depends on $X_l$, $Y_l$, $W_l$, and $Z_l$ (and consequently on $\Upsilon_m$, and $\Gamma_m$) through matrix-vector product operations. 
Similar to \cite{farajtabar2014activity} for the computation of the product of matrix exponential with a vector, one can use the iterative algorithm by Al-Mohy et al.~\cite{al2011computing}, which combines a scaling and squaring method  with a truncated Taylor series approximation to the matrix exponential.
For solving the sparse linear system of equation, we use the well-known GMRES method, which is an Arnoldi process for constructing an $l_2 $-orthogonal basis of Krylov subspaces. The method solves the linear system by iteratively minimizing the norm of the residual vector over a Krylov subspace. 
For details please refer to \cite{farajtabar2014activity}.
Last but not least, we don't need to explicitly build $X_l$, $Y_l$, $W_l$, and $Z_l$. 
%For example, $X_l$ is block-wise triangular which helps solving for linear systems of equation. 
%Furthermore, at
At each step of gradient computation all the operations involving them are multiplication of$\Upsilon_1, \ldots, \Upsilon_{M}$, and  $\Gamma_1, \ldots, \Gamma_{M}$ to vectors such as $u_0, \ldots, u_{M-1}$ and $\mu$.% which can be computed once and used many times during that step of optimization.

%To study the effectiveness of the above ideas on fast computation of $X_l \uhat_l$ - which is the core time-consuming part of the gradient computation we empirically study the time of this operation on increasing the size of network from 200 to 1000 on 6-stage optimal control problem. The naive method needs to first compute a matrix of $6000 \times 6000$ non-sparse lower triangular matrix via explicit matrix exponential computations. 

%%%%%%%%%%%%%%%%%%%%%%%%%%%%%%%%%%%%%%%%%%%%%%%%%%%%
\begin{figure*}[!t]
  % \vspace{-3mm}
  \centering
  \setlength{\tabcolsep}{6pt}
  \begin{tabular}{ccc}
          \hspace{-4mm}
          \includegraphics[width=0.33\textwidth]{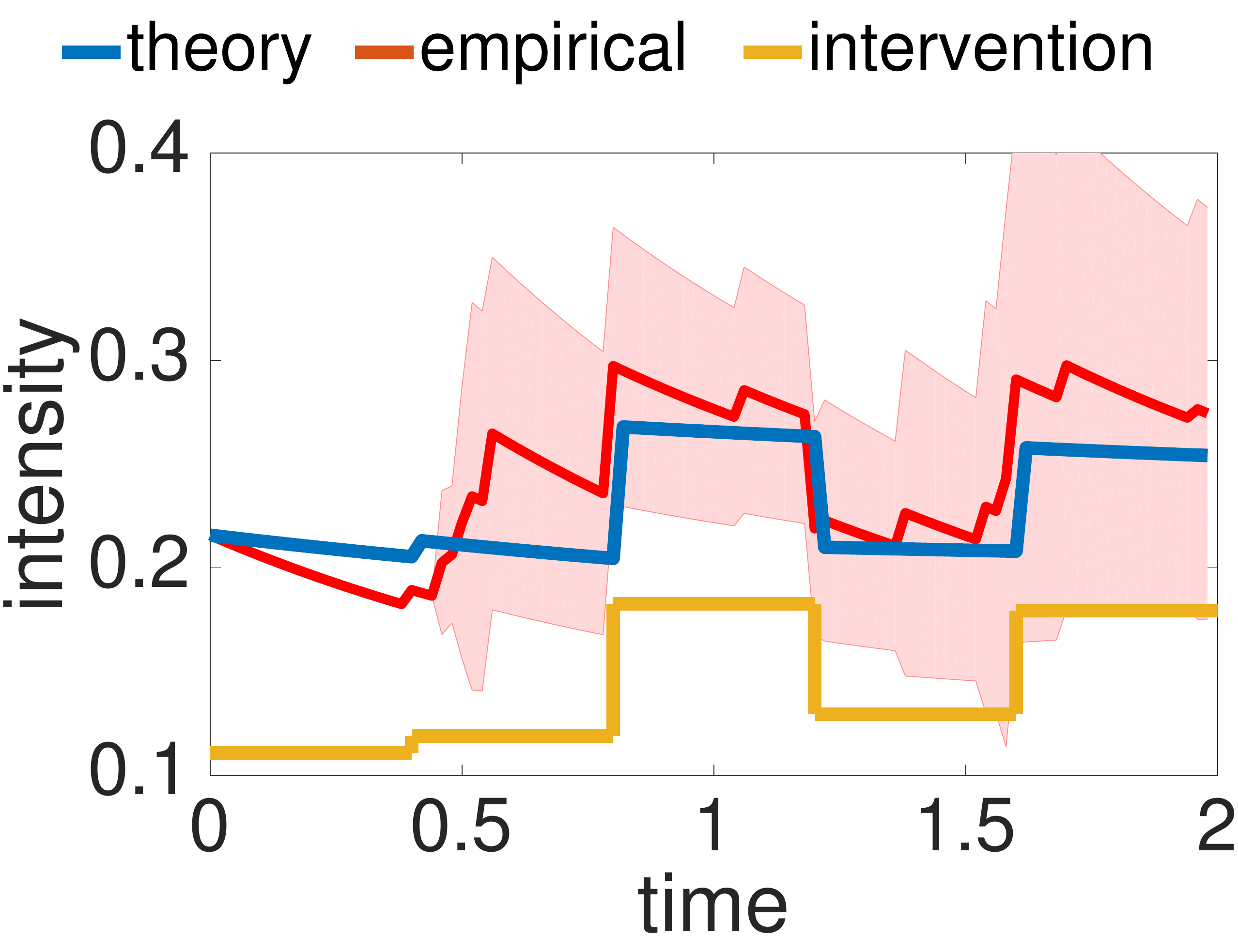} &
          \hspace{-4mm}
          \includegraphics[width=0.33\textwidth]{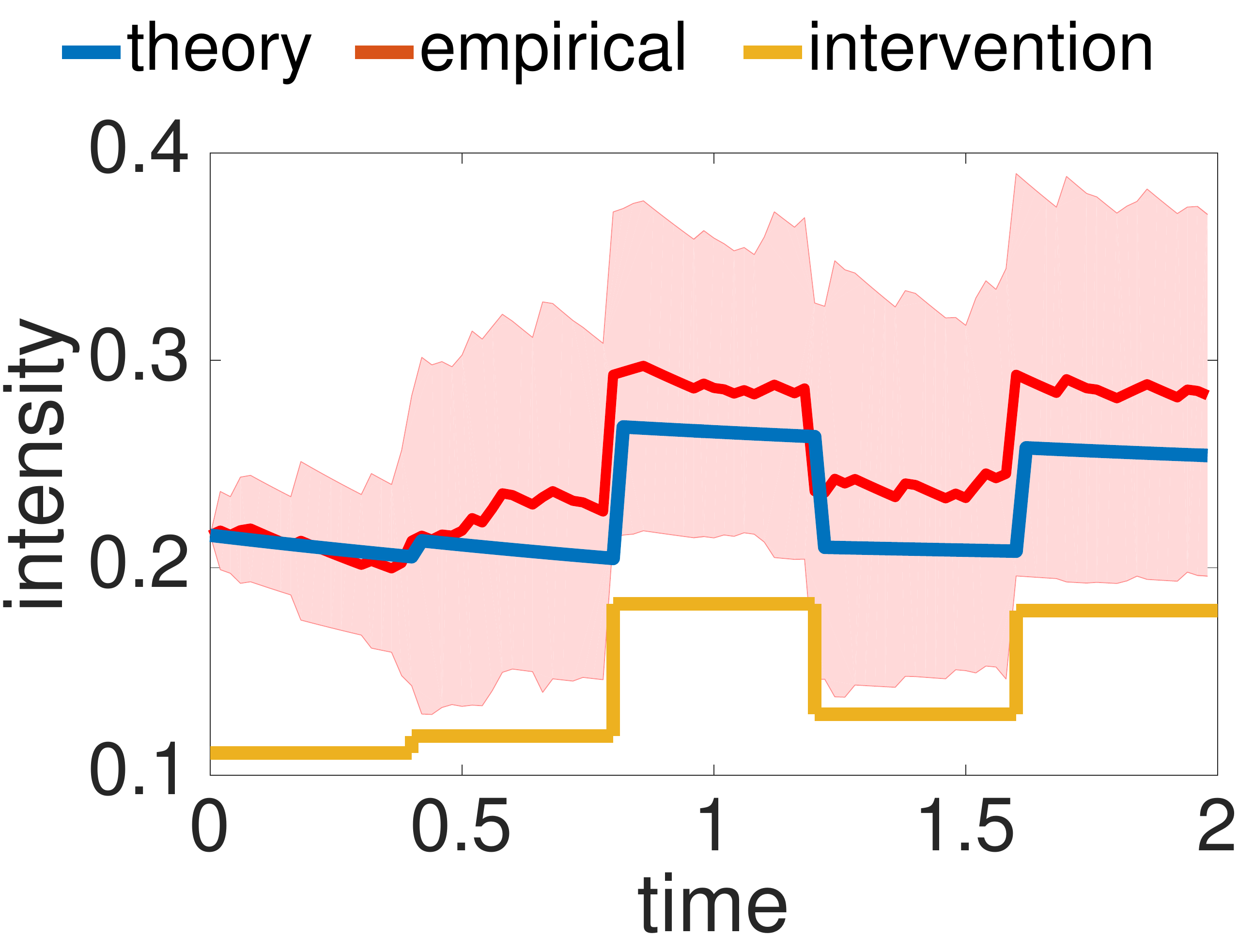} &
          \hspace{-4mm}
          \includegraphics[width=0.33\textwidth]{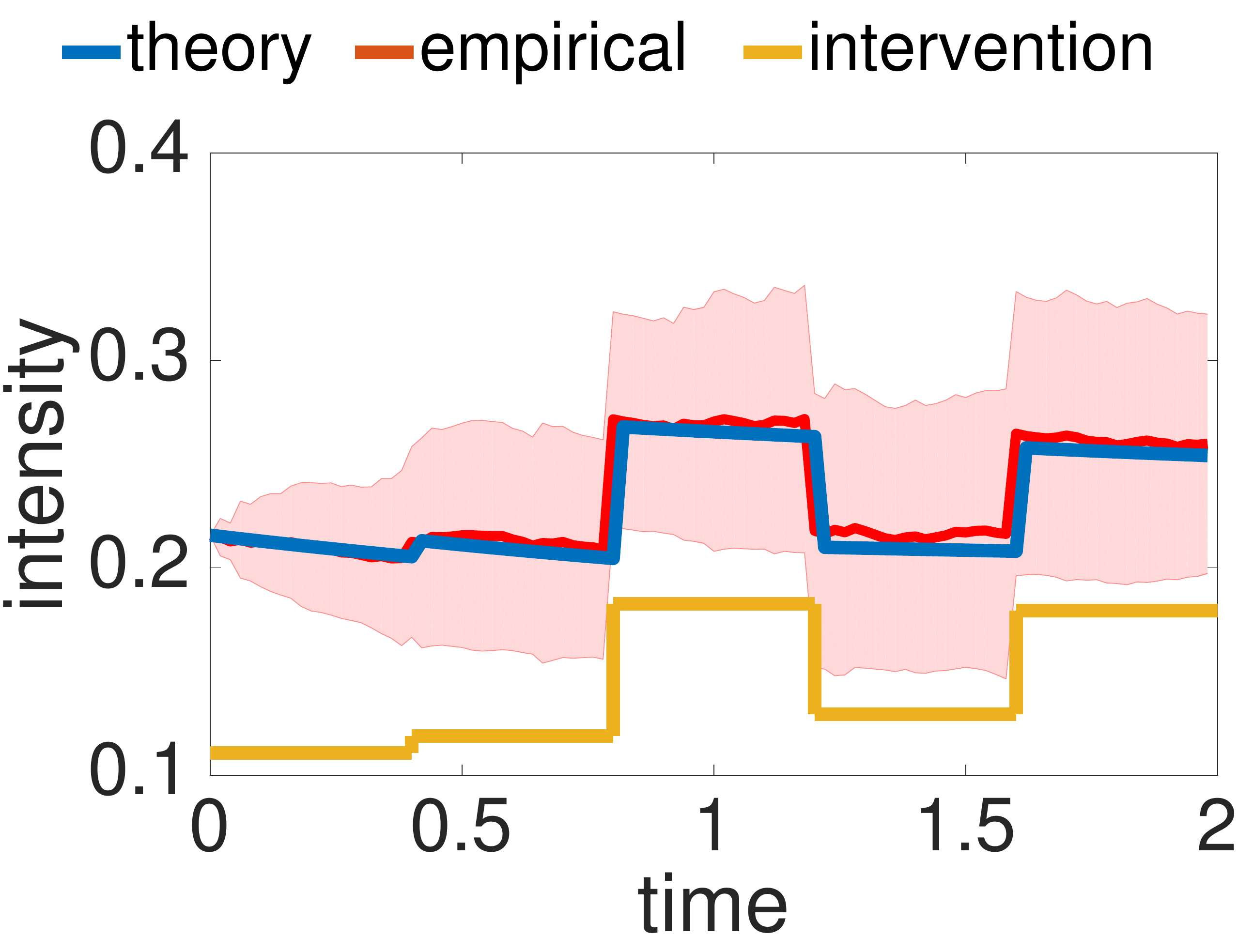} \\
                 \hspace{-4mm}
          \includegraphics[width=0.33\textwidth]{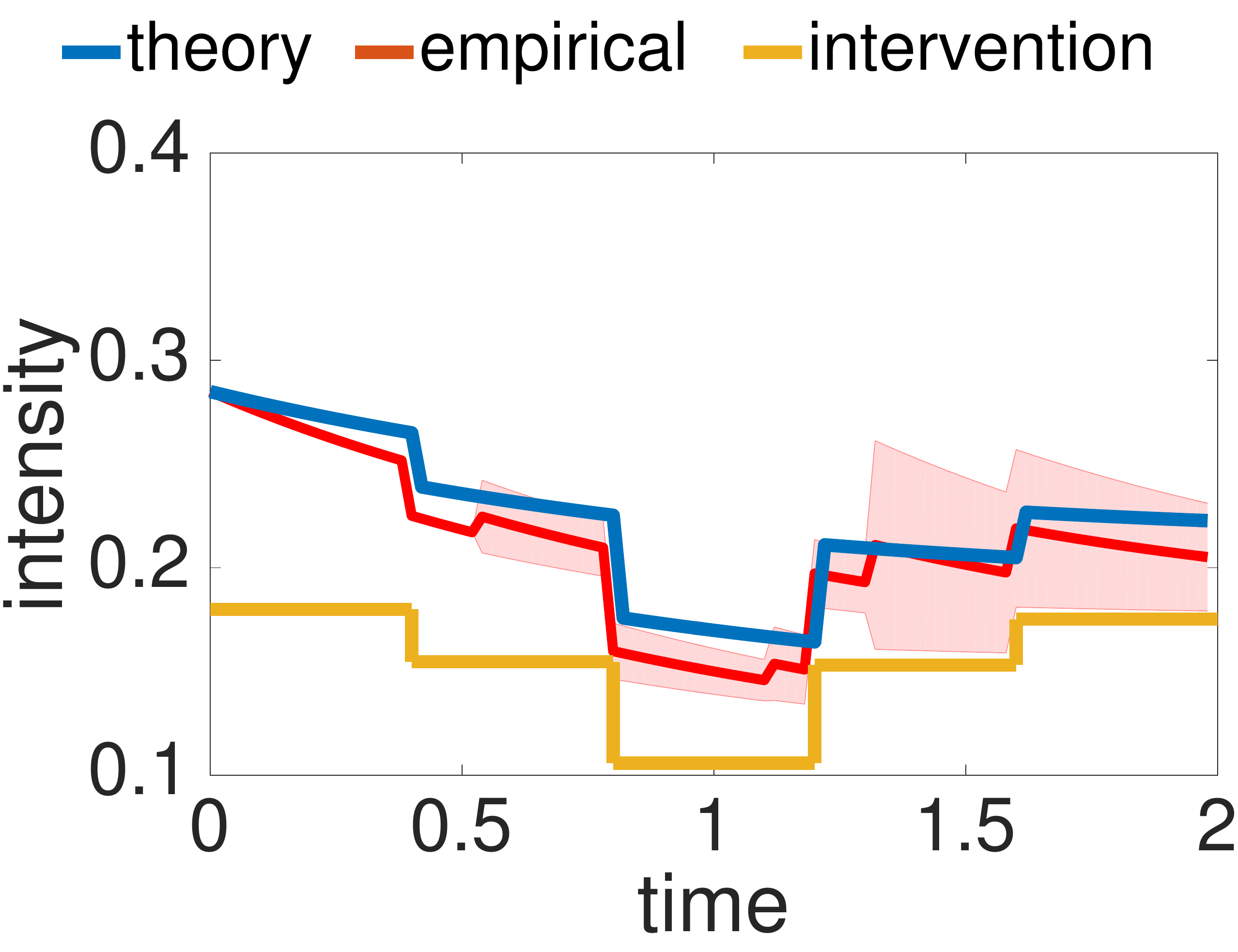} &
          \hspace{-4mm}
          \includegraphics[width=0.33\textwidth]{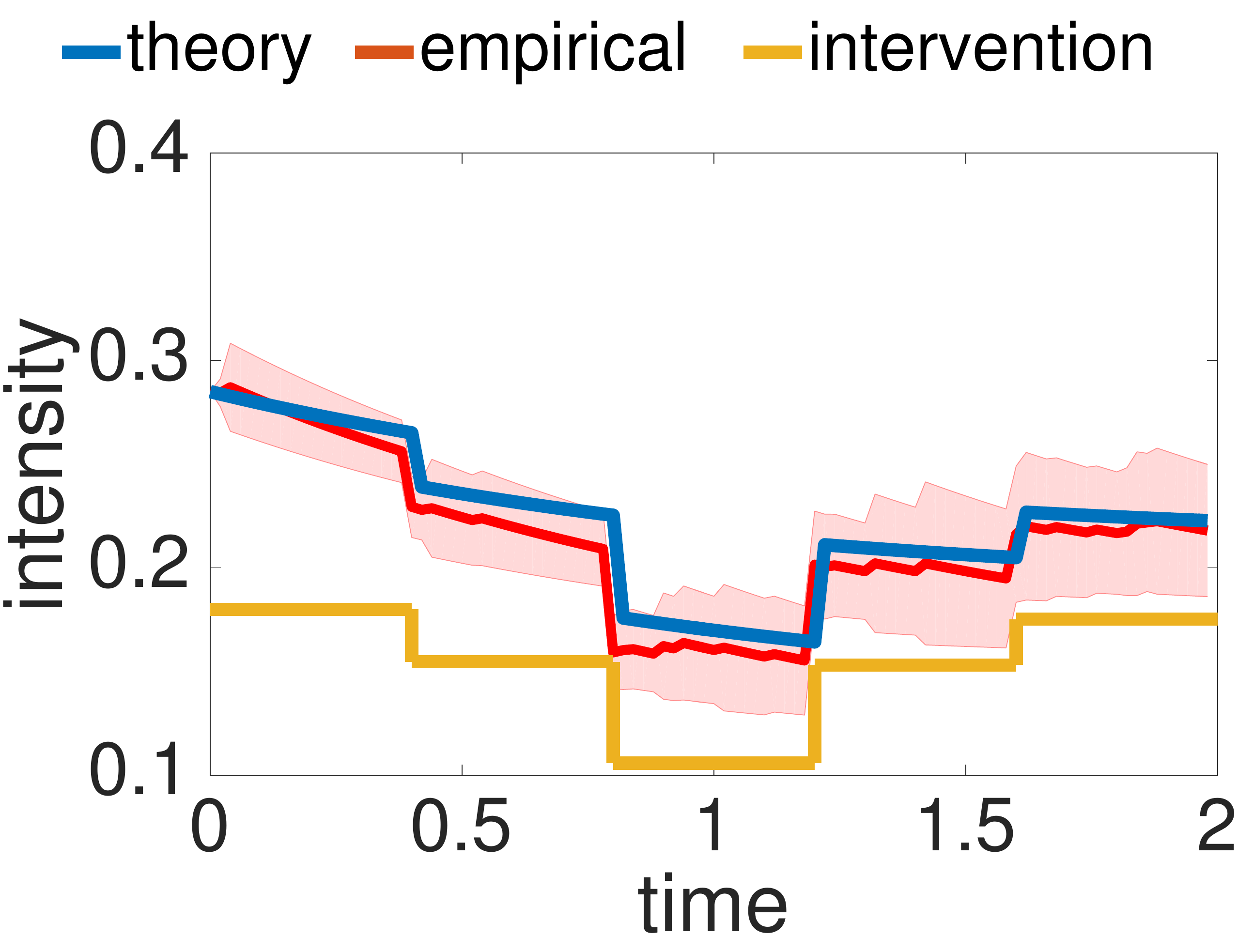} &
          \hspace{-4mm}
          \includegraphics[width=0.33\textwidth]{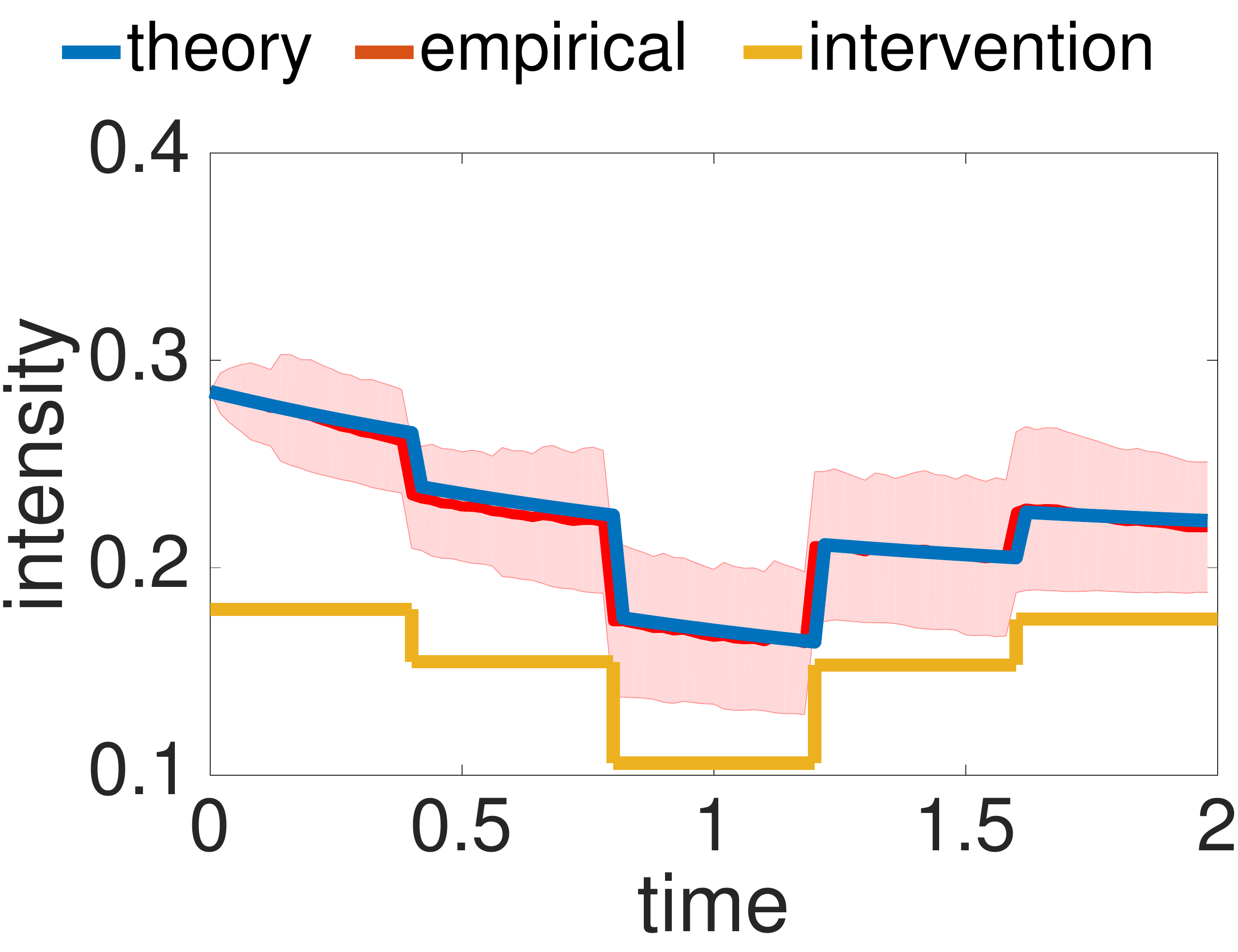} \\
                    \hspace{-4mm}
          \includegraphics[width=0.33\textwidth]{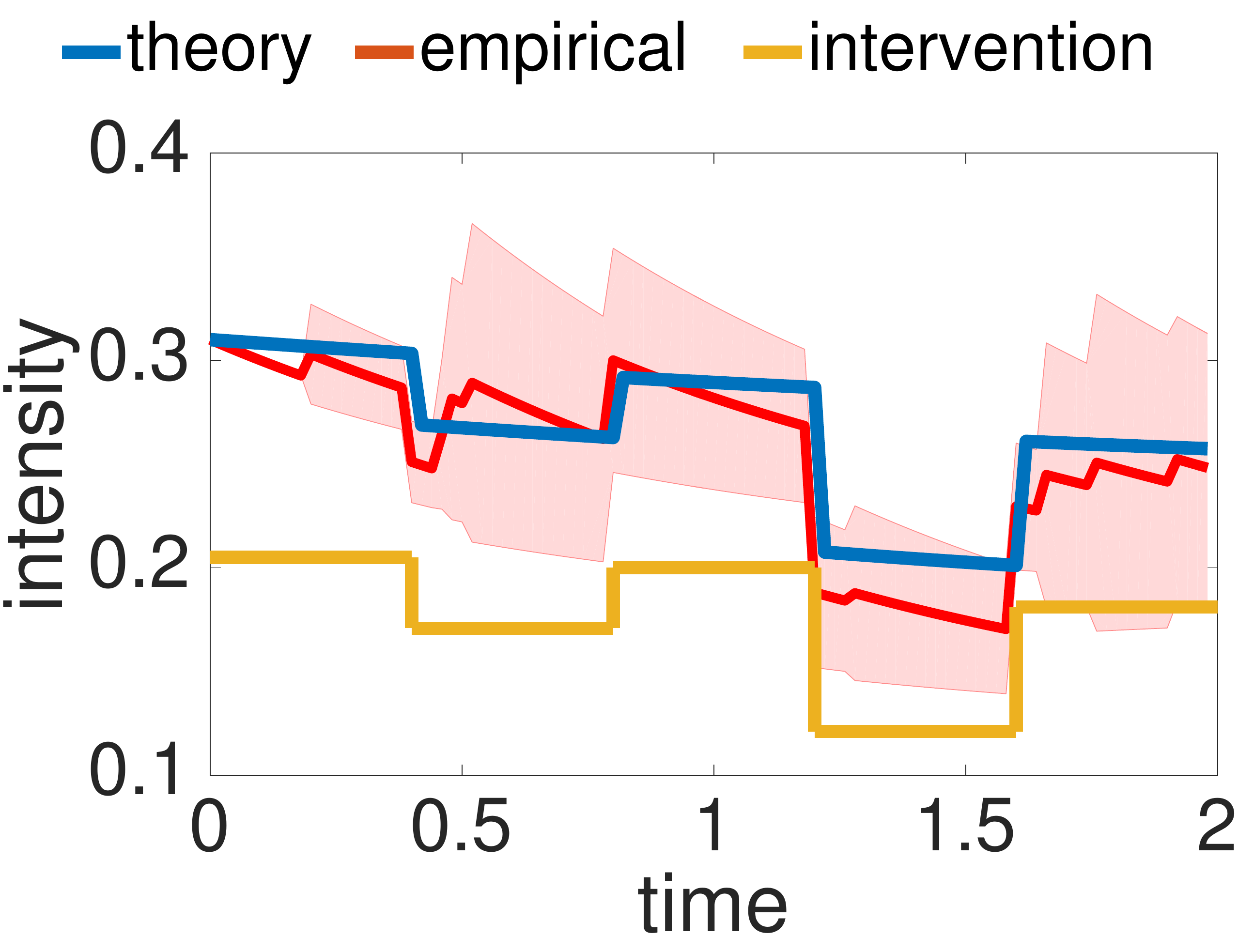} &
          \hspace{-4mm}
          \includegraphics[width=0.33\textwidth]{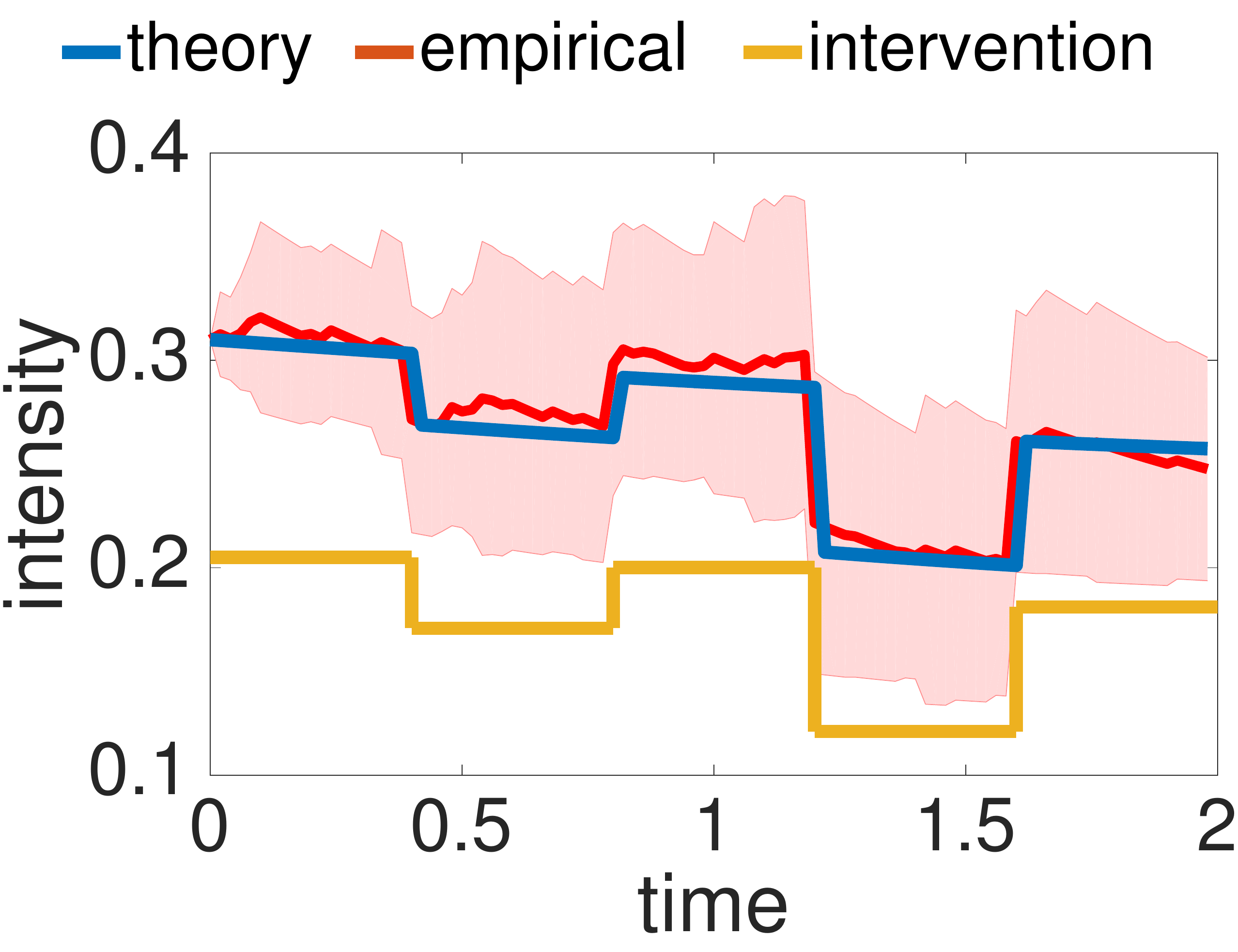} &
          \hspace{-4mm}
          \includegraphics[width=0.33\textwidth]{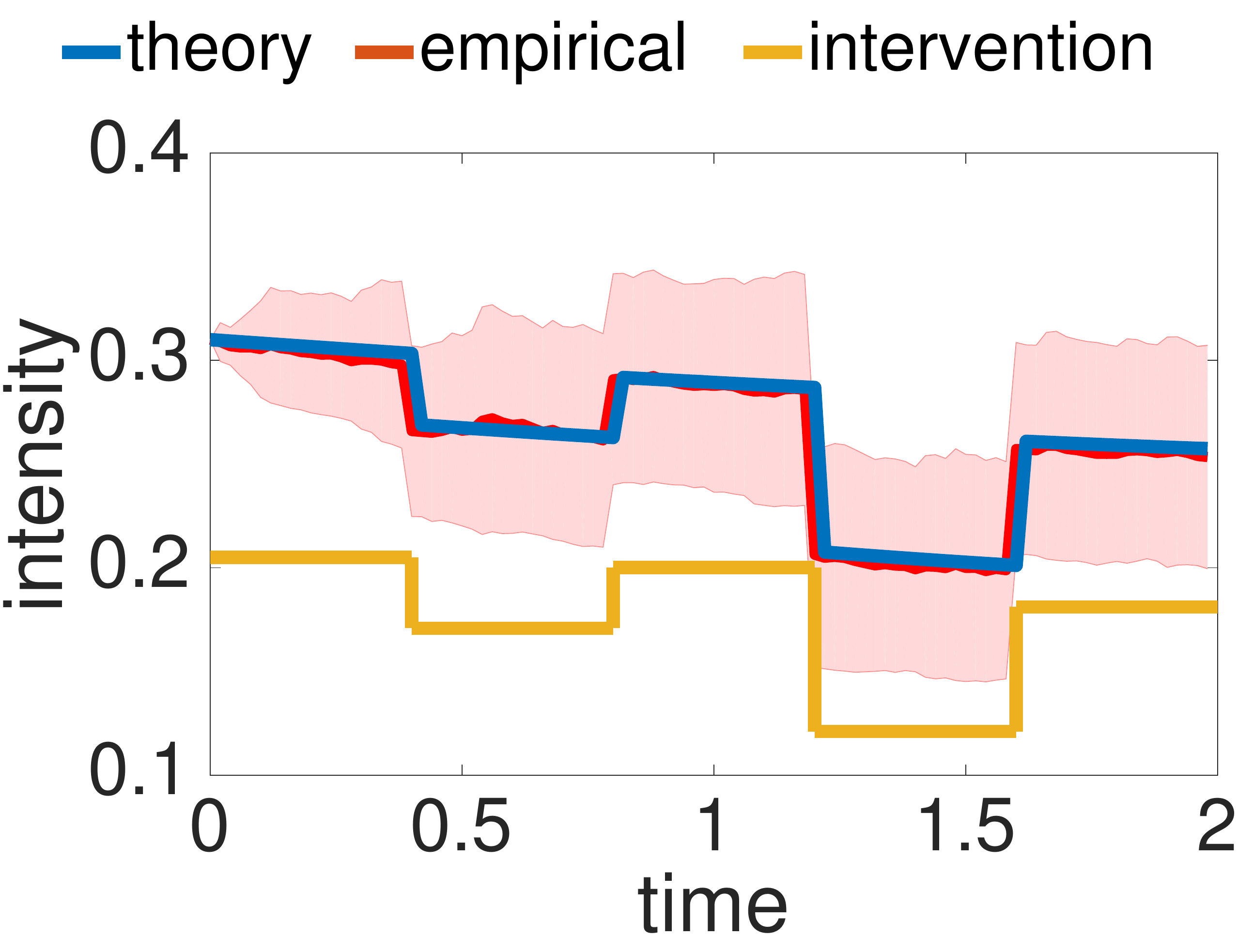}\\
          a) $5$ runs & b) $20$ runs & c) $100$ runs
  \end{tabular}
  \caption{Empirical investigation of theoretical results in Thm. ~\ref{theo:piecewise_constant_average}. blue: theoretical average intensity; red: empirical average intensity and sample standard deviation; orang: piecewise-constant exogenous intensity (interventions)}
  \label{fig:intensity}
\end{figure*}

\begin{figure*}[!t]
  % \vspace{-3mm}
  \centering
  \setlength{\tabcolsep}{6pt}
  \begin{tabular}{ccc}
          \hspace{-4mm}
          \includegraphics[width=0.33\textwidth]{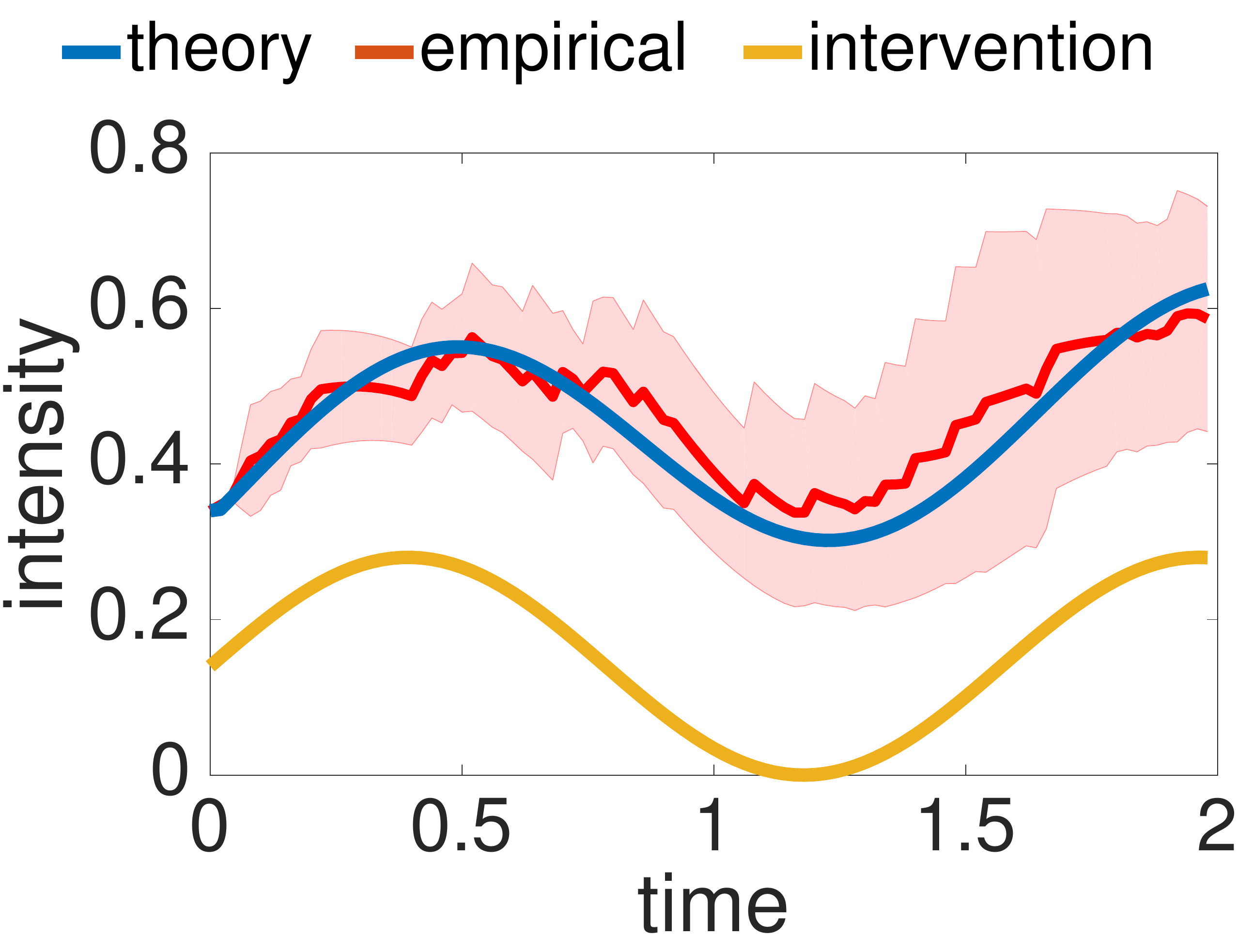} &
          \hspace{-4mm}
          \includegraphics[width=0.33\textwidth]{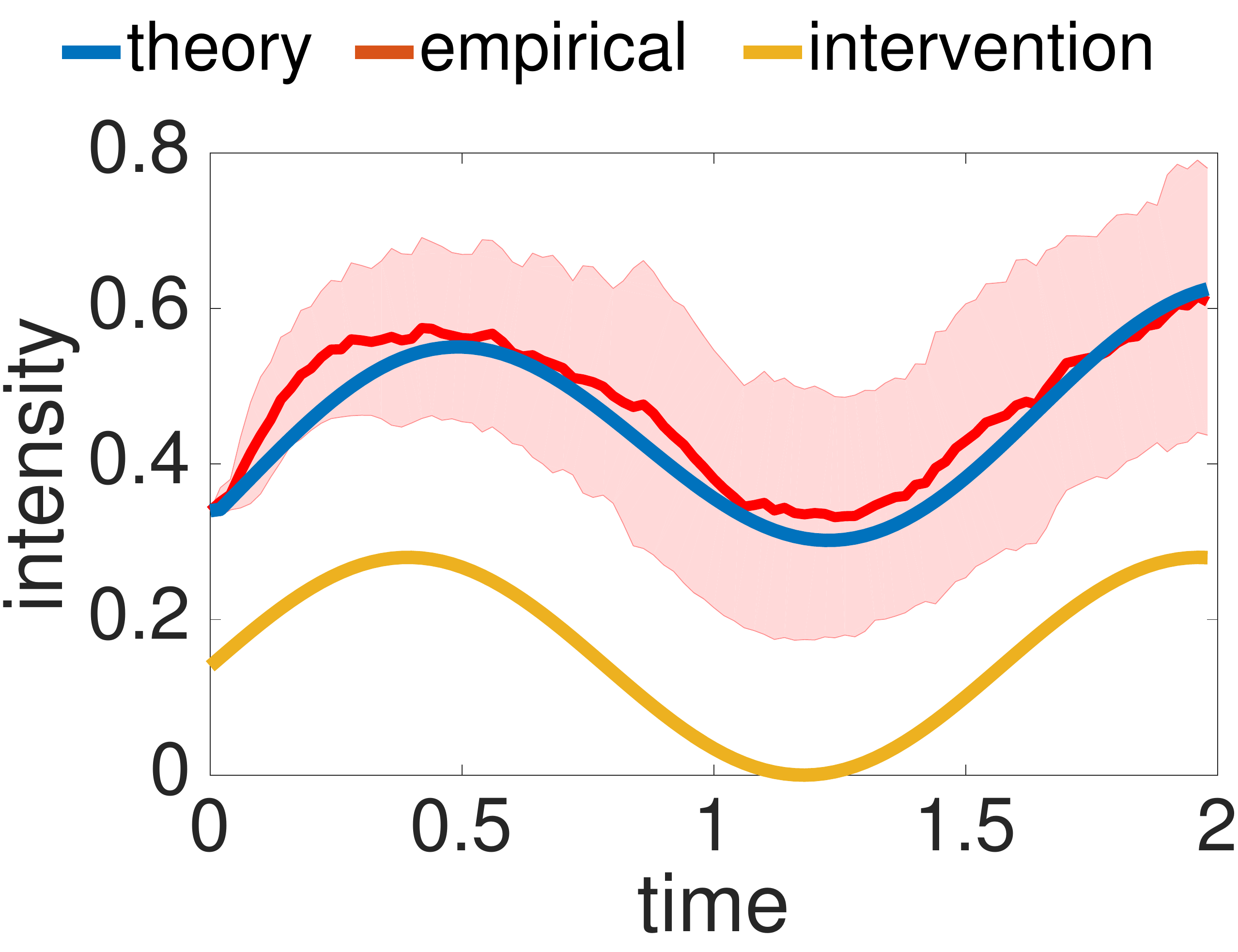} &
          \hspace{-4mm}
          \includegraphics[width=0.33\textwidth]{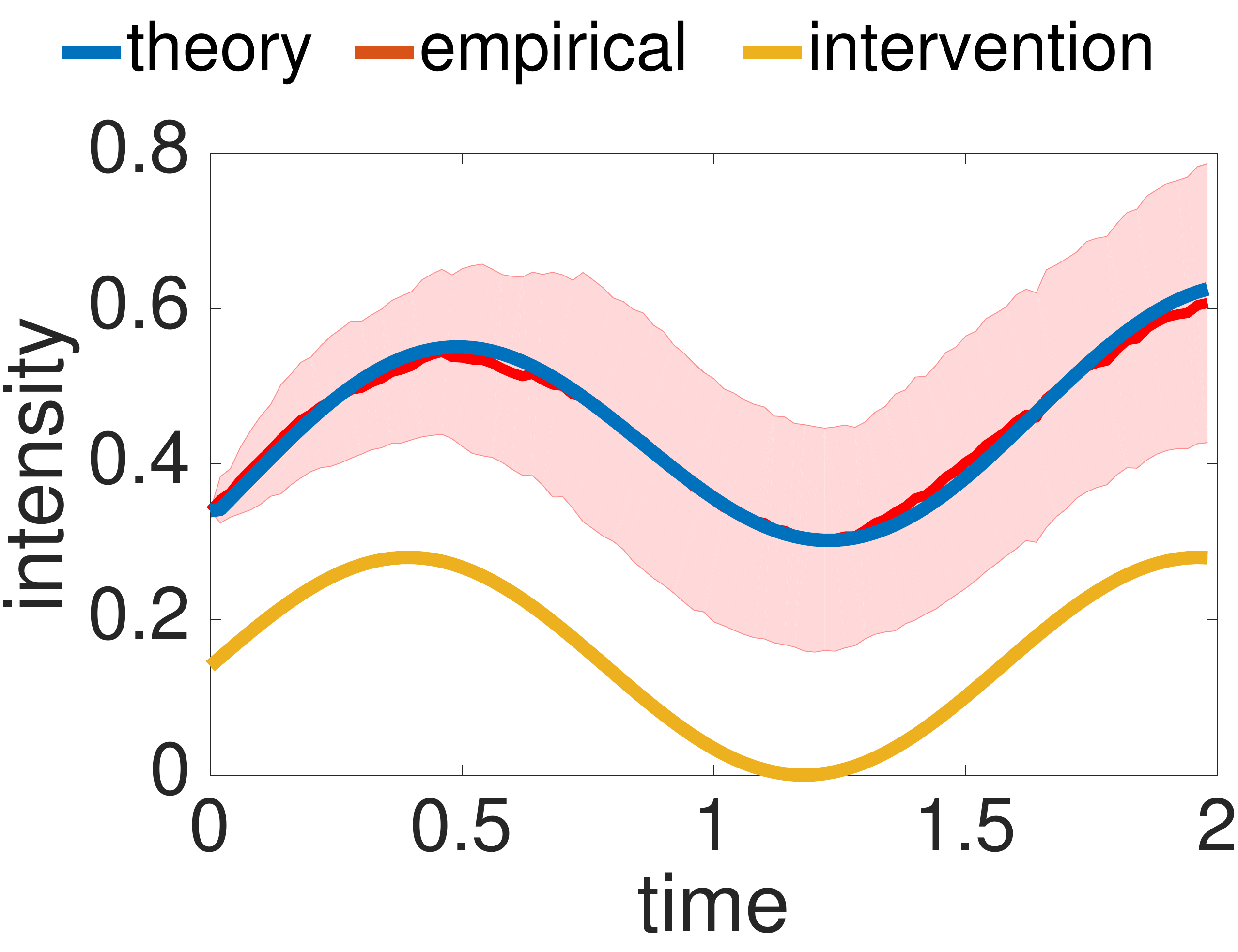} \\
                    \hspace{-4mm}
          \includegraphics[width=0.33\textwidth]{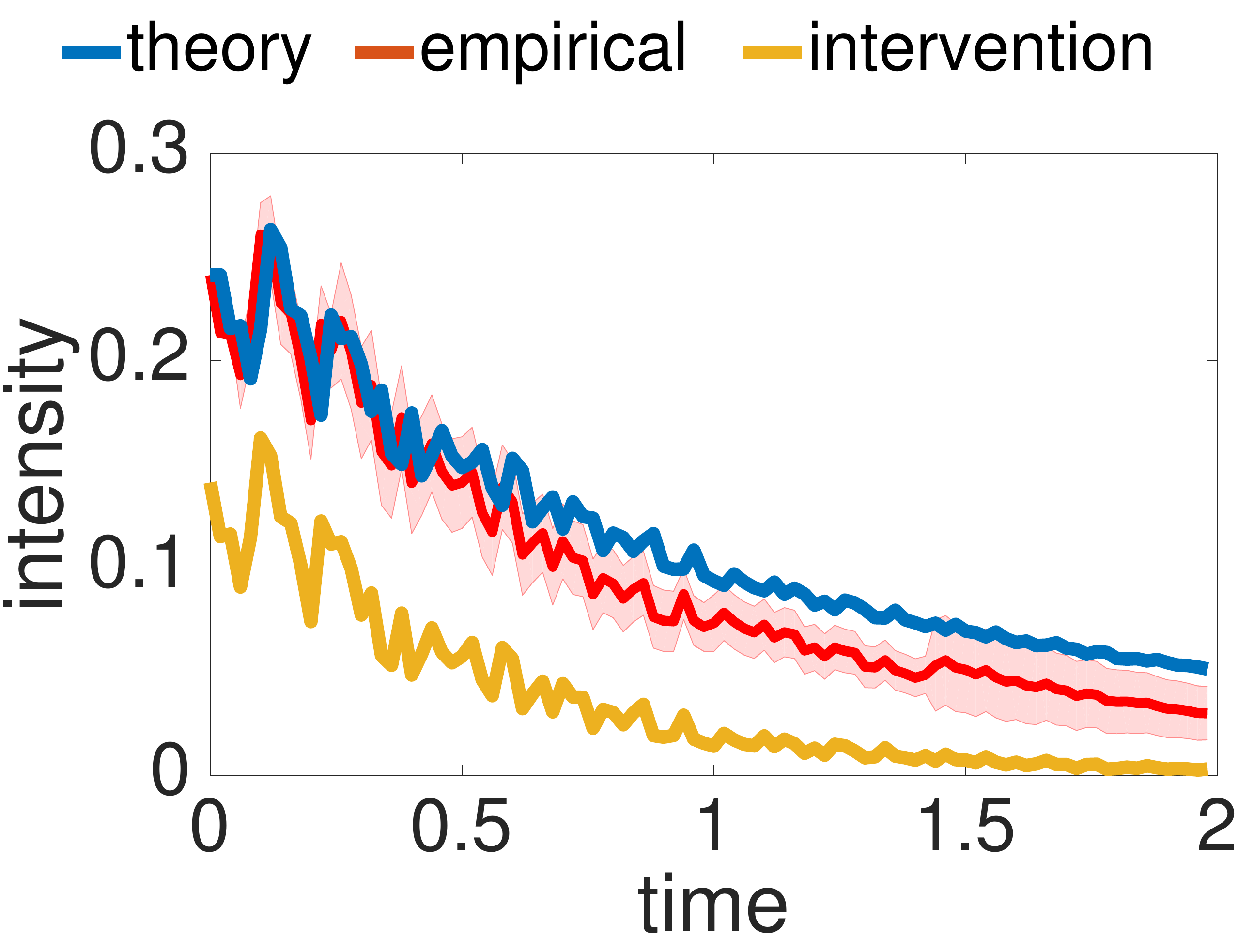} &
          \hspace{-4mm}
          \includegraphics[width=0.33\textwidth]{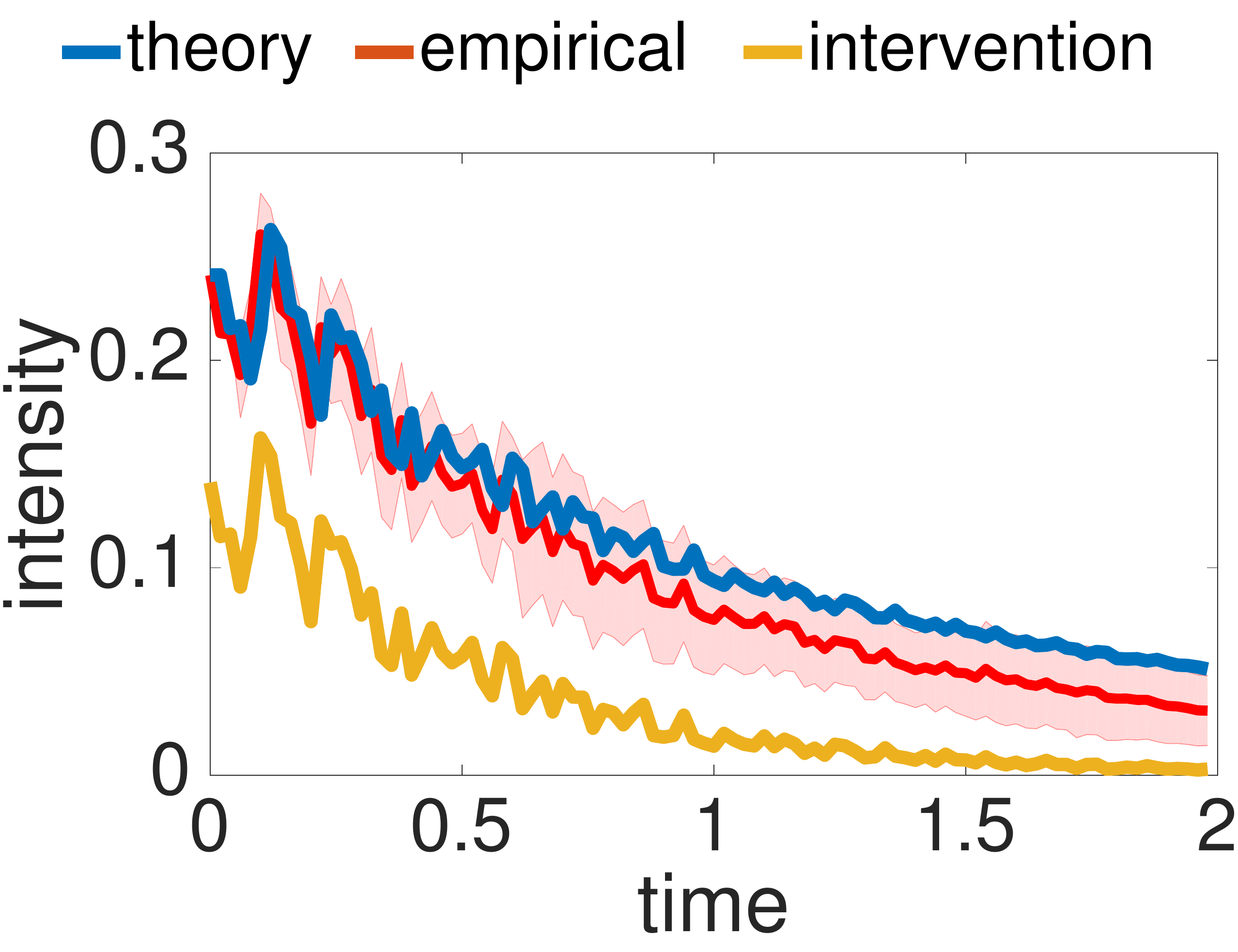} &
          \hspace{-4mm}
          \includegraphics[width=0.33\textwidth]{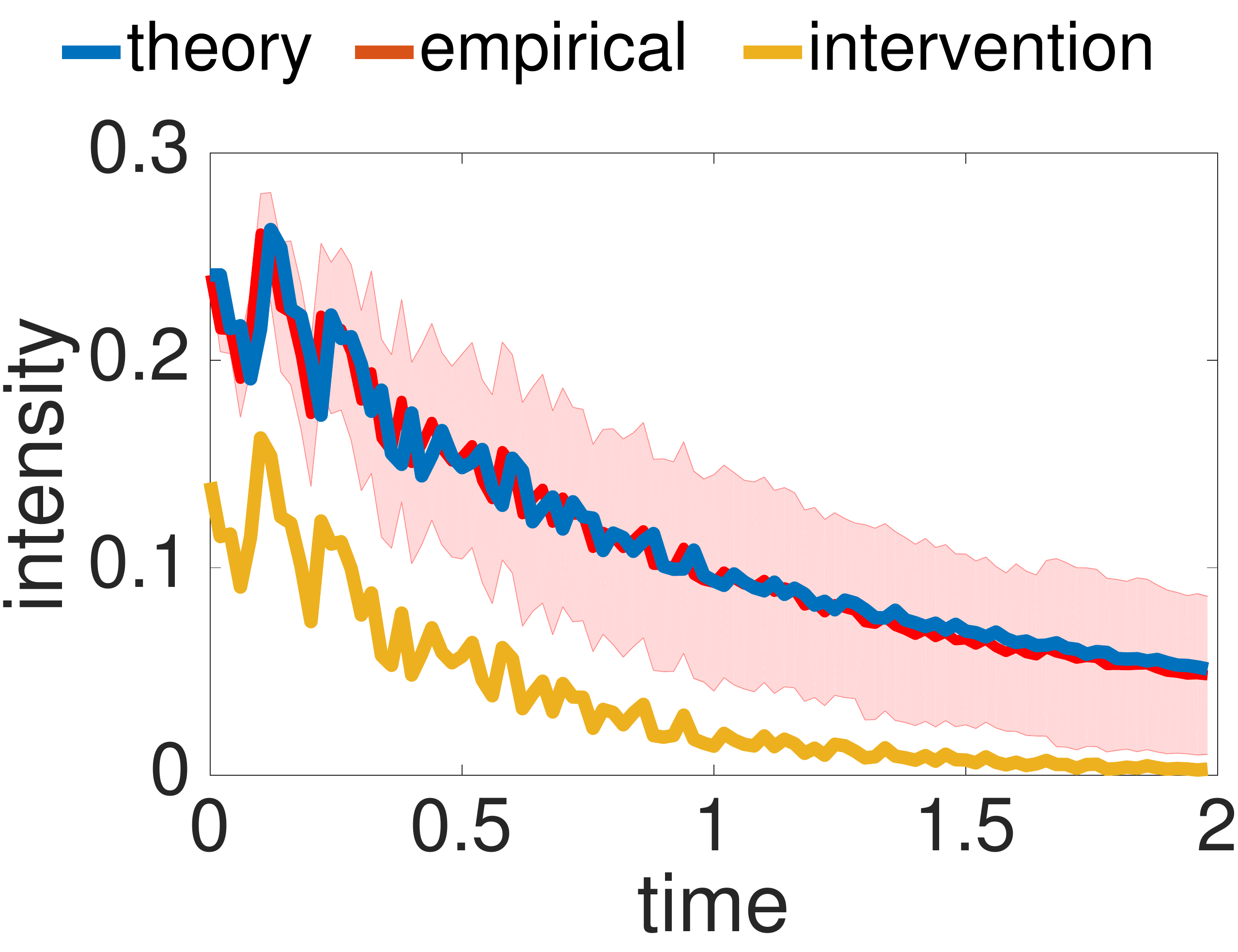}\\
                 \hspace{-4mm}
          \includegraphics[width=0.33\textwidth]{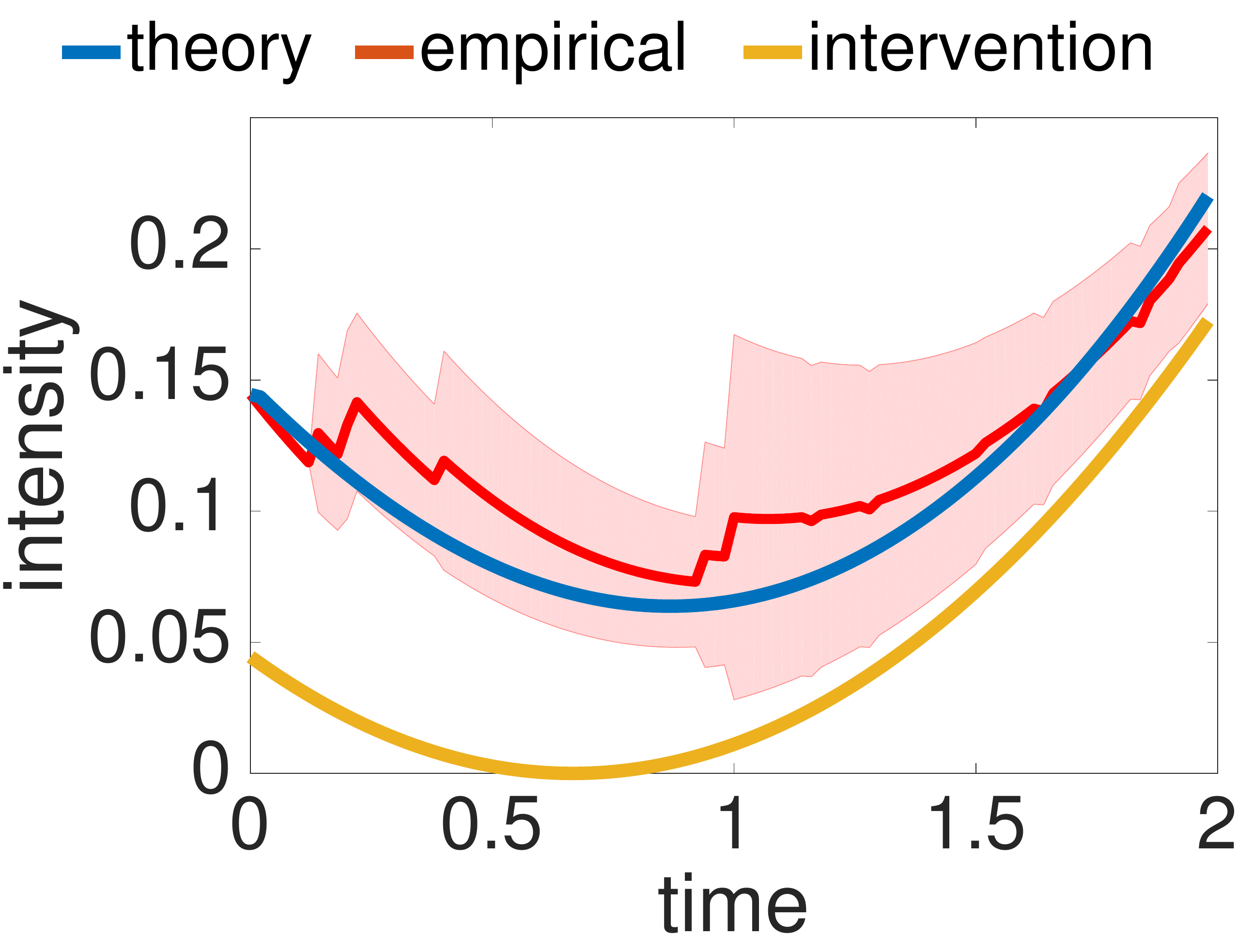} &
          \hspace{-4mm}
          \includegraphics[width=0.33\textwidth]{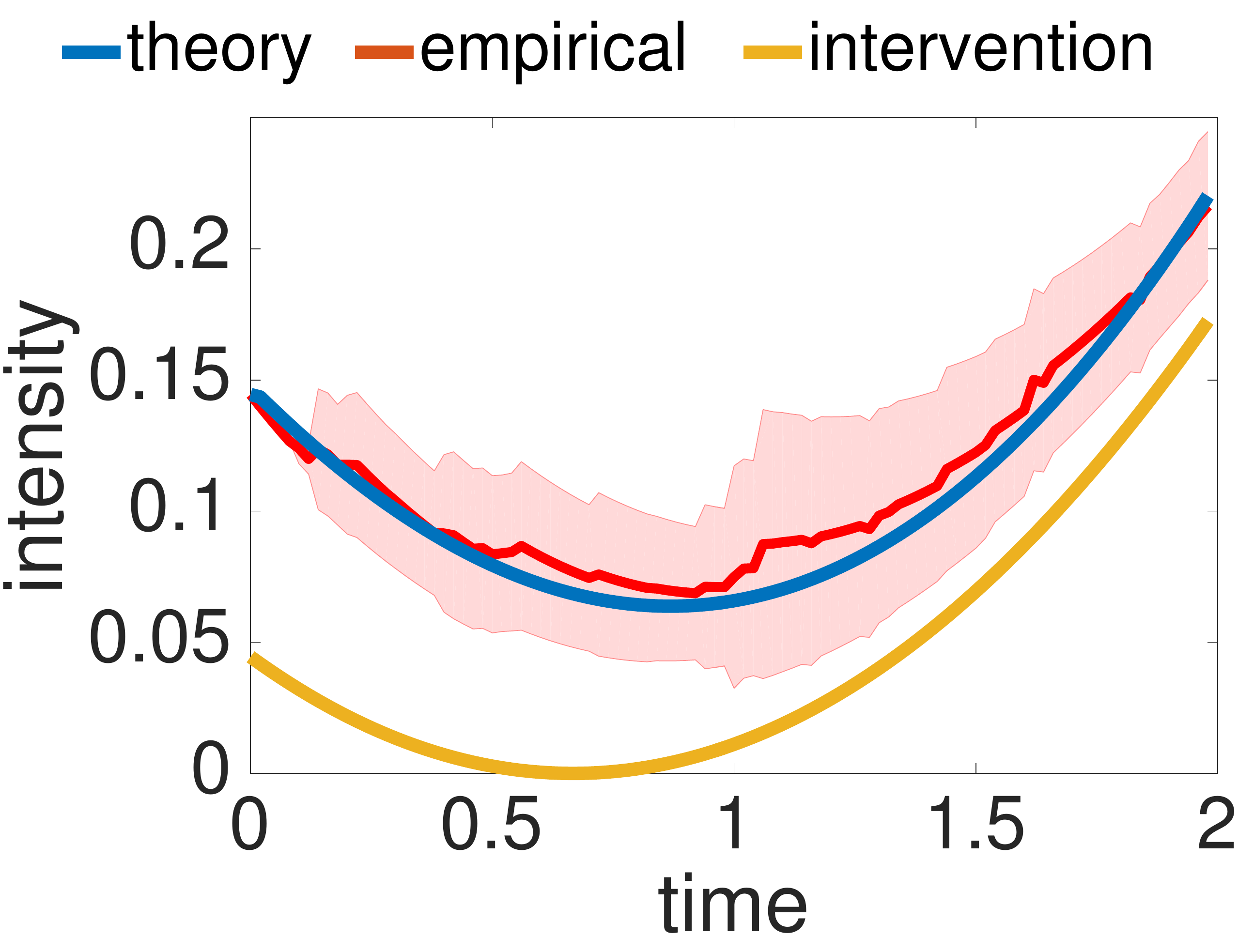} &
          \hspace{-4mm}
          \includegraphics[width=0.33\textwidth]{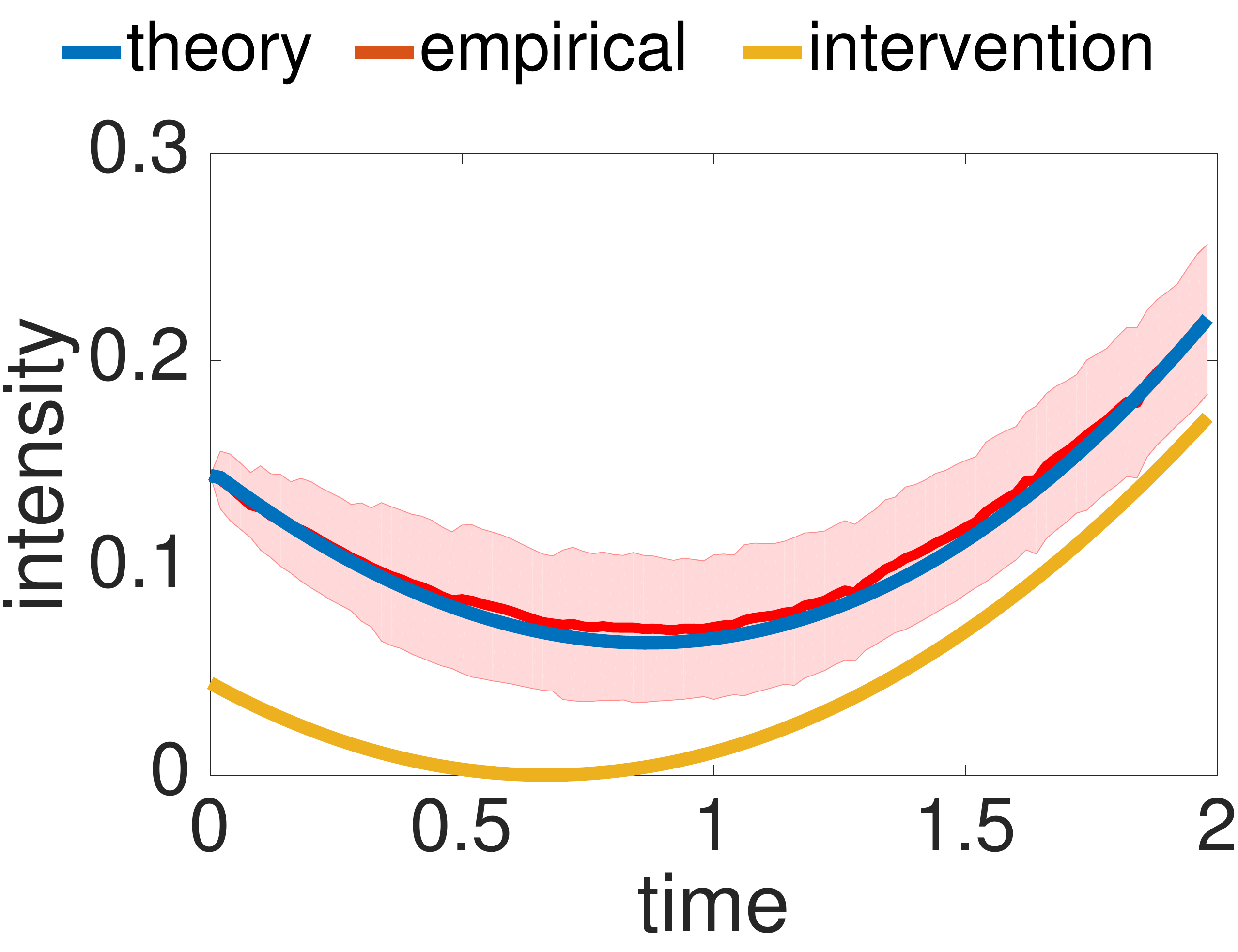} \\

          a) $5$ runs & b) $20$ runs & c) $100$ runs
  \end{tabular}
  \caption{Empirical investigation of theoretical results in Thm. ~\ref{theo:average_general}. blue: theoretical average intensity; red: empirical average intensity and sample standard deviation; orang: general time-varying intensity (interventions)}
  \label{fig:intensity-cont}
\end{figure*}

\section{Temporal Properties}
\label{appen-temporal}
In this section we empirically study the theoretical results of section~\ref{sec:mean}. The empirical mean and standard deviation of the intensity averaged over multiple number of cascades is compared theoretical mean. Besides this, the other purpose of the experiment is to advocate verification process in the synthetic experiments when we used simulation to evaluate the merits of the proposed algorithm and compare to the baselines. In other words, we show that the empirical activity (and hence the average exposure) is very close to its theoretical value and it is justifiable to be used for the comparison.

Fig.~\ref{fig:intensity} demonstrates the activity profile of 3 random users picked in a network of 300 ones simulated 100 times to investigate theorem.~\ref{theo:piecewise_constant_average}. The setting is similar to synthetic experiment in the main paper. Piecewise exogenous intensity (interventions) are picked randomly in $[0.1,0.2]$ with a slight noise. We consider 5 stages for changing the exogenous intensity. The empirical average and standard deviation is compared to theoretical average intensity for 3 different number of runs namely, 5, 20, and 100 times. 
Also, Fig.~\ref{fig:intensity-cont} demonstrate the general case where the exogenous intensity is a time-varying function in theorem. ~\ref{theo:average_general}. We take 3 sample functions to investigate this case; a sinusoidal function; an exponential decaying function with added noise; and a quadratic function.

We observe a couple of interesting facts.
Firstly, it's apparent that by increasing the number of averages the empirical intensity tends to theoretical one very fast. Secondly, as the mean becomes more accurate by increasing the number of cascades the standard deviation increases; e.g., compare the standard deviation in first and third column.
Thirdly, the standard deviation is increasing with time. This is due to the fact that as time passes random elements are aggregated more and this increases variance.

%%%%%%%%%%%%%%%     SECTION     %%%%%%%%%%%%%%%% 

\section{Experiments}
We evaluate our campaigning framework using both simulated and real world data and show that our approach significantly outperforms several baselines.

\subsection{Synthetic data generation}
The network is generated synthetically using varying number of nodes. Initial exogenous intensity is set uniformly at random, $u_m^i \sim \Ucal[0, 0.1]$. Endogenous intensity coefficients (influence matrix elements) are set similarly, $a_{ij} \sim \Ucal[0, 0.1]$. To mimic sparse real networks half of the elements are set to 0 randomly.
The matrix is scaled appropriately such that the spectral radius of the coefficient matrix is a random number smaller than one and the stability of the process is ensured.

The upper bound for the intervention intensity is set randomly from interval $\alpha^i \sim \Ucal[0,0.1]$. The price of each person is set $c_m^i =1$, and the total budget at stage $m$ is randomly generated as $C_m \sim (n/10) \Ucal[0, 0.1]$. 
For the capped exposure maximization case the upper bound is set $\alpha_m^i \sim \Ucal[0, 1]$ and target in least-squares exposure shaping is set similarly $v_m^i \sim (n/10) \Ucal[0, 1]$. Furthermore, the shaping matrix $D$ is set to $I$.
In all the synthetic experiments $\omega = 0.01$ which roughly means loosing 63 \% of influence after 100 units of time (minutes, hours, etc). Furthermore, the exposing matrix is set to the unweighted adjacency matrix \ie, $B_{ij} = 1$ if  and only if $A_{ij} \geq 10^{-4}$. This way the results are reported in terms of the exact number of exposures and are easily interpretable. In general applications any $B$ can be used for example using the influence matrix $A$ yields to a wighted exposure count. In all the synthetic and real experiments the above settings are assumed unless it is explicitly mentioned.

\subsection{Real data description and network inference}
In real data, we use a temporal resolution of one hour and selected the bandwidth $\omega=0.001$ by cross validation. Roughly speaking, it corresponds to loosing almost 50 \% of the initial influence after 1 month. The upper bound for intervention intensity is set uniformly at random with mean equal to empirical intensity learned from data. The upper bound for the cap and target exposure are set  similarly.
For the 10 pairs of cascades we used first 3 months of data to learn the network parameters.
We then drop the exogenous intensity $\mu$, and keep the influence network parameters $A$. By fixing $A$ we use the next 6 months of data to learn the exogenous intensity of sites in the two cascades at each of the $M$ stages and name them $\mu^{c_1}_m$ and $\mu^{c_2}_m$. Given $A$ we find the optimal intervention intensity $u^{opt}_m$ stage by stage, for each of the three exposure shaping tasks assuming $\mu=0$. Then, our prediction is: cascade $c1$ will reach a better objective value at stage $m$ if $dist(u^{opt}_m, \mu^{c_1}_m) < dist(u^{opt}_m, \mu^{c_2}_m)$ and vice versa measured by cosine similarity. The prediction accuracy is then reported as a performance measure.

\subsection{Baselines}
In this section, we describe several baselines we compare our approach.
Most often, these baseline methods utilize a property to prioritize users for budget assignment. 

For the capped exposure maximization problem, we consider the following four baselines: 
 \begin{itemize}
 \item {\bf OPL}: It allocates the budget according to the solution to the dynamic programming in an \emph{open loop} setting, \ie, the decisions on the allocation policy are made once and for all at the initial intervention points at initial time $t=0$. This is very important baseline to which comparison quantify the so called \emph{value of information} in the context of dynamic programming and optimal control. As the name suggests it indicates how much knowing what happened so far helps making decisions for future. For the minimum and capped exposure maximization  creasing $n$ the objective function is normalized by the size of network.

\item {\bf RND}: It assigns a random point in the convex space of feasible solutions.
 %t assigns positive budget to the users proportionally to their sum of out-going influence  ($\sum_u a_{uu'}$). This heuristic allows us (by comparing its results to CAM) to  understand the effect of considering the whole network with respect to only consider the direct (out-going) influence.

\item {\bf PRK}: At each stage it subtracts the previous state ($x_m^i$) from the cap ($\alpha_m^i$) and multiply by the page rank score of the the node ($r^i$) computed with damping factor $0.85$ and allocates the budget proportional to this value, \ie, $u_m^i \propto \max \rbr{(\alpha_m^i-x_m^i) r^i, 0}$. The proposed solution is then projected to the feasible set of actions in that stage and the extra amount is redistributed similarly. The process is iterated until all the budget are allocated.
This baseline assumes that more central users can leverage the total activity, therefore, assigns the budget dynamically to the more connected users proportional to their page rank score.

\item {\bf WEI}:
 This baseline uses sum of out-going influence  ($\, {q}^i = \sum_j a_{ji}$) as a measure of centrality of users. Similar to the previous one it assigns budget dynamically to the users proportionally to 
 $u_m^i \propto \max \rbr{(\alpha_m^i-x_m^i)\, {q}^i, 0}$.
This heuristic allows us to understand the effect of considering the whole network and the propagation layout with respect to only consider the direct (out-going) influence.

 \end{itemize}

For the max-min exposure shaping problem, we implement the following four baselines: 
% Again, since the baselines are incapable of finding a proper number of users to assign the budget, we chose to manipulate manipulate 50 \% of users. They only can find an ordering of users.

\begin{itemize}
\item {\bf OPL}: Similar to the previous objective it represents the open loop solution.

\item {\bf RND}: Similar to the previous objective it allocates the budget randomly within the feasible set.
\item {\bf WFL}: It takes a \emph{water filling} approach. It sorts the users in ascending order of the exposure in the previous stage. Then allocates budget to the first users until the the summation of its previous exposure and the allocated budget reaches the second lowest value or it violates a constraint. Then, assigns the budget to these two until they reach the third user with lowest exposure or a constraint is violated. This process is continued until the budget is allocated. 

\item {\bf PRP}: It allocates the budget inversely proportional to the the exposure at the previous stage.
 \end{itemize}

For the least-square exposure shaping problem, we compare our method with four baselines:
\begin{itemize}
\item {\bf OPL}: Similar to the previous objective it represents the open loop solution.

\item {\bf RND}: Similar to the previous objective it allocates the budget randomly within the feasible set.

 \item {\bf GRD}: It finds the difference between the exposure at previous stage ($x_m^i$ and the target from $v$ and sorts them decreasingly. Then, allocates budget one at a time until a constraint is violated. It iterates over the users until the budget is fully allocated.
 
 \item {\bf REL}: Similar to the above finds the difference from the target but allocates the budget proportionally, \ie, $u_m^i \propto \max \rbr{(v^i-x_m^i), 0}$ for all users. If one allocation violates a constraint the extra amount is reallocated in the same manner.
 
 \end{itemize}

\begin{figure*}[!t]
  % \vspace{-3mm}
  \centering
  \setlength{\tabcolsep}{6pt}
  \begin{tabular}{ccc}
          \hspace{-3mm}
          \includegraphics[width=0.27\textwidth]{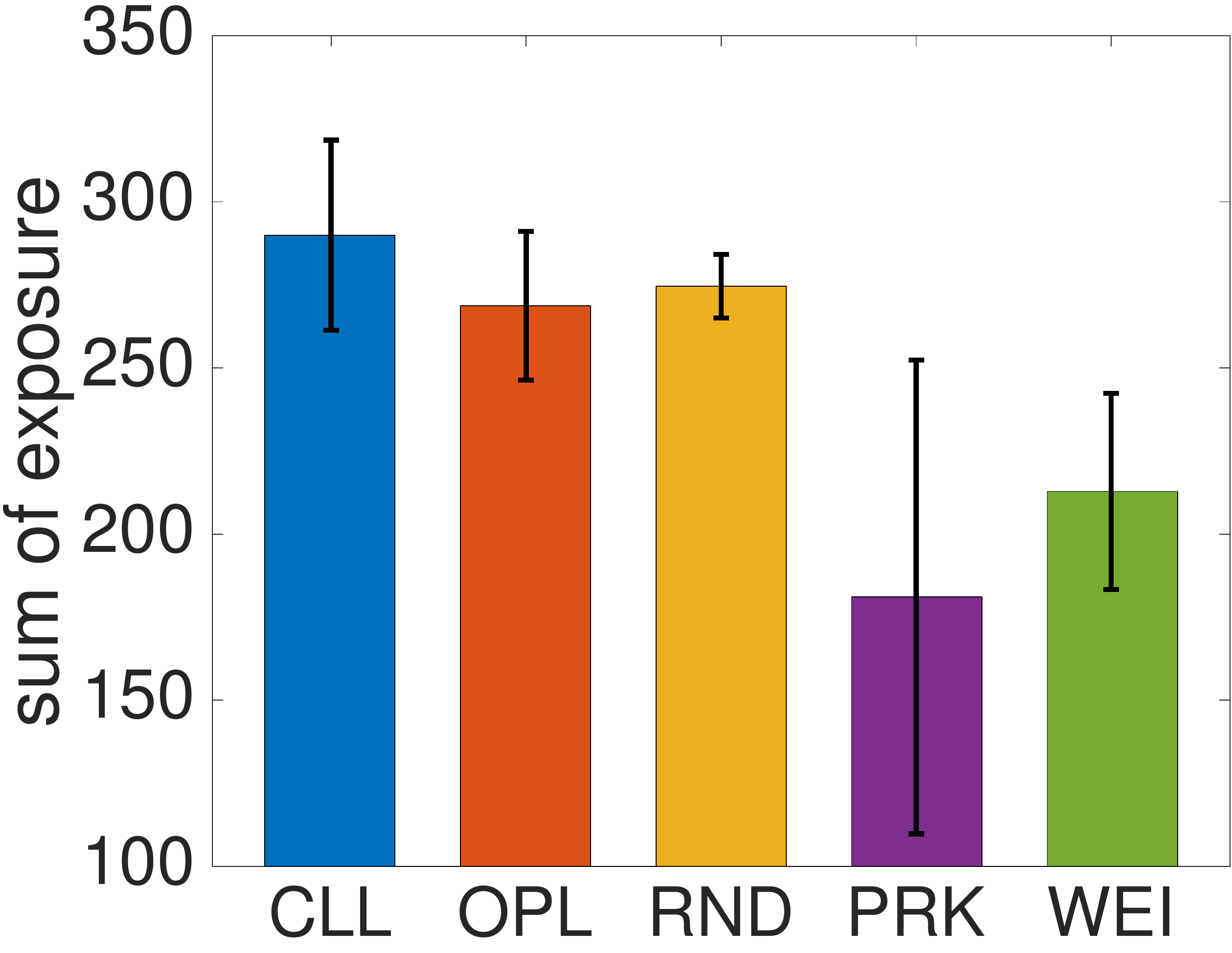} &
          \hspace{-3mm}
          \includegraphics[width=0.27\textwidth]{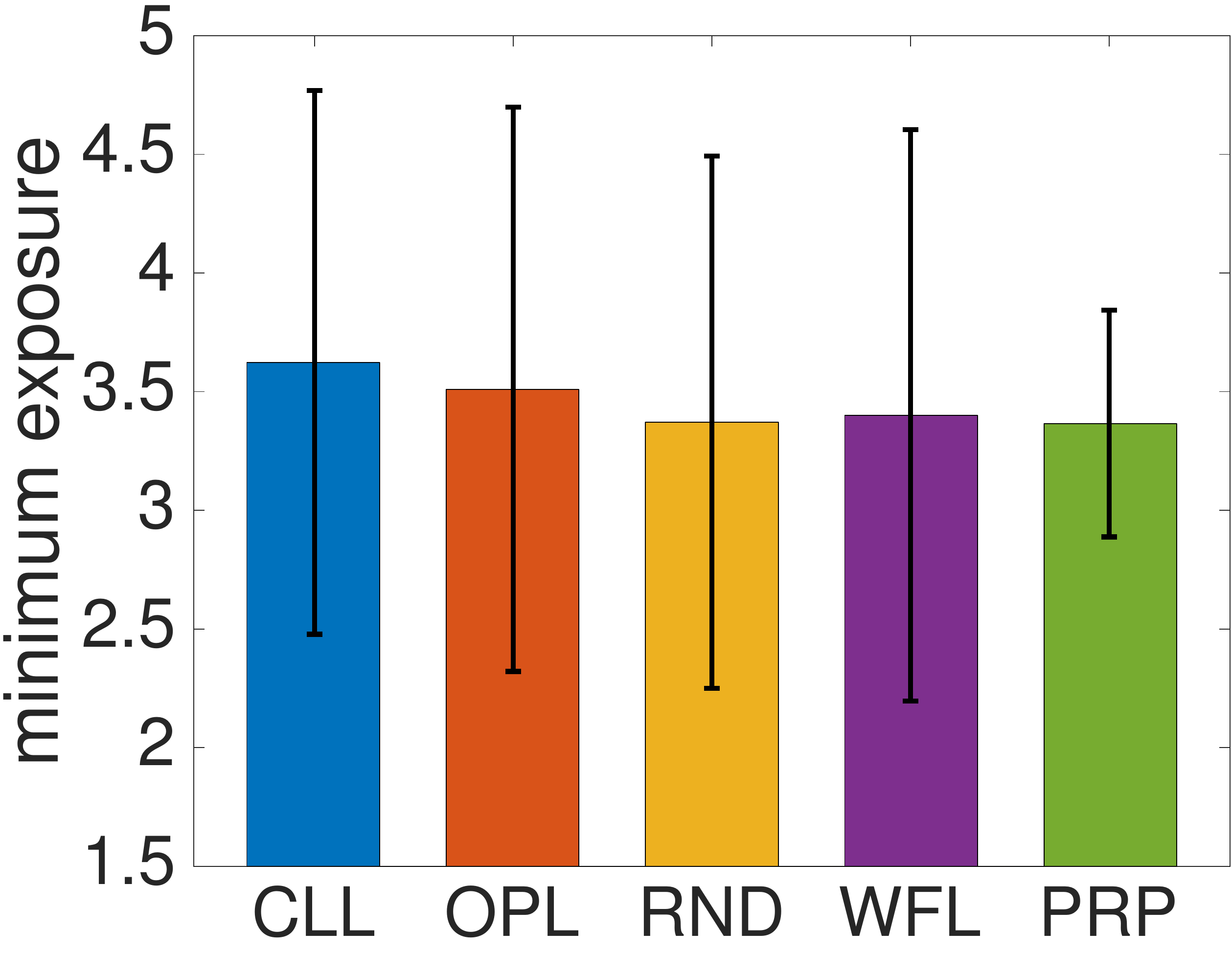} &
          \hspace{-3mm}
          \includegraphics[width=0.27\textwidth]{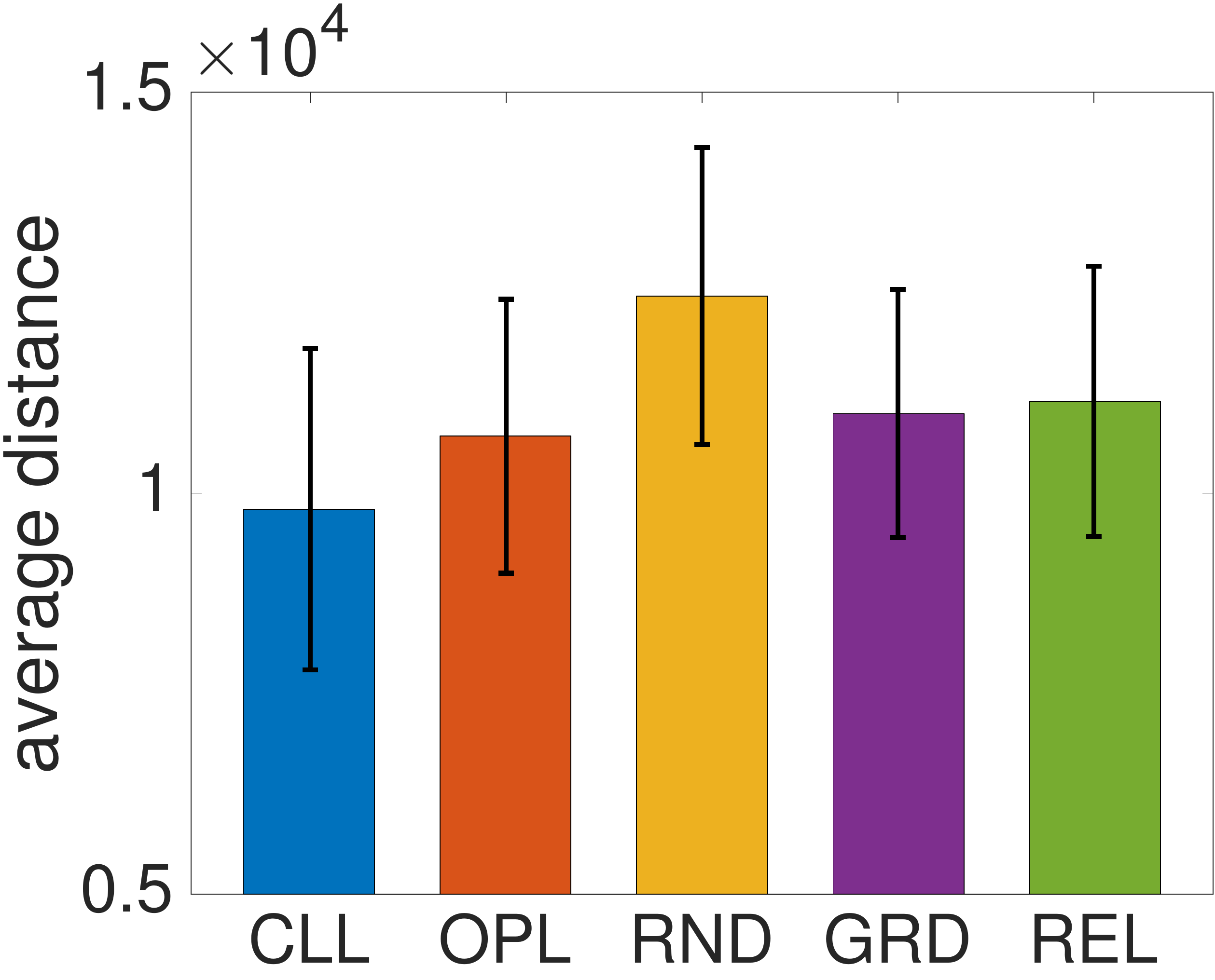} \\
          \hspace{-2mm}
          a) Capped maximization &  
          \hspace{-2mm} 
          b) Minimum maximization & 
          \hspace{-2mm}
          c) Least-squares shaping
  \end{tabular}% \vspace{-3mm}
  \caption{The objective on simulated events and synthetic network; $n=300$, $M=6$, $T=40$}
  \label{fig:synth-results}
\end{figure*}

\subsection{Campaigning results on synthetic networks}
In this section, we experiment with a synthetic network of $300$ nodes. %Details of the experimental setup and parameter setting are found in appendix \ref{appen-synth}.
We focus on three tasks: capped exposure maximization, minimax exposure shaping, and least square exposure shaping. To compare the methods we simulate the network with the prescribed intervention intensity and compute the objective function based on the events happened during the simulation. The mean and standard deviation of the objective function out of 10 runs are reported.

Fig.~\ref{fig:synth-results} summarizes the performance of the proposed algorithm (CLL) and 4 other baselines on different campaigning tasks.
For {\bf CEM}, our approach consistently outperforms the others by at least 10. This means it exposes each user to the campaign at least 10 times more than the rest consuming the same budget and within the same constraints.
The extra 20 units of exposures of over OPL or value of information shows how much we gain by incorporating a dynamic closed-loop solution as opposed to open-loop one-time optimization over all stages. 
For {\bf MEM}, the proposed method outperforms the others by a smaller margin, however, the 0.1 exposure difference with the second best method is not trifling. This is expected as lifting the minimum exposure is a difficult task~\cite{farajtabar2014activity}. 
For {\bf LES}, results demonstrate the superiority of CLL by a large margin. The $10^3$ difference with the second best algorithm aggregated over 6 stages roughly is translated to $\sqrt{10^3/6} \sim 13$ difference in the number of exposures per user. Given the heterogeneity of the network activity and target shape, this is a significant improvement over the baselines.
Appendix \ref{appen-synth} includes further results on varying number of nodes, number of stages, and duration of each stage.

\begin{figure}[!t]
  \centering
  \begin{tabular}{cccc}
          \hspace{-3mm}
          { \footnotesize \rotatebox{90}{Capped Maximization}} &
          \hspace{-1mm}
          \includegraphics[width=0.305\textwidth]{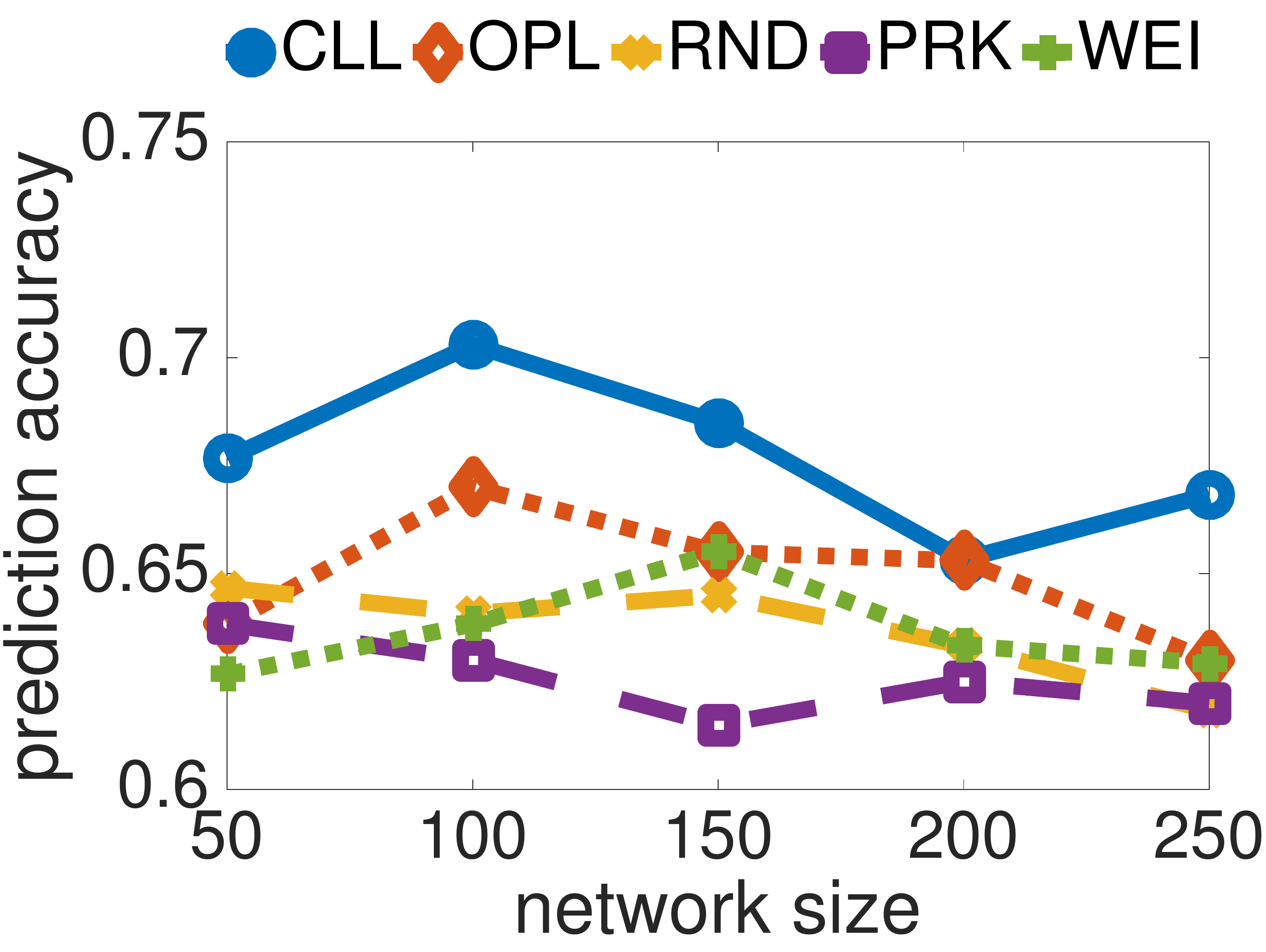} &
          \hspace{-3mm}
          \includegraphics[width=0.305\textwidth]{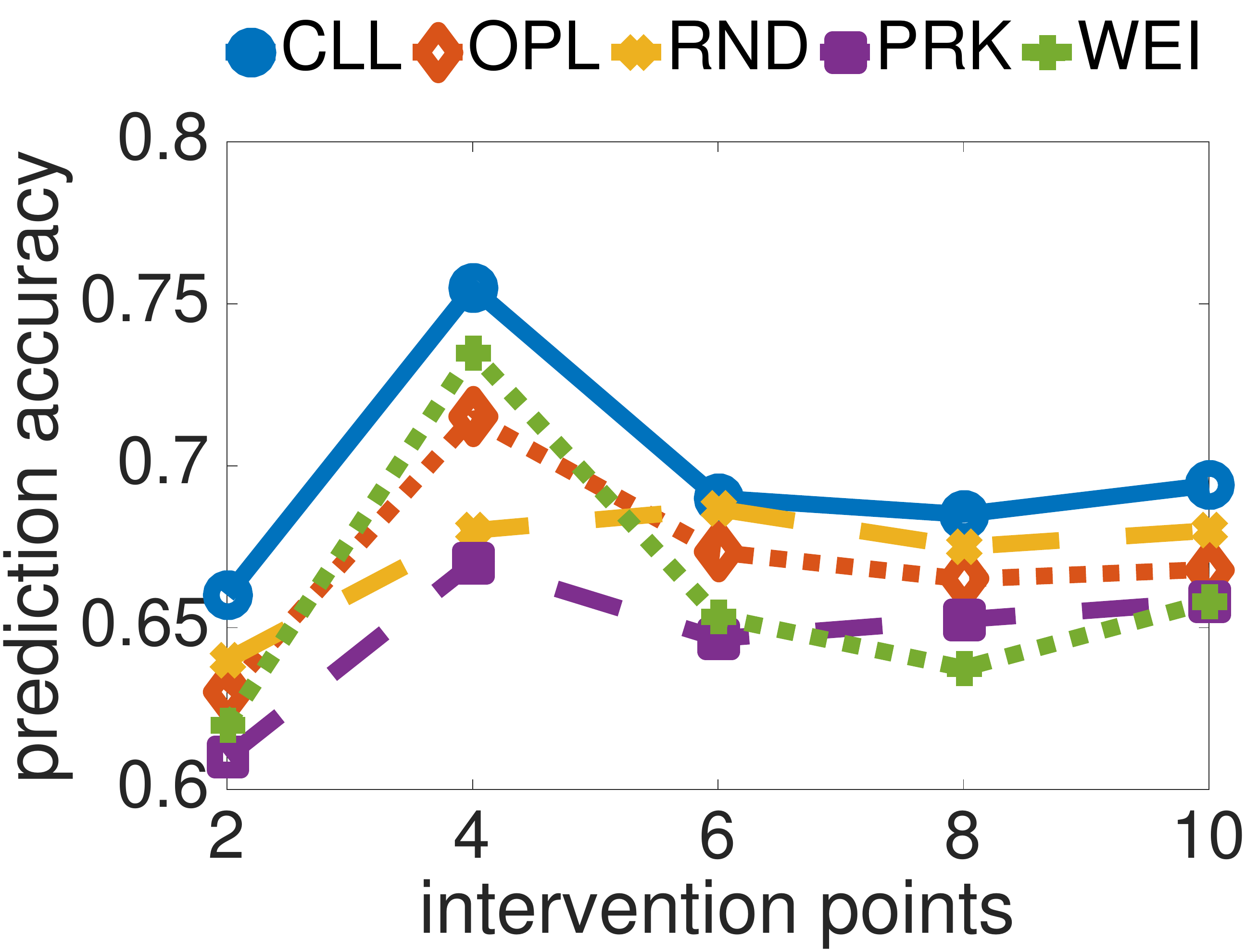} &
          \hspace{-3mm}
          \includegraphics[width=0.305\textwidth]{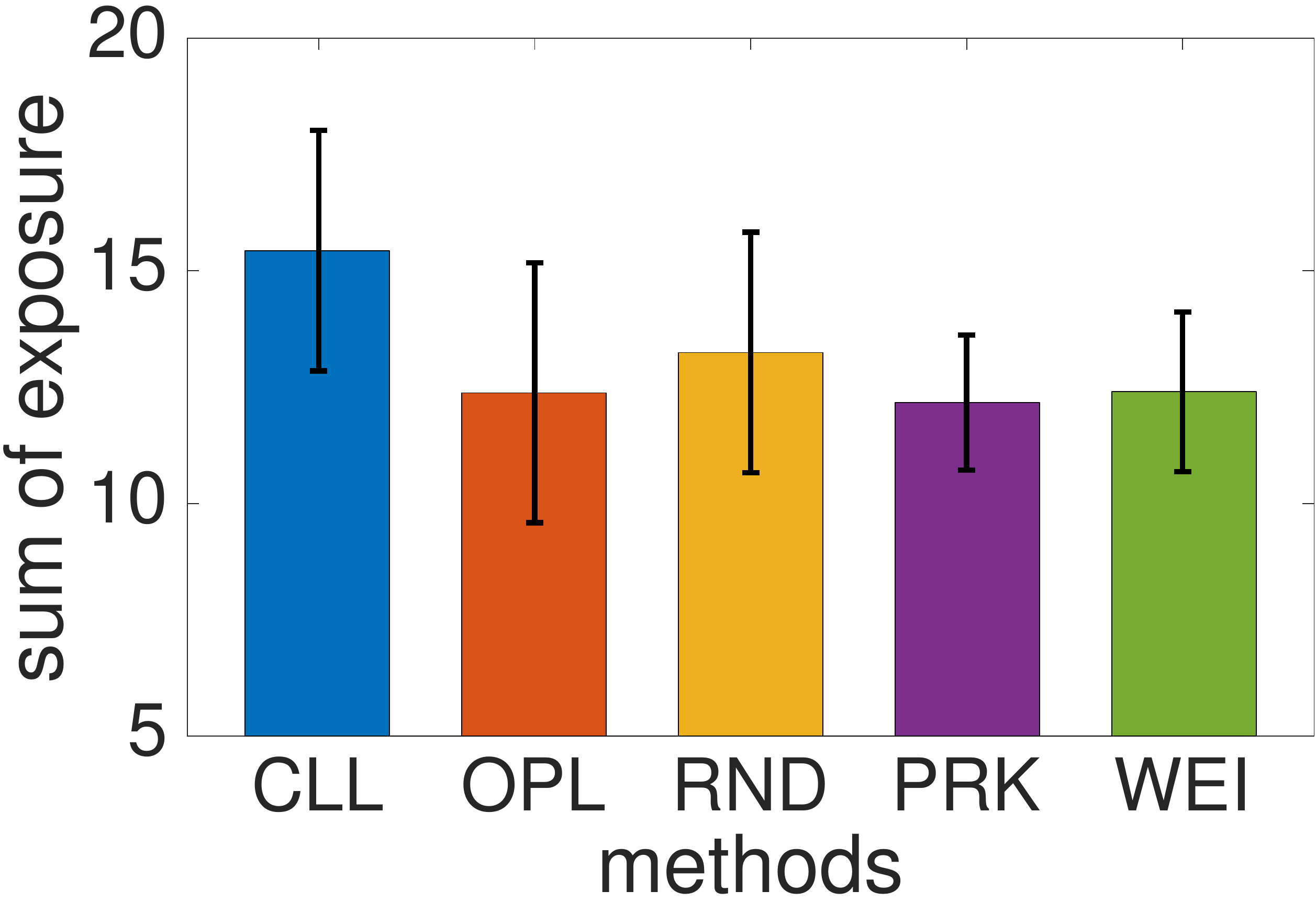}  \\
          \hspace{-3mm}
          { \footnotesize \rotatebox{90}{Minimum Maximization}} &
          \hspace{-1mm}
          \includegraphics[width=0.305\textwidth]{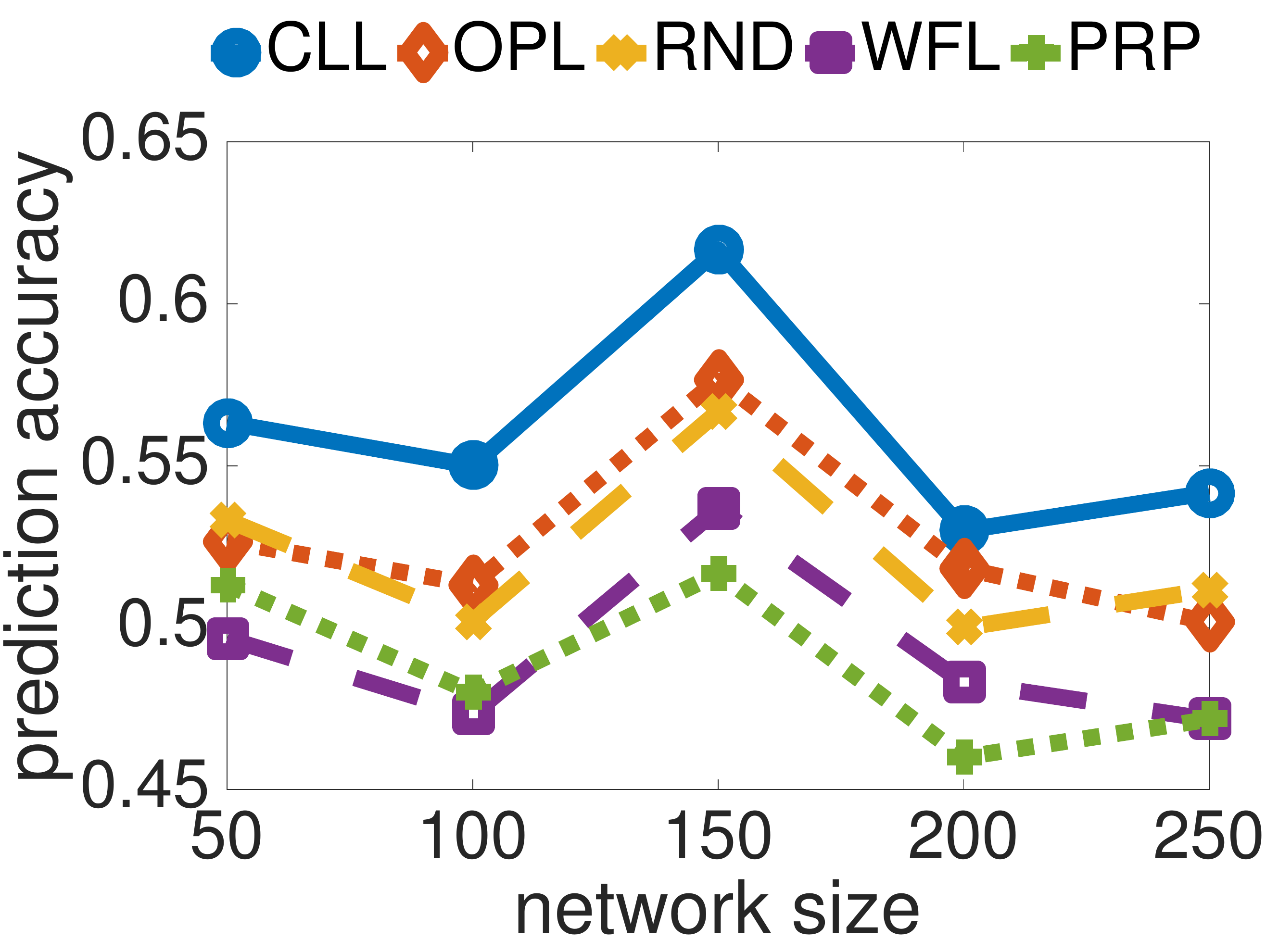} &
          \hspace{-3mm}
          \includegraphics[width=0.305\textwidth]{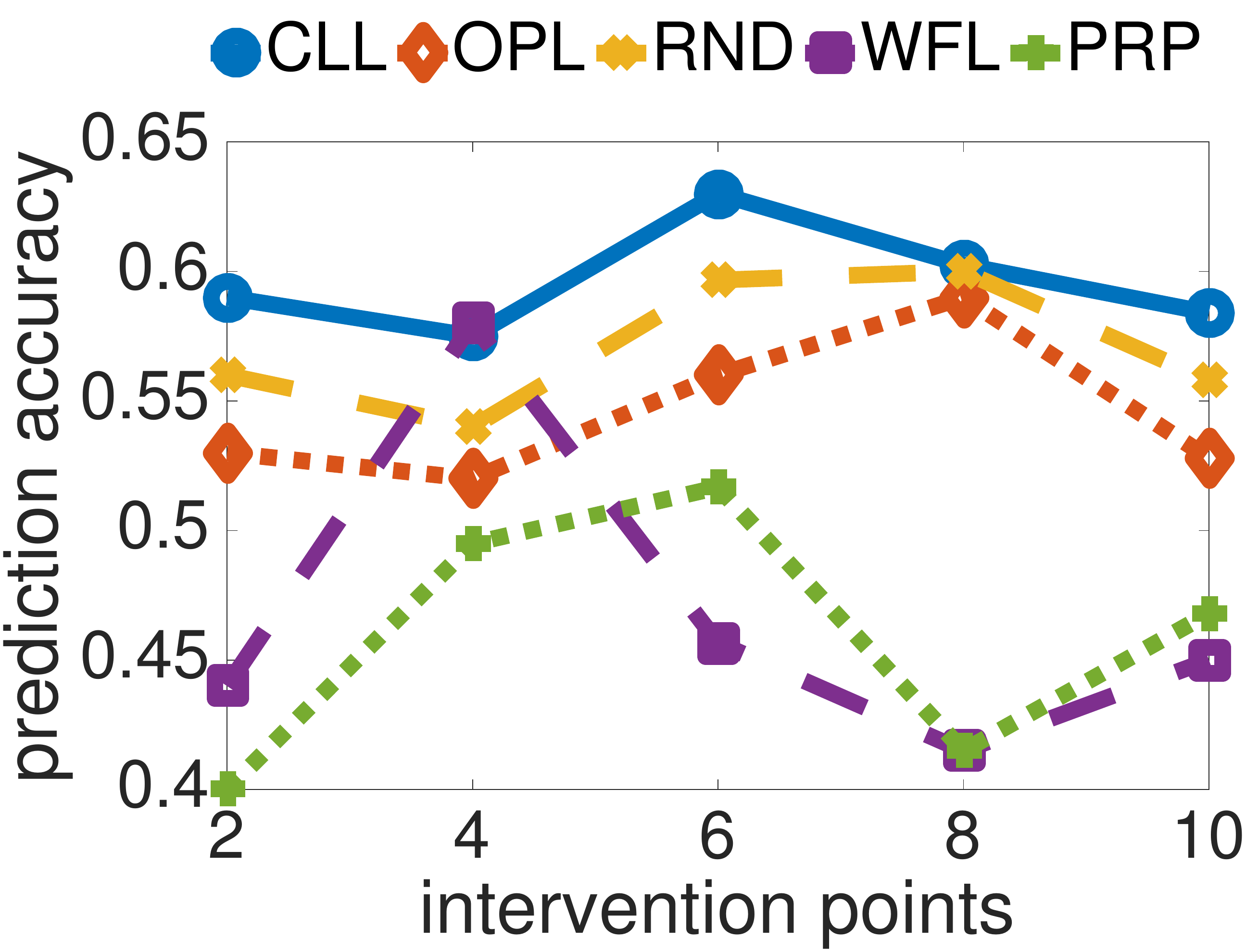} &
          \hspace{-3mm}
          \includegraphics[width=0.305\textwidth]{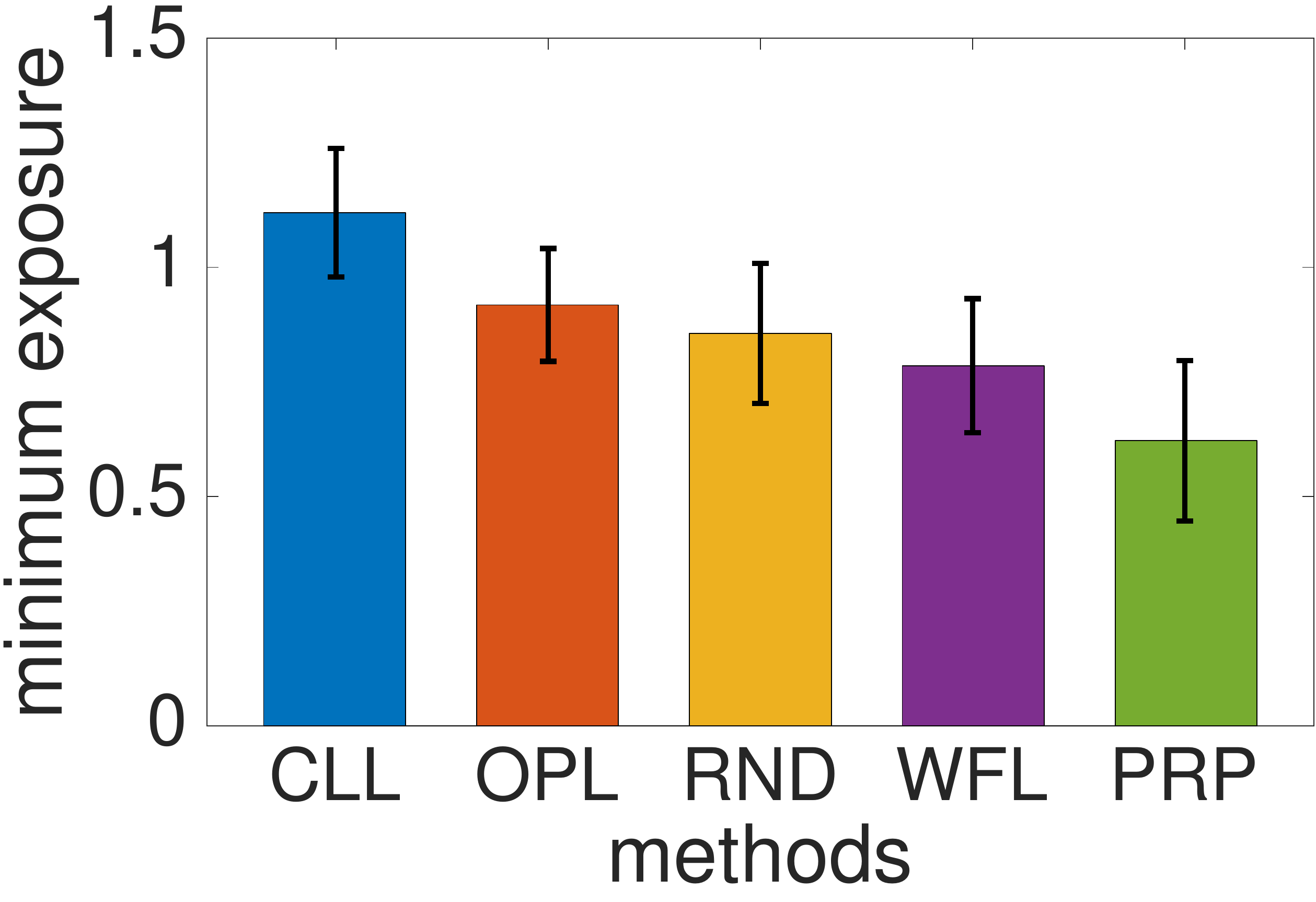} \\
           \hspace{-3mm}
          { \footnotesize \rotatebox{90}{Least-squares Shaping}} &
           \hspace{-1mm}
          \includegraphics[width=0.305\textwidth]{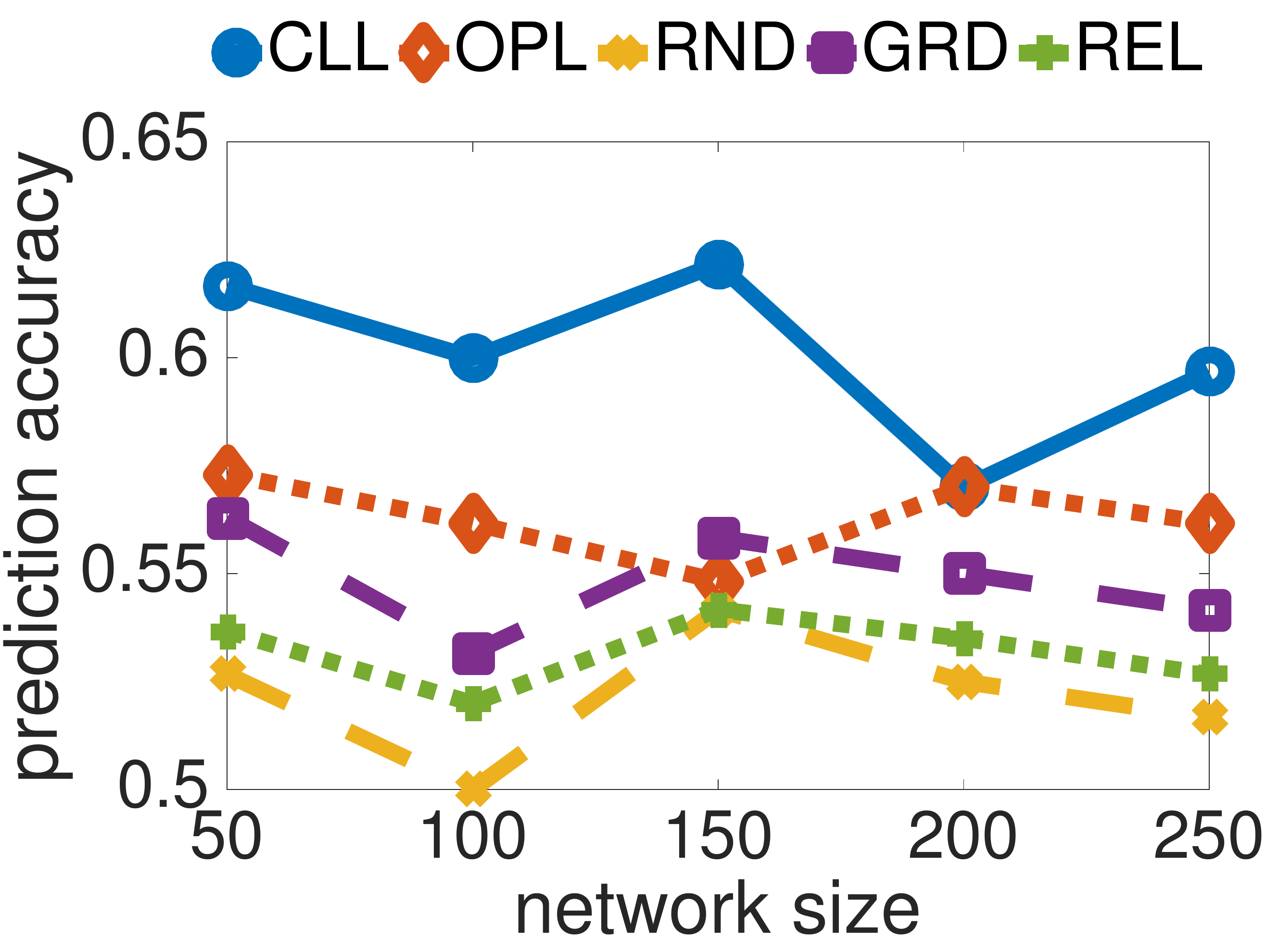} &
          \hspace{-3mm}
          \includegraphics[width=0.305\textwidth]{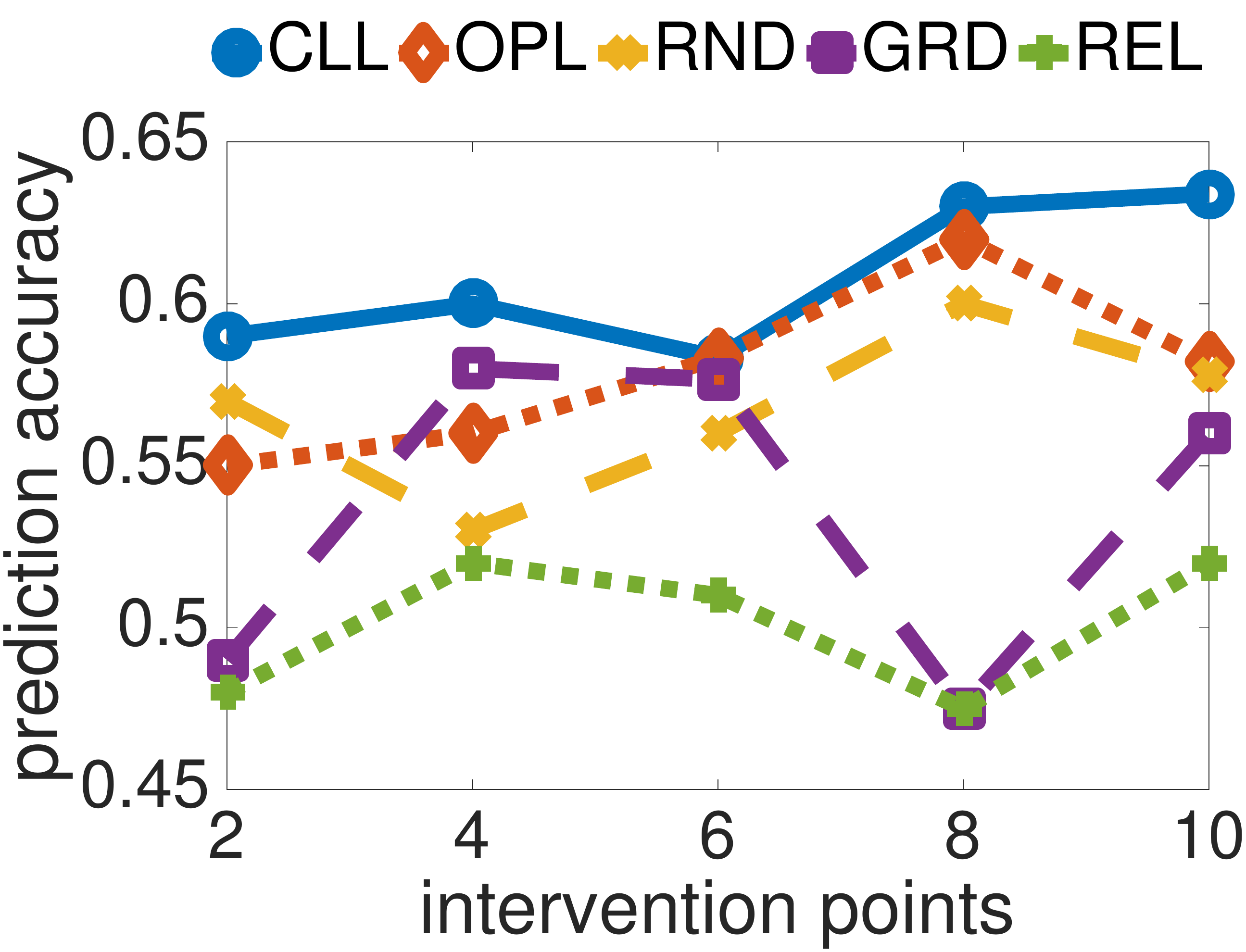} &
          \hspace{-3mm}
          \includegraphics[width=0.305\textwidth]{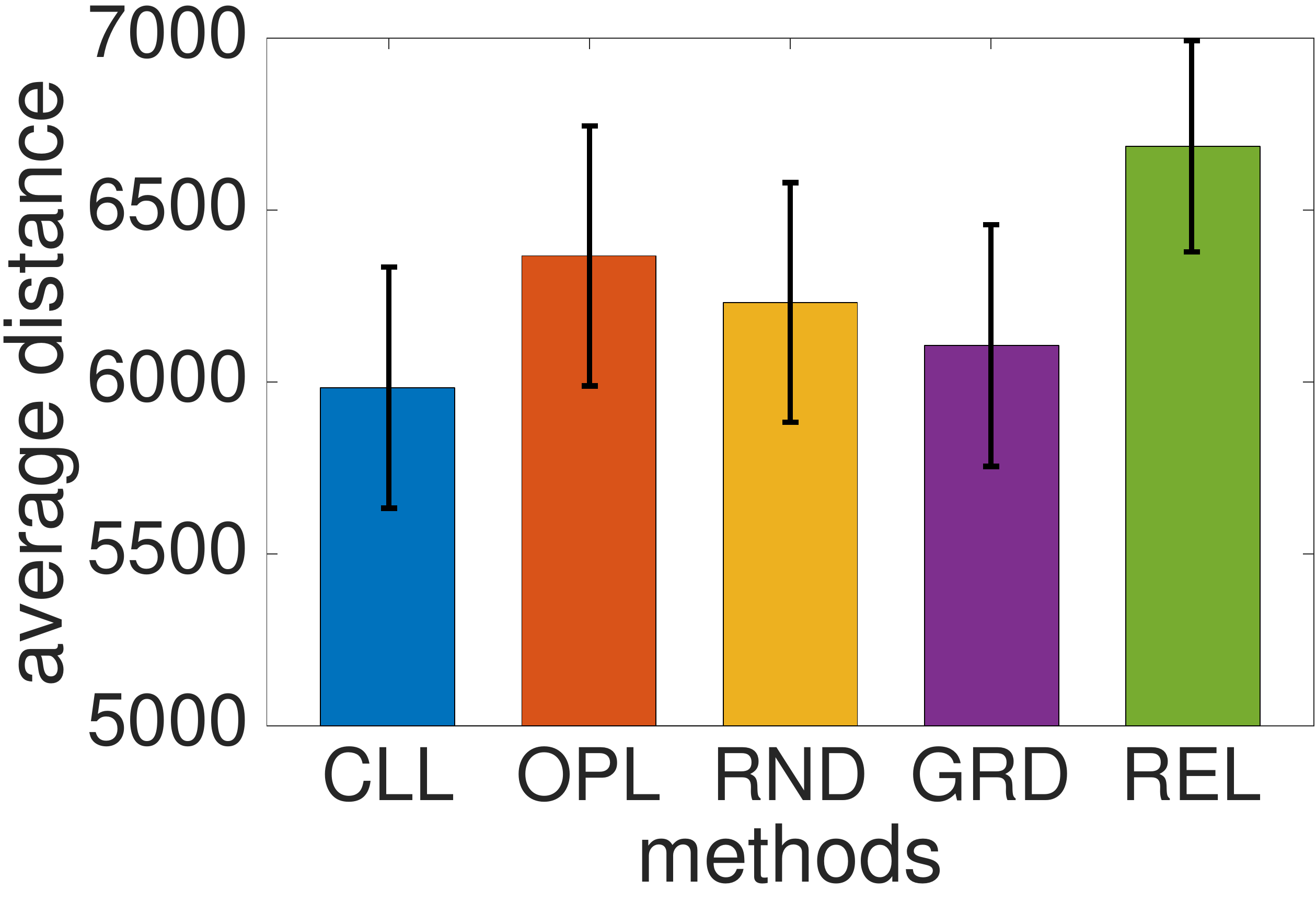} \\
          &
           Performance vs. \# users &
           Performance vs. \# points &
          Objective function \\
  \end{tabular} 
  \caption{real world dataset results; $n=300$, $M=6$, $T=40$}
  \label{fig:real-results}
\end{figure}

\subsection{Campaigning results on real world networks}
We also evaluate the proposed framework on real world data. 
To this end, we utilize the MemeTracker dataset~\cite{leskovec2009meme} which contains the information flows captured by hyperlinks between different sites with timestamps during 9 months. 
This data has been previously used to validate Hawkes process models of social activity~\cite{ZhoZhaSon13,YanZha13}.
%
%We take 5 largest clusters and extract top 1000 most repeated sites and their blog posts. Each is a cascade of events around a common subject. 
%
For the real data, we utilize two evaluation procedures. First, similar to the synthetic case, we simulate the network, but now on a network based on the learned parameters from real data. 
However, the more interesting evaluation scheme would entail carrying out real intervention in a social media platform. Since this is very challenging to do, instead, in this evaluation scheme we used held-out data to mimic such procedure.
Second, we form 10 pairs of clusters/cascades by selecting any 2 combinations of 5 largest clusters in the Memetracker data. Each is a cascade of events around a common subject. 
For any of these 10 pairs, the methods are faced to the question of predicting which cascade will reach the objective function better. 
They should be able to answer this by measuring how similar their prescription is to the real exogenous intensity. 
 The key point here is that the real events happened are used to evaluate the objective function of the methods.
Then the results are reported on average prediction accuracy on all stages over 10 runs of random constraint and parameter initialization on 10 pairs of cascades. %The details of the experimental setup is further explained in Appendix~\ref{appen-synth}. 

Fig.~\ref{fig:real-results}, left column illustrates the performance with respect to increasing the number of users in the network. The performance drops slightly with the network size. This means that prediction becomes more difficult as more random variables are involved.
The middle panel shows the performance with respect to increasing the number of intervention points. Here, a slight increase in the performance is apparent. As the number of intervention points increases the algorithm has more control over the outcome and can reach the objective function better.

Fig.~\ref{fig:real-results} top row summarizes the results of {\bf CEM}. The left panel demonstrates the predictive performance of the algorithms. CLL consistently outperforms the rest. With 65-70 \% of accuracy in predicting the optimal cascade.
The right panel shows the objective function simulated 10 times with the learned parameters for 
network of $n=300$ users on $6$ intervention points. 
The extra 2.5 extra exposure per user compared to the second best method with the same budget and constraint would be a significant advertising achievement.
Among the competitors OPL and RND seem to perform good. 
If there where no cap over the resultant exposure, all methods would perform comparably because of the linearity of sum of exposure. However, the successful method is the one who manage to maximize exposure considering the cap. 
Failure of PRK and WEI indicates that structural properties are not enough to capture the influence. Compared to these two, RND performs better in average, however exhibits a larger variance as expected.

Fig.~\ref{fig:real-results} middle row summarizes the results for {\bf MEM} and shows CLL outperforms others consistently.
CLL still is the best algorithm and OPL and RND are the significant baselines. Failure of WFL and PRP shows the network structure plays a significant role in the activity and exposure processes. 

The bottom row in Fig.~\ref{fig:real-results} demonstrates the results of {\bf LES}. CLL is still the best method.
Among the competitors, OPL is still strong but RND  is not performing well for this task. The objective function is summation of the square of the gap between target and current exposure. This explains why GRD is showing a comparable success, since, it starts with the highest gap in the exposure and greedily allocates the budget.

%%%%%%%%%%%%%%%%%%%%%%%%%%%%%%%%%%%%%%%%%%%%%%%%%%%%
\section{Related Work}
\label{appen-related}

Exposure shaping problems are significantly more challenging than traditional influence maximization problems, which aim to identify a set of users who influence others in the network and trigger a large cascade of adoptions~\cite{richardson2002mining,kempe2003maximizing}.  
First, in influence maximization, the state of each user is often assumed to be binary.  However, such assumption does not capture the recurrent nature of social activity. 
Second, while influence maximization methods identify a set of users to provide incentives, they do not typically provide a quantitative prescription on how much incentive should be provided to each user. 
Third, exposure shaping concerns about a larger variety of target states, such as minimum exposure requirement and homogeneity, not just maximization.

Existing work in stochastic optimal control includes jump diffusion 
stochastic differential equations (SDE)~\cite{hanson2007applied,hu2008partial} which focuses on controlling the SDEs with the jump term driven by Poisson processes not for Hawkes processes. 
Inspired by the opinion dynamics model proposed in~\cite{de2015modeling}, the authors in~\cite{wang2016steering} proposes  a multivariate jump diffusion process framework for modeling opinion dynamics over networks and determining the control over such networks.

In \cite{bloembergen2014influencing}, a continuous action iterated prisoners' dilemma was used to model
the interactions in a social network and extended by incorporating a mechanism for external influence 
on the behavior of individual nodes. 
Markov Decision Process (MDP) framework is proposed to develop several scheduling
algorithms for optimal control of information epidemics with susceptible-infected (SI) model on 
Erd\H{o}s-R\'{e}nyi and scale-free networks \cite{kandhway2015campaigning}.
In~\cite{chen2014optimal}, the authors provided an analytically tractable model for information 
dissemination over networks and solved the optimal control signal distribution time for minimizing 
the accumulated network cost via dynamic programming. 
Furthermore, \cite{karnik2012optimal} formulated the maximization of spread of a given message in the population within the stipulated time as continuous-time deterministic optimal control problem.

In contrast, our work has been built on the well-developed theory of point processes~\cite{AalBorGje08, DalVer2007}.
Their usage in modeling activity in social network is becoming increasingly popular \cite{lian2015multitask, parikh2012conjoint, hall2014tracking, FarGomDuZamZhaSon15, PerWol13}. 
More specifically, we utilizes the Hawkes process~\cite{Hawkes71} which its self-exciting property has been proved to be an appropriate choice in modeling processes of and on the networks:
\cite{LinAdaRya14, ZhoZhaSon13, IwaShaGha13, YanZha13, BluBecHelKat12} model and infer the social activity in networks
Based on Hawkes process assumption of activity in social networks
\cite{FarWanGomLiZhaSon15, TraFarSonZha15, chow2015influence, xu2016learning, he2015hawkestopic, competing15icdm} study one or several phenomena in the social network.
In~\cite{karimi2016smart} authors proposed a broadcasting algorithm to maximize the visibility of posts in Twitter. It only consider the direct followers and does not involve peer influence in propagation process.
Our work is closely related to~\cite{farajtabar2014activity}, which is extended 
in two significant directions here: First, we generalize their result on driving a time-dependent 
average intensity in the case where the exogenous intensity is not constant. 
Second, instead of one-shot optimization we pose the problem as a multi-stage optimal control problem which is more fit to real world applications. Then we propose a dynamic programming solution to the multi-stage optimization problem.

%%%%%%%%%%%%%%%%%%%%%%%%%%%%%%%%%%%%%%%%%%%%%%%%%%%%

\section{Conclusion and Future work}
In this paper, we introduced the optimal multistage campaigning problem, which is a generalization of the activity shaping and influence maximization problems, and it allows for more elaborate goal functions. 
Our model of social activity is based on multivariate Hawkes process, and for the first time, we manage to derive a linear connection between a time-varying exogenous intensity  
(\ie, the part that can be easily manipulated via incentives) and the overall network exposure of the campaign. 
The multistage optimal control problem is introduced and an approximate closed-loop dynamic programming approach is proposed to find the optimal interventions.
This linear connection between exogenous intensity and campaign's exposure enables developing a convex optimization framework for exposure shaping, deriving the necessary incentives to reach a global exposure pattern in the network. The method is evaluated 
on both synthetic and real-world held-out data and is shown to outperform several heuristics.

Experiments on synthetic and real world datasets reveal a couple of interesting facts: 
\begin{itemize}
\item
Most notable lesson is the presence of the so-called \emph{value of information}. We have witnessed, both in synthetic and real dataset, it is possible to achieve lower cost, essentially by taking advantage of extra information.
If the information was not available the controller couldn't adapt appropriately to the unexpected behavior and consequently the cost could have been adversely affected. 
\item
What we have empirically observed is that the performance, measured in achieving the lower cost and accurate prediction, improves with increasing  the number of intervention points. The more control over social network the better one can steer the campaign towards a goal.
\item
The performance slightly decreases with increasing the number of nodes. That might be due to the increased dimensionality of the optimization problem.

\end{itemize}
We acknowledge that our method has indeed limitations. For the networks at the scale of web or large social networks faster and scalable methods need to be explored and developed which remains as future works.
There are many other interesting venues for future work too. For example, considering competing/collaborating campaigns and their equilibria and interactions, a continuous-time intervention scheme, and exploring other approximate dynamic programming approaches remain as future work.

\bibliographystyle{unsrt}

%{
%\bibliographystyle{unsrt}
%\bibliography{campaign}

\begin{thebibliography}{10}

\bibitem{west2013air}
Darrell~M West.
\newblock {\em Air Wars: Television Advertising and Social Media in Election
  Campaigns, 1952-2012: Television Advertising and Social Media in Election
  Campaigns, 1952-2012}.
\newblock Sage, 2013.

\bibitem{vergeer2013online}
Maurice Vergeer, Liesbeth Hermans, and Steven Sams.
\newblock Online social networks and micro-blogging in political campaigning
  the exploration of a new campaign tool and a new campaign style.
\newblock {\em Party Politics}, 19(3):477--501, 2013.

\bibitem{bertsekas1995dynamic}
Dimitri~P Bertsekas.
\newblock {\em Dynamic programming and optimal control}, volume~1.

\bibitem{DalVer2007}
Daryl~J Daley and David Vere-Jones.
\newblock {\em An introduction to the theory of point processes}.
\newblock Springer Science \& Business Media, 2007.

\bibitem{ZhoZhaSon13}
Ke~Zhou, Hongyuan Zha, and Le~Song.
\newblock Learning social infectivity in sparse low-rank networks using
  multi-dimensional hawkes processes.
\newblock In {\em Artificial Intelligence and Statistics (AISTATS)}, 2013.

\bibitem{AalBorGje08}
Odd Aalen, Ornulf Borgan, and Hakon Gjessing.
\newblock {\em Survival and event history analysis: a process point of view}.
\newblock Springer, 2008.

\bibitem{Hawkes71}
Alan~G Hawkes.
\newblock Spectra of some self-exciting and mutually exciting point processes.
\newblock {\em Biometrika}, 58(1):83--90, 1971.

\bibitem{farajtabar2014activity}
M.~Farajtabar, N.~Du, M.~Gomez-Rodriguez, I.~Valera, L.~Song, and H.~Zha.
\newblock Shaping social activity by incentivizing users.
\newblock In {\em Neural Information Processinng Systems Conference 2014}, NIPS
  '14, 2014.

\bibitem{al2011computing}
Awad~H Al-Mohy and Nicholas~J Higham.
\newblock Computing the action of the matrix exponential, with an application
  to exponential integrators.
\newblock {\em SIAM journal on scientific computing}, 33(2):488--511, 2011.

\bibitem{leskovec2009meme}
Jure Leskovec, Lars Backstrom, and Jon Kleinberg.
\newblock Meme-tracking and the dynamics of the news cycle.
\newblock In {\em Proceedings of the 15th ACM SIGKDD international conference
  on Knowledge discovery and data mining}, pages 497--506. ACM, 2009.

\bibitem{YanZha13}
Shuang-Hong Yang and Hongyuan Zha.
\newblock Mixture of mutually exciting processes for viral diffusion.
\newblock In {\em Proceedings of the 30th International Conference on Machine
  Learning (ICML-13)}, pages 1--9, 2013.

\bibitem{richardson2002mining}
Matthew Richardson and Pedro Domingos.
\newblock Mining knowledge-sharing sites for viral marketing.
\newblock In {\em Proceedings of the eighth ACM SIGKDD international conference
  on Knowledge discovery and data mining}, pages 61--70. ACM, 2002.

\bibitem{kempe2003maximizing}
David Kempe, Jon Kleinberg, and {\'E}va Tardos.
\newblock Maximizing the spread of influence through a social network.
\newblock In {\em Proceedings of the ninth ACM SIGKDD international conference
  on Knowledge discovery and data mining}, pages 137--146. ACM, 2003.

\bibitem{hanson2007applied}
Floyd~B Hanson.
\newblock {\em Applied stochastic processes and control for Jump-diffusions:
  modeling, analysis, and computation}, volume~13.
\newblock Siam, 2007.

\bibitem{hu2008partial}
Yaozhong Hu and Bernt Oksendal.
\newblock Partial information linear quadratic control for jump diffusions.
\newblock {\em SIAM Journal on Control and Optimization}, 47(4):1744--1761,
  2008.

\bibitem{de2015modeling}
Abir De, Isabel Valera, Niloy Ganguly, Sourangshu Bhattacharya, and
  Manuel~Gomez Rodriguez.
\newblock Modeling opinion dynamics in diffusion networks.
\newblock {\em arXiv preprint arXiv:1506.05474}, 2015.

\bibitem{wang2016steering}
Yichen Wang, Evangelos Theodorou, Apurv Verma, and Le~Song.
\newblock Steering opinion dynamics in information diffusion networks.
\newblock {\em arXiv preprint arXiv:1603.09021}, 2016.

\bibitem{bloembergen2014influencing}
Daan Bloembergen, Bijan~Ranjbar Sahraei, Haitham Bou-Ammar, Karl Tuyls, and
  Gerhard Weiss.
\newblock Influencing social networks: An optimal control study.
\newblock In {\em ECAI}, pages 105--110, 2014.

\bibitem{kandhway2015campaigning}
Kundan Kandhway and Joy Kuri.
\newblock Campaigning in heterogeneous social networks: Optimal control of si
  information epidemics.
\newblock 2015.

\bibitem{chen2014optimal}
Pin-Yu Chen, Shin-Ming Cheng, and Kwang-Cheng Chen.
\newblock Optimal control of epidemic information dissemination over networks.
\newblock {\em Cybernetics, IEEE Transactions on}, 44(12):2316--2328, 2014.

\bibitem{karnik2012optimal}
Aditya Karnik and Pankaj Dayama.
\newblock Optimal control of information epidemics.
\newblock In {\em Communication Systems and Networks (COMSNETS), 2012 Fourth
  International Conference on}, pages 1--7, 2012.

\bibitem{lian2015multitask}
Wenzhao Lian, Ricardo Henao, Vinayak Rao, Joseph Lucas, and Lawrence Carin.
\newblock A multitask point process predictive model.
\newblock In {\em Proceedings of the 32nd International Conference on Machine
  Learning (ICML-15)}, pages 2030--2038, 2015.

\bibitem{parikh2012conjoint}
Ankur~P Parikh, Asela Gunawardana, and Christopher Meek.
\newblock Conjoint modeling of temporal dependencies in event streams.
\newblock In {\em UAI Bayesian Modelling Applications Workshop}. Citeseer,
  2012.

\bibitem{hall2014tracking}
Eric~C Hall and Rebecca~M Willett.
\newblock Tracking dynamic point processes on networks.
\newblock {\em arXiv preprint arXiv:1409.0031}, 2014.

\bibitem{FarGomDuZamZhaSon15}
Mehrdad Farajtabar, Manuel Gomez-Rodriguez, Nan Du, Mohammad Zamani, Hongyuan
  Zha, and Le~Song.
\newblock Back to the past: Source identification in diffusion networks from
  partially observed cascades.
\newblock In {\em Proceedings of the 18th International Conference on
  Artificial Intelligence and Statistics (AISTATS)}, 2015.

\bibitem{PerWol13}
Patrick~O Perry and Patrick~J Wolfe.
\newblock Point process modeling for directed interaction networks.
\newblock {\em Journal of the Royal Statistical Society: Series B (Statistical
  Methodology)}, 75(5):821--849, 2013.

\bibitem{LinAdaRya14}
Scott~W Linderman and Ryan~P Adams.
\newblock Discovering latent network structure in point process data.
\newblock In {\em International Conference on Machine Learning (ICML)}, 2014.

\bibitem{IwaShaGha13}
Tomoharu Iwata, Amar Shah, and Zoubin Ghahramani.
\newblock Discovering latent influence in online social activities via shared
  cascade poisson processes.
\newblock In {\em Proceedings of the 19th ACM SIGKDD international conference
  on Knowledge discovery and data mining}, pages 266--274. ACM, 2013.

\bibitem{BluBecHelKat12}
Charles Blundell, Jeff Beck, and Katherine~A Heller.
\newblock Modelling reciprocating relationships with hawkes processes.
\newblock In {\em NIPS}, 2012.

\bibitem{FarWanGomLiZhaSon15}
Mehrdad Farajtabar, Yichen Wang, Manuel Gomez-Rodriguez, Suang Li, Hongyuan
  Zha, and Le~Song.
\newblock Coevolve: A joint point process model for information diffusion and
  network co-evolution.
\newblock In {\em NIPS '15: Advances in Neural Information Processing Systems},
  2015.

\bibitem{TraFarSonZha15}
Long Tran, Mehrdad Farajtabar, Le~Song, and Hongyuan Zha.
\newblock Netcodec: Community detection from individual activities.
\newblock SIAM.

\bibitem{chow2015influence}
Shui-Nee Chow, Xiaojing Ye, Hongyuan Zha, and Haomin Zhou.
\newblock Influence prediction for continuous-time information propagation on
  networks.
\newblock {\em arXiv preprint arXiv:1512.05417}, 2015.

\bibitem{xu2016learning}
Hongteng Xu, Mehrdad Farajtabar, and Hongyuan Zha.
\newblock Learning granger causality for hawkes processes.
\newblock {\em arXiv preprint arXiv:1602.04511}, 2016.

\bibitem{he2015hawkestopic}
Xinran He, Theodoros Rekatsinas, James Foulds, Lise Getoor, and Yan Liu.
\newblock Hawkestopic: A joint model for network inference and topic modeling
  from text-based cascades.
\newblock In {\em Proceedings of the 32nd International Conference on Machine
  Learning (ICML-15)}, pages 871--880, 2015.

\bibitem{competing15icdm}
I.~Valera and M.~Gomez-Rodriguez.
\newblock Modeling adoption and usage of competing products.
\newblock In {\em 2015 IEEE International Conference on Data Mining}, 2015.

\bibitem{karimi2016smart}
Mohammad~Reza Karimi, Erfan Tavakoli, Mehrdad Farajtabar, Le~Song, and Manuel
  Gomez-Rodriguez.
\newblock Smart broadcasting: Do you want to be seen?
\newblock {\em arXiv preprint arXiv:1605.06855}, 2016.

\bibitem{Hijab2007}
Omar Hijab.
\newblock {\em Introduction to calculus and classical analysis}.
\newblock Springer, 2007.

\bibitem{Folland2013}
Gerald~B Folland.
\newblock {\em Real analysis: modern techniques and their applications}.
\newblock John Wiley \& Sons, 2013.

\bibitem{bracewell1965fourier}
Ron Bracewell.
\newblock The fourier transform and iis applications.
\newblock {\em New York}, 5, 1965.

\end{thebibliography}
%}

%%%%%%%%%%%%%%%%%%%%%%%%%%%%%%%%%%%%%%%%%%%%%%%%%%%%
%%%%%%%%%%%%%%%%%%%%%%%%%%%%%%%%%%%%%%%%%%%%%%%%%%%%
%%%%%%%%%%%%%%%%%%%%%%%%%%%%%%%%%%%%%%%%%%%%%%%%%%%%
%%%%%%%%%%%%%%%%%%%%%%%%%%%%%%%%%%%%%%%%%%%%%%%%%%%%
%%%%%%%			     		 	APPENDIX 					%%%%%%%

\clearpage
\newpage

\appendix

%%%%%%%%%%%%%%%%%%%%%%%%%%%%%%%%%%%%%%%%%%%%%%%%%%%%
\section{Proofs}
\label{appen-proofs}
\setcounter{theorem}{0}
\begin{lemma}
Suppose $\Psi:[0,T] \to \mathbb{R}^{n\times n}$ is a matrix function, then
for every fixed constant intensity $\mu(t) = c\in\mathbb{R}_{+}^n$, $\eta_c(t):=\Psi(t)c$
solves the semi-infinite integral equation
\begin{equation}\label{eqn_app:WienerHopf_constant}
\eta(t) = c + \int_{0}^{t} \Phi(t-s)\eta(s)ds,\quad \forall t\in[0,T],
\end{equation}
if and only if $\Psi(t)$ satisfies 
\begin{equation}\label{eqn_app:Psi_gen}
\Psi(t) = I + \int_{0}^{t} \Phi(t-s)\Psi(s)ds,\quad \forall t\in[0,T].
\end{equation}
In particular, if $\Phi(t)=Ae^{-\omega t}\one_{\geq0}(t)=[a_{ij}e^{-\omega t}\one_{\geq0}(t)]_{ij}$
where $0\leq \omega\notin \text{Spectrum}(A)$, then 
 \begin{equation}\label{eqn_app:defPsi}
     \Psi(t) =  e^{(A-\omega I) t} + \omega(A-\omega I )^{-1} ( e^{(A-\omega I) t} - I)
\end{equation}
for $t\in[0,T]$, where, $\one_{\geq0}(t)$ is an indicator function for $t \geq 0$.
\end{lemma}

\begin{proof}
Suppose that $\Psi(t)c$ solves \eqref{eqn_app:WienerHopf_constant} for every $c$, then
substituting $\eta(t)$ by $\eta_c(t):=\Psi(t)c$ in \eqref{eqn_app:WienerHopf_constant} we obtain
$\left[\Psi(t)-I-\int_{0}^{t}\Phi(t-s)\Psi(s)ds\right]c=0$. Since $c\in\mathbb{R}_{+}^n$
is arbitrary, we know that $\Psi(t)-I-\int_{0}^{t}\Phi(t-s)\Psi(s)ds=0$ for all $t$,
and hence \eqref{eqn_app:Psi_gen} follows.
The converse is trivial to verify. Furthermore, one can readily check that
\eqref{eqn_app:defPsi} satisfies \eqref{eqn_app:Psi_gen} for $\Phi(t)=Ae^{-\omega t}\one_{\geq0}(t)$.
\end{proof}

\begin{theorem}
\label{theo:piecewise_constant_average}
Let $\Psi(t)$ satisfy \eqref{eqn_app:Psi_gen} and $\mu(t)$ be a right-continuous piecewise 
constant intensity function of form \eqref{eq:piecewise-constant-intensity}, then
the rate function $\eta(t)$ is given by
\begin{align}\label{eqn_app:eta_pwconst}
\eta(t) = \sum_{k=0}^m \Psi(t-\tau_k) (c_{k}-c_{k-1}),
\end{align}
for all $t\in (\tau_{m-1},\tau_m]$ and $m=1,\dots,M$, where $c_{-1}:=0$ by convention.
\end{theorem}

\begin{proof}
We prove this result by induction on partition size $M$. 
The previous lemma shows \eqref{eqn_app:eta_pwconst} for constant $\mu(t)=c\one_{[0,T]}(t)$
(i.e., $M=1$).
Suppose \eqref{eqn_app:eta_pwconst} is true for any given piecewise constant $\mu(t)$ of form 
\eqref{eq:piecewise-constant-intensity}
with $M$ partitions. If we impose a constant control $c\in\mathbb{R}_{+}^n$ (different from original
$c_{M}$) since time $\tau \in (\tau_{M-1},T]$, 
namely the piecewise constant intensity function is updated to $\hat{\mu}(t):=\mu(t)+(c-c_{M})\one_{(\tau,T]}(t)$,
then we need to show that the updated rate function $\hat{\eta}(t)$ is
\begin{align}\label{eqn_app:eta_hat_pwconst}
\hat{\eta}(t) = \eta(t) + \Psi(t-\tau)(c-c_{M})\one_{(\tau,T]}(t),
\end{align}
for all $t\in[0,T]$. This result can by verified easily for $t\in [0,\tau]$. If $t\in (\tau,T]$, then
$\hat{\mu}(t)=\mu(t)+(c-c_M)\one_{(\tau,T]}(t)=\mu(t)+(c-c_M)$ and
\begin{align}
& \hat{\mu}(t)+ \int_{0}^{t} \Phi(t-s)\hat{\eta}(s)ds \nonumber \\
=\ & \mu(t) + (c-c_{M}) + \int_{0}^{t} \Phi(t-s)[\eta(s)+\Psi(s-\tau)(c-c_M)\one_{(\tau,T]}(s)] ds \nonumber \\
=\ & \eta(t) + (c-c_{M}) + \int_{0}^{t} \Phi(t-s)\Psi(s-\tau)(c-c_M)\one_{(\tau,T]}(s)ds \\
=\ & \eta(t) + \left[ I + \int_{0}^{t-\tau} \Phi(t-\tau-u)\Psi(u)du \right](c-c_{M}) \nonumber \\
=\ & \eta(t) + \Psi(t-\tau)(c-c_{M}) = \hat{\eta}(t), \nonumber
\end{align}
where we used the fact that $\eta(t)$ is the rate function for intensity $\mu(t)$ to get
the second equality, applied change of variables $u=s-\tau$ to obtain the third equality,
and the property \eqref{eqn_app:Psi_gen} of $\Psi(t)$ to get the fourth equality.
This implies that the rate function is $\hat{\eta}(t)$ given in \eqref{eqn_app:eta_hat_pwconst}
for the updated piecewise constant intensity $\hat{\mu}(t)$ with $M+1$ partitions,
and hence completes the proof.
\end{proof}

\begin{theorem}
\label{theo:average_general}
If $\Psi\in C^{1}([0,T])$ and satisfies \eqref{eqn_app:Psi_gen}, and exogenous intensity $\mu$ 
is bounded and piecewise absolutely continuous
on $[0,T]$ where $\mu(t+)=\mu(t)$ at all discontinuous points $t$, then $\mu$ is differentiable almost everywhere,
and the semi-indefinite integral
\begin{equation}\label{eqn_app:WienerHopf_gen}
\eta(t) = \mu(t) + \int_{0}^{t} \Phi(t-s)\eta(s)ds,\quad \forall t\in[0,T],
\end{equation}
yields a rate function $\eta:[0,T]\to \mathbb{R}_{+}^n$ given by
\begin{align}\label{eqn_app:eta_gen}
\eta(t) = \int_{0}^{t}\Psi(t-s) d\mu(s).
\end{align}
\end{theorem}

\begin{proof}
It suffices to show \eqref{eqn_app:eta_gen} for absolutely continuous $\mu(t)$ on $[0,T)$
since extending the proof to piecewise absolutely continuous function is straightforward. 
We first define $\mu(T)=\mu(T-)$ and obtain a continuous $\mu(t)$ on $[0,T]$.
Since $[0,T]$ is compact, we know $\mu(t)$ is uniformly continuous, and hence
there exists a sequence of piecewise constant functions $\{\mu_k\}_{k=1}^{\infty}$
such that $\mu_k\to \mu$ uniformly on $[0,T]$, i.e., $\lim_{k\to\infty}\sup_{0\leq t\leq T}|\mu_k(t)-\mu(t)|=0$ 
\cite[Thm.  2.3.6]{Hijab2007} 
This also implies that $\{\mu_k\}$ is uniformly bounded.
For every $k$, piecewise constant function $\mu_k$ has bounded variation, therefore we have by 
\cite[Thm.  3.36]{Folland2013} that
\begin{equation}\label{eqn_app:eta_seq}
\eta_k(t):=\int_{0}^{t}\Psi(t-s)d\mu_k(s)=\int_{0}^{t}\Psi'(t-s)\mu_k(s)ds+\Psi(0)\mu_k(t)-\Psi(t)\mu_k(0),
\end{equation}
for all $t\in[0,T]$.
Since $\Psi\in C^1$ we know $\Psi'$ is continuous and bounded on $[0,T]$. 
By Lebesgue's bounded convergence theorem we know
\begin{equation}
\int_{0}^{t}\Psi'(t-s)\mu_k(s)ds \to \int_{0}^{t}\Psi'(t-s)\mu(s)ds.
\end{equation}
Furthermore, using the uniform convergence of $\{\mu_k\}$ to $\mu$, we
know the right hand side \eqref{eqn_app:eta_seq} converges to
$\int_{0}^{t}\Psi'(t-s)\mu(s)ds+\Psi(0)\mu(t)-\Psi(t)\mu(0)$. 
Then integration by parts for piecewise absolutely continuous function
$\mu$ which has bounded variation implies that 
$\eta(t) = \int_{0}^{t}\Psi(t-s) d\mu(s)$ for all $t\in[0,T]$.
\end{proof}

\begin{corollary}
Suppose $\Psi$ and $\mu$ satisfy the same conditions as in Thm.  \ref{theo:average_general},
and define $\psi=\Psi'$, then the rate function is $\eta(t)=(\psi * \mu)(t)$.
In particular, if $\Phi(t)=Ae^{-\omega t}\one_{\geq0}(t)=[a_{ij}e^{-\omega t}\one_{\geq0}(t)]_{ij}$
then the rate function $\eta(t)=A \int_{0}^{t} e^{(A-wI)(t-s)}\mu(s)ds$.
\end{corollary}

\begin{proof}
Note that both $\psi$ and $\mu$ have supports in $\mathbb{R}_{+}$, 
therefore integration by parts and the property of derivative of convolution 
\cite[p.~126]{bracewell1965fourier} imply that $\eta(t)=\int_{0}^{t}\Psi'(t-s)\mu(s)ds=(\psi * \mu)(t)$.
If $\Phi(t)=Ae^{-\omega t}\one_{\geq0}(t)=[a_{ij}e^{-\omega t}\one_{\geq0}(t)]_{ij}$,
then $\Psi(t)$ is given in \eqref{eqn_app:defPsi}, and hence $\psi(t)=Ae^{(A-wI)t}$
and the closed form of $\eta(t)$ follows as in the claim.
\end{proof}

%%%%%%%%%%%%%%%%%%%%%%%%%%%%%%%%%%%%%%%%%%%%%%%%%%%%

\begin{figure*}[!t]
  % \vspace{-3mm}
  \centering
  \setlength{\tabcolsep}{6pt}
  \begin{tabular}{cccc}
          \hspace{-5mm}
          \includegraphics[width=0.25\textwidth]{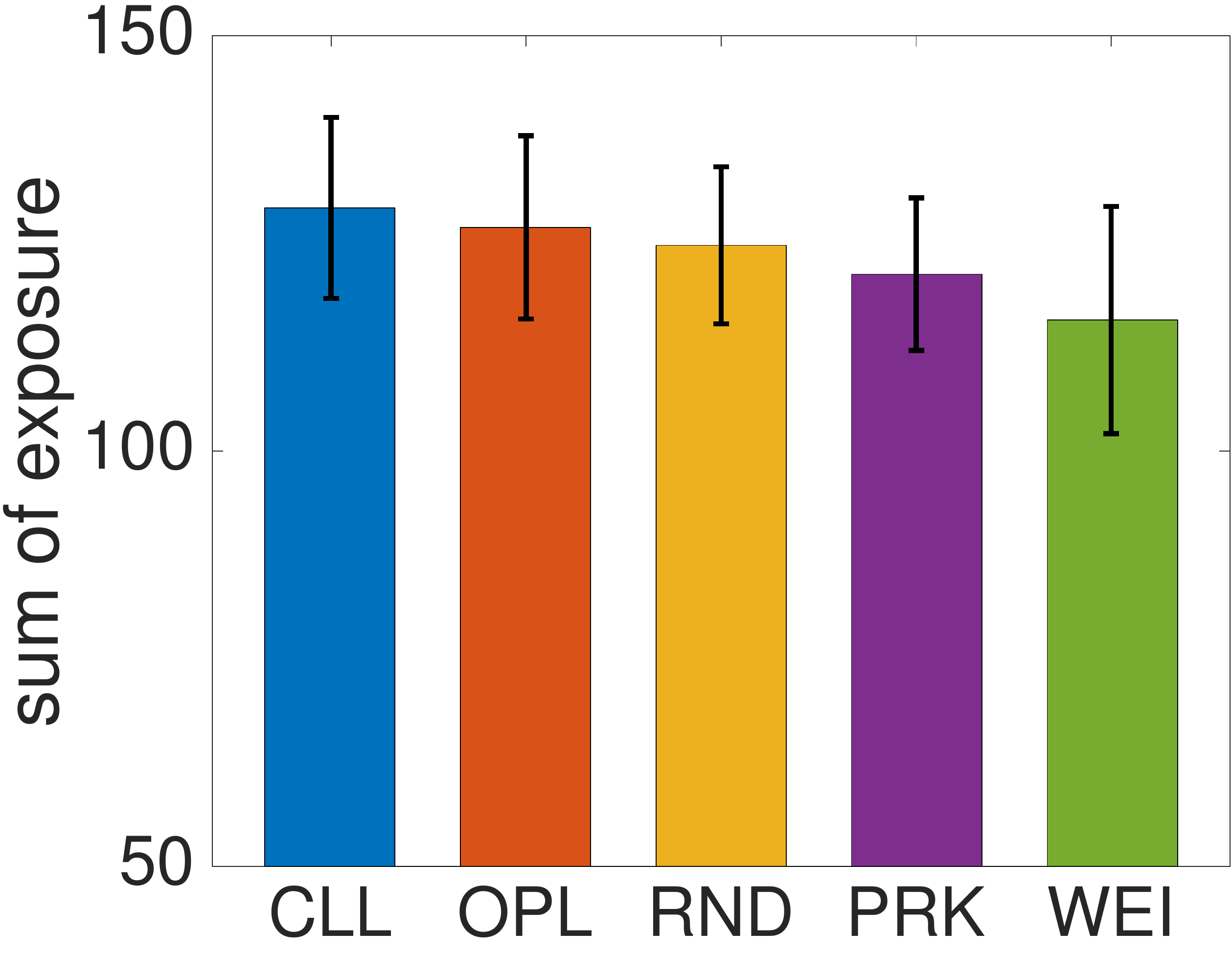} &
          \hspace{-5mm}
          \includegraphics[width=0.25\textwidth]{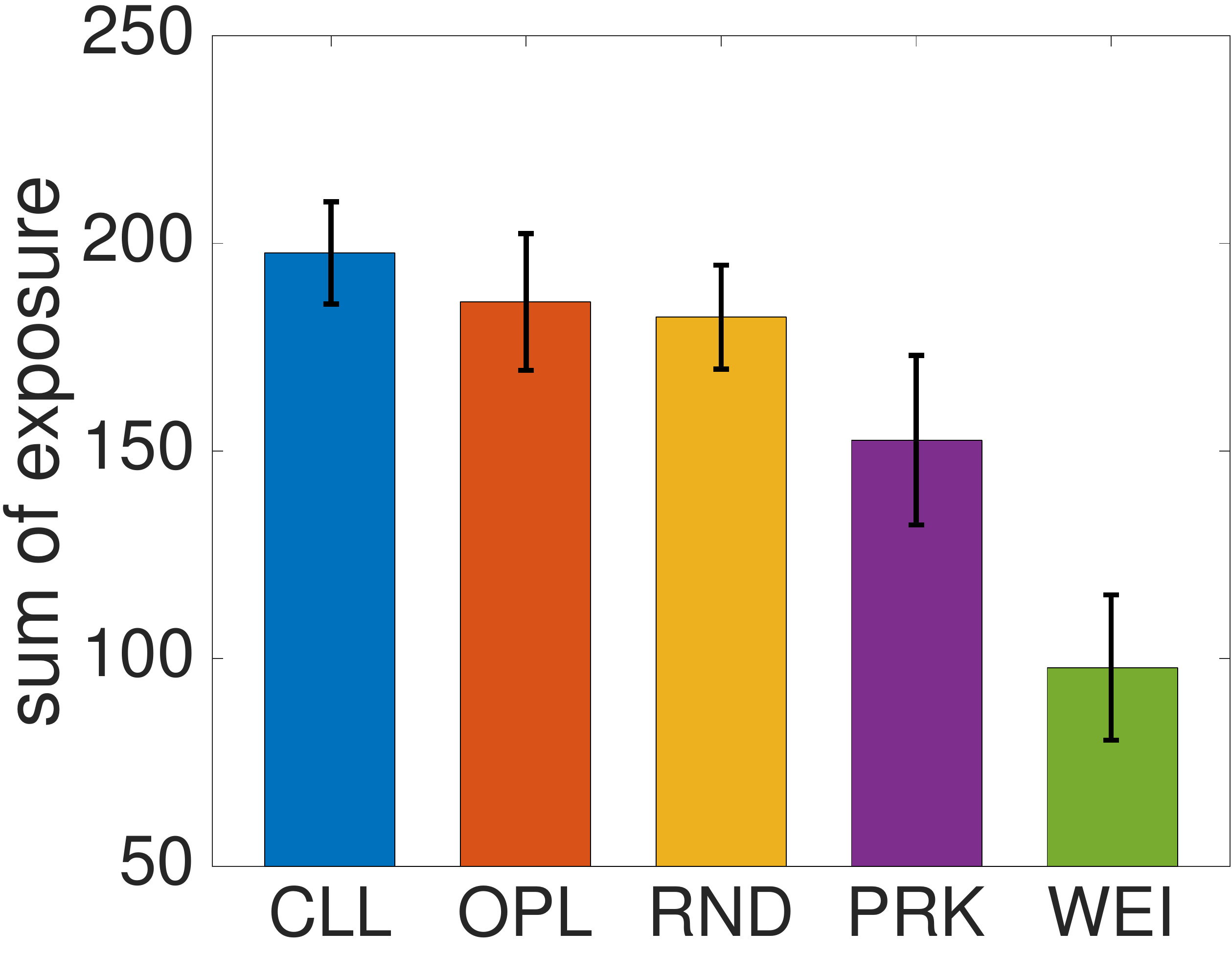} &
          \hspace{-5mm}
          \includegraphics[width=0.25\textwidth]{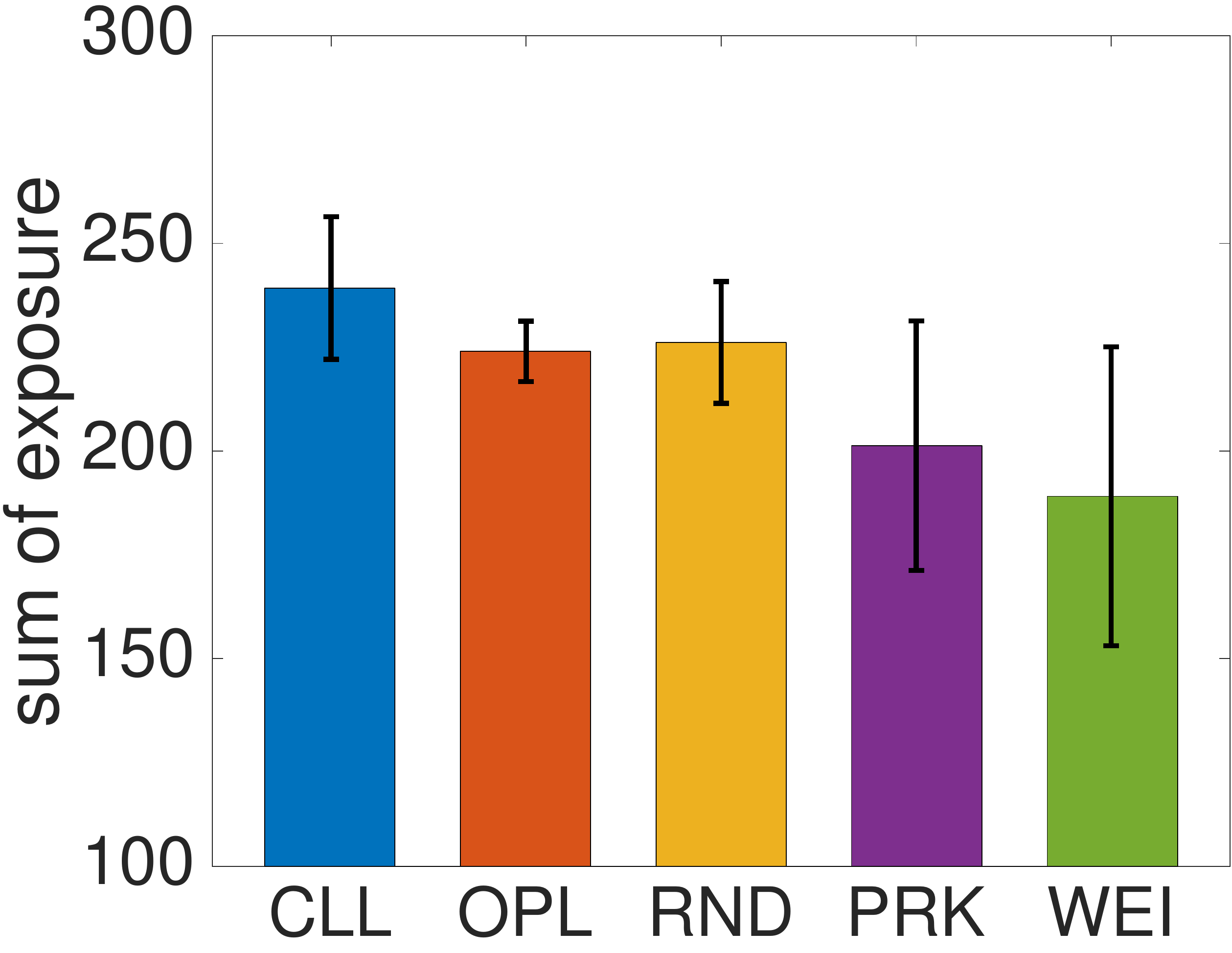} &         
          \hspace{-5mm}
          \includegraphics[width=0.25\textwidth]{fig-synth-CEM-size-6} \\
          a) $n=150$ & b) $n=200$ & c) $n=250 $ & d) $n=300$ \\
          \hspace{-5mm}
          \includegraphics[width=0.25\textwidth]{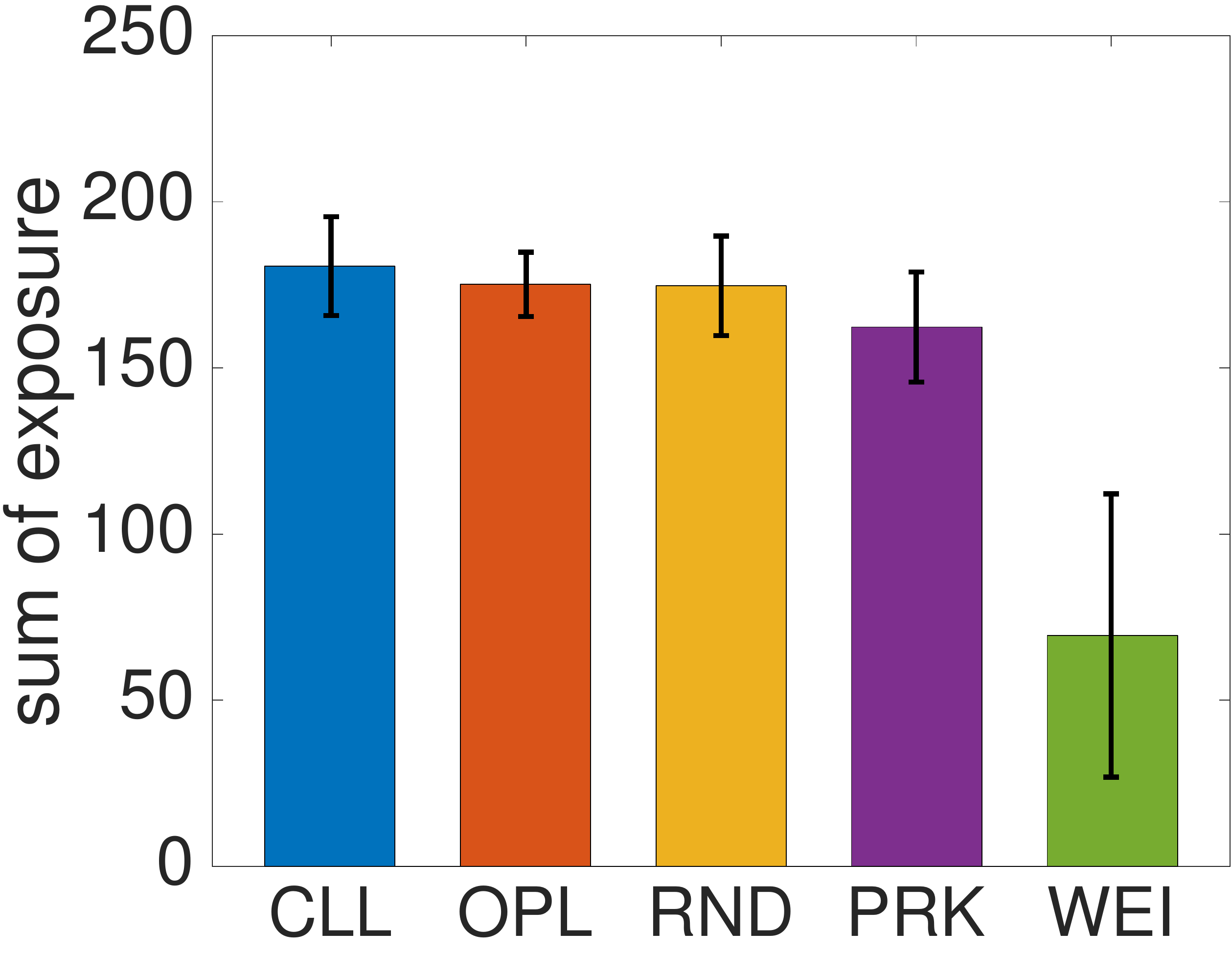} &
        \hspace{-5mm}
          \includegraphics[width=0.25\textwidth]{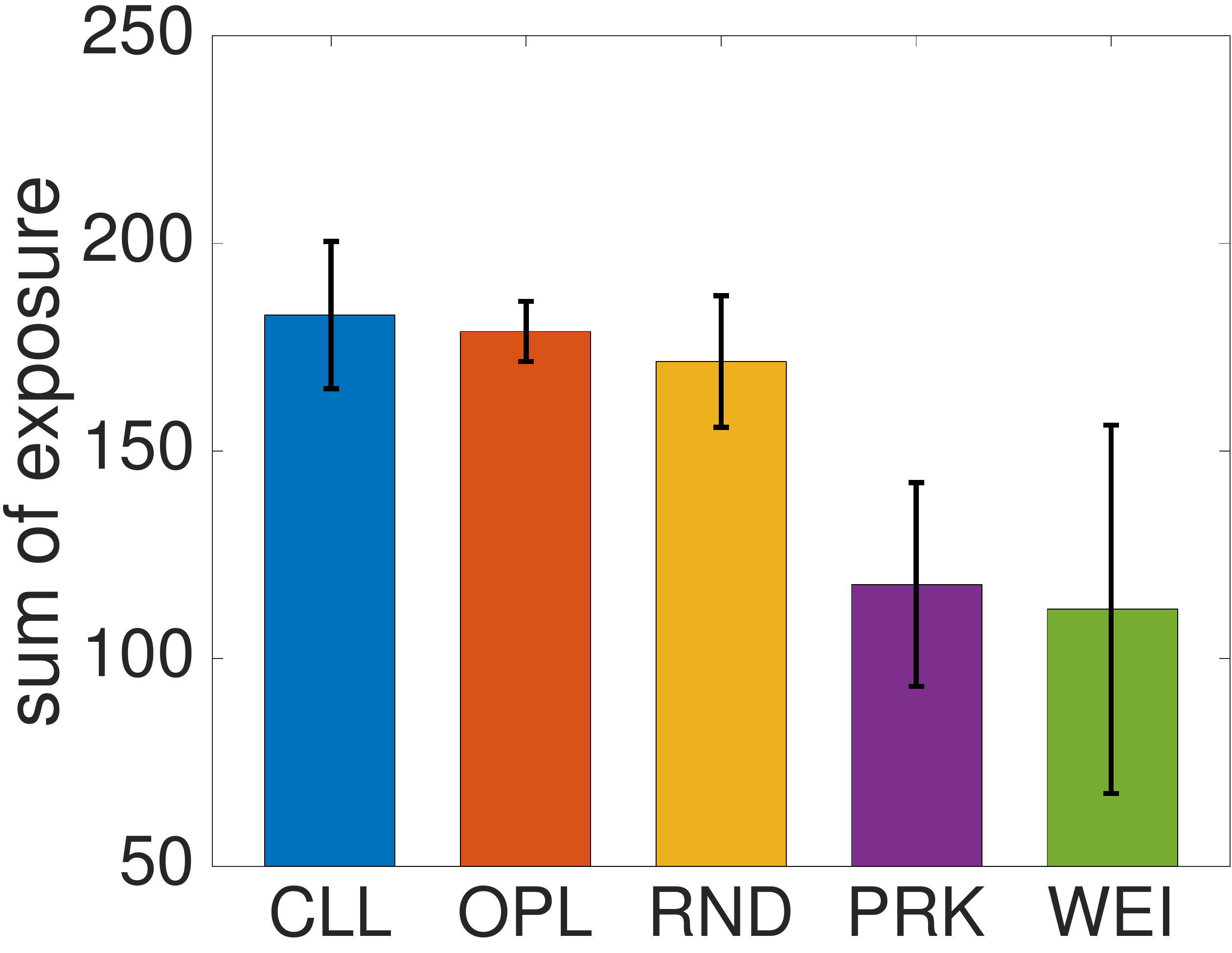} &
          \hspace{-5mm}
          \includegraphics[width=0.25\textwidth]{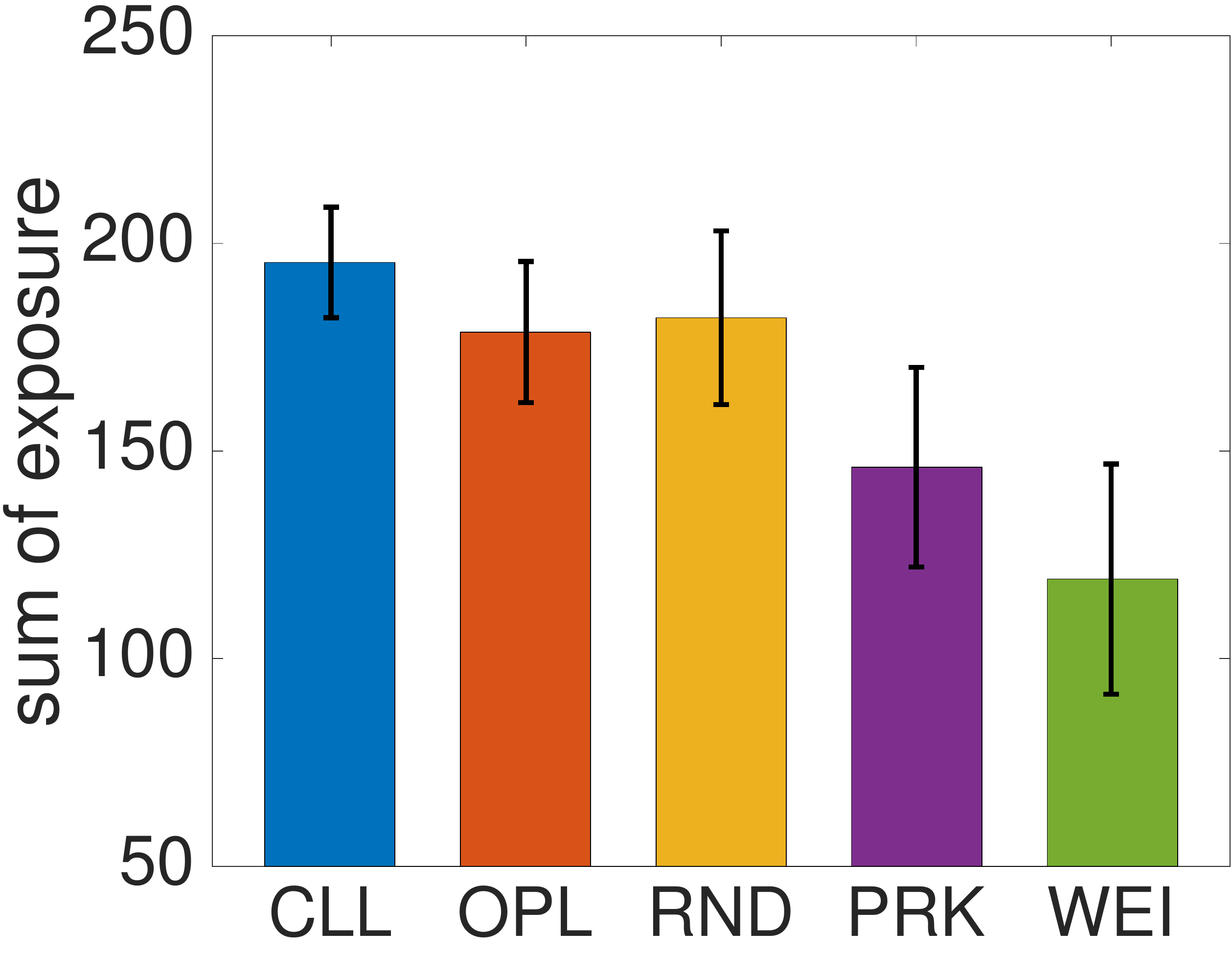} &
          \hspace{-5mm}
          \includegraphics[width=0.25\textwidth]{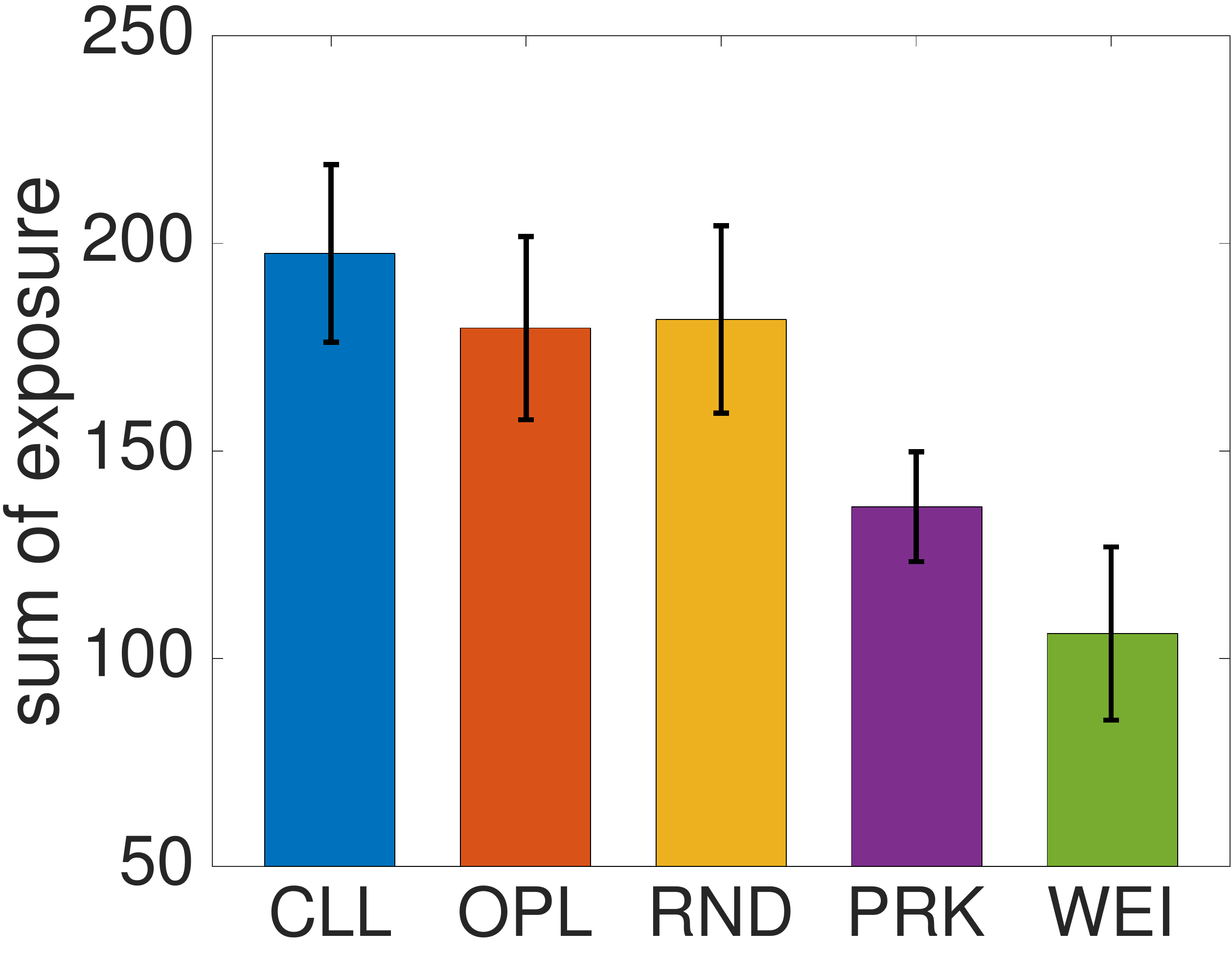}  \\
            e) $M=4$ & f) $M=6$ & g) $M=8$ & h) $M=10$ \\
          \hspace{-5mm}
          \includegraphics[width=0.25\textwidth]{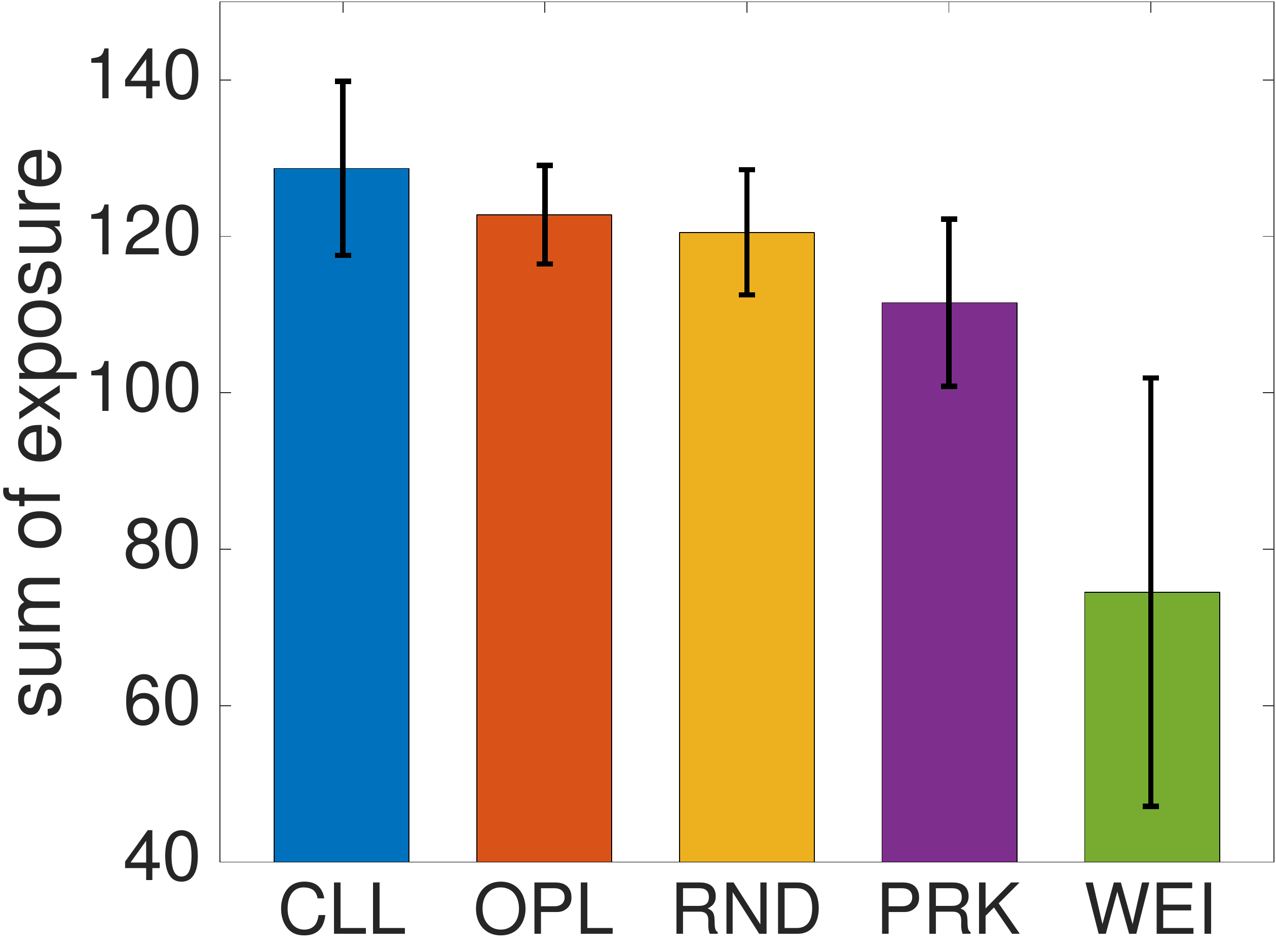} &
        \hspace{-5mm}
          \includegraphics[width=0.25\textwidth]{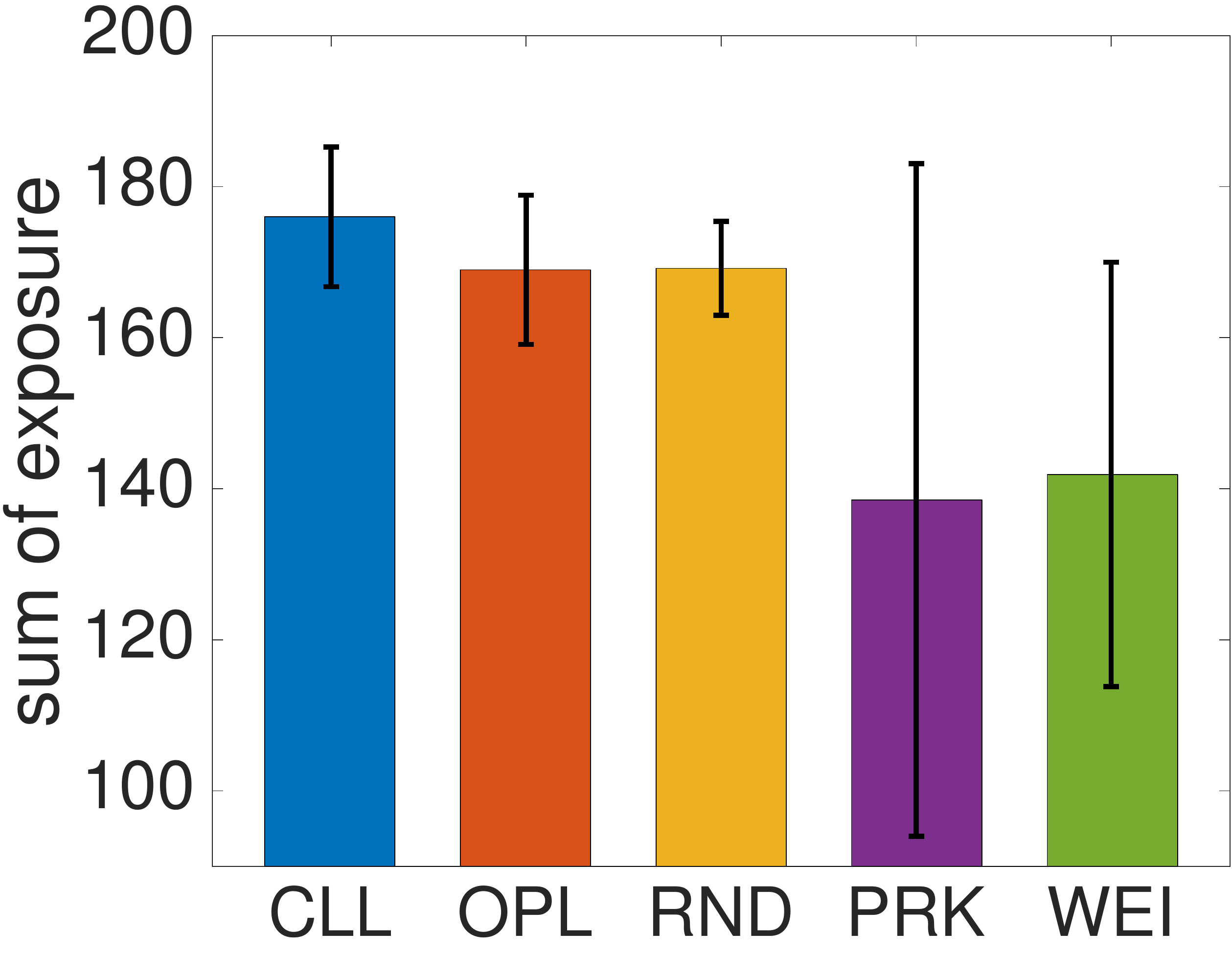} &
          \hspace{-5mm}
          \includegraphics[width=0.25\textwidth]{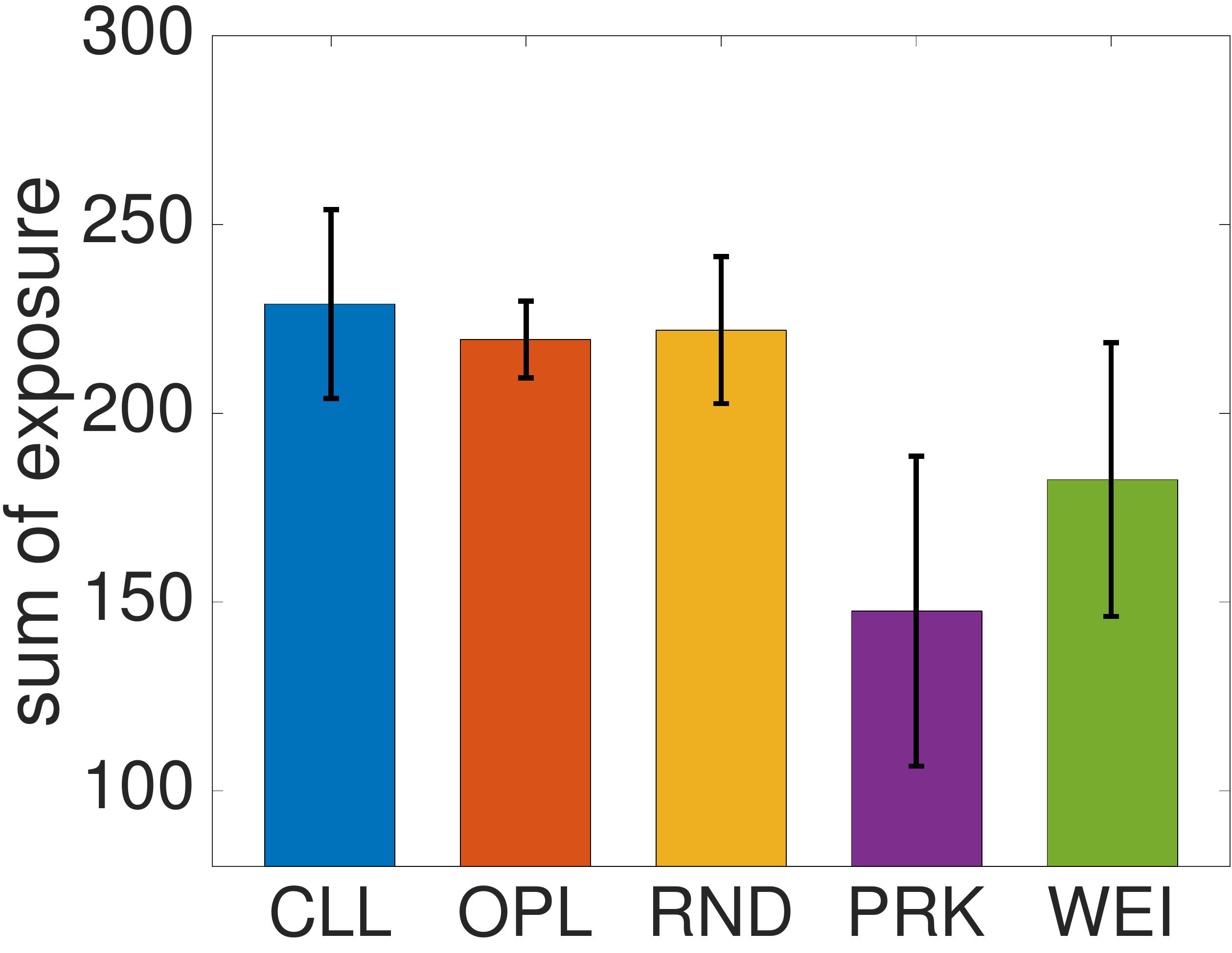} &
          \hspace{-5mm}
          \includegraphics[width=0.25\textwidth]{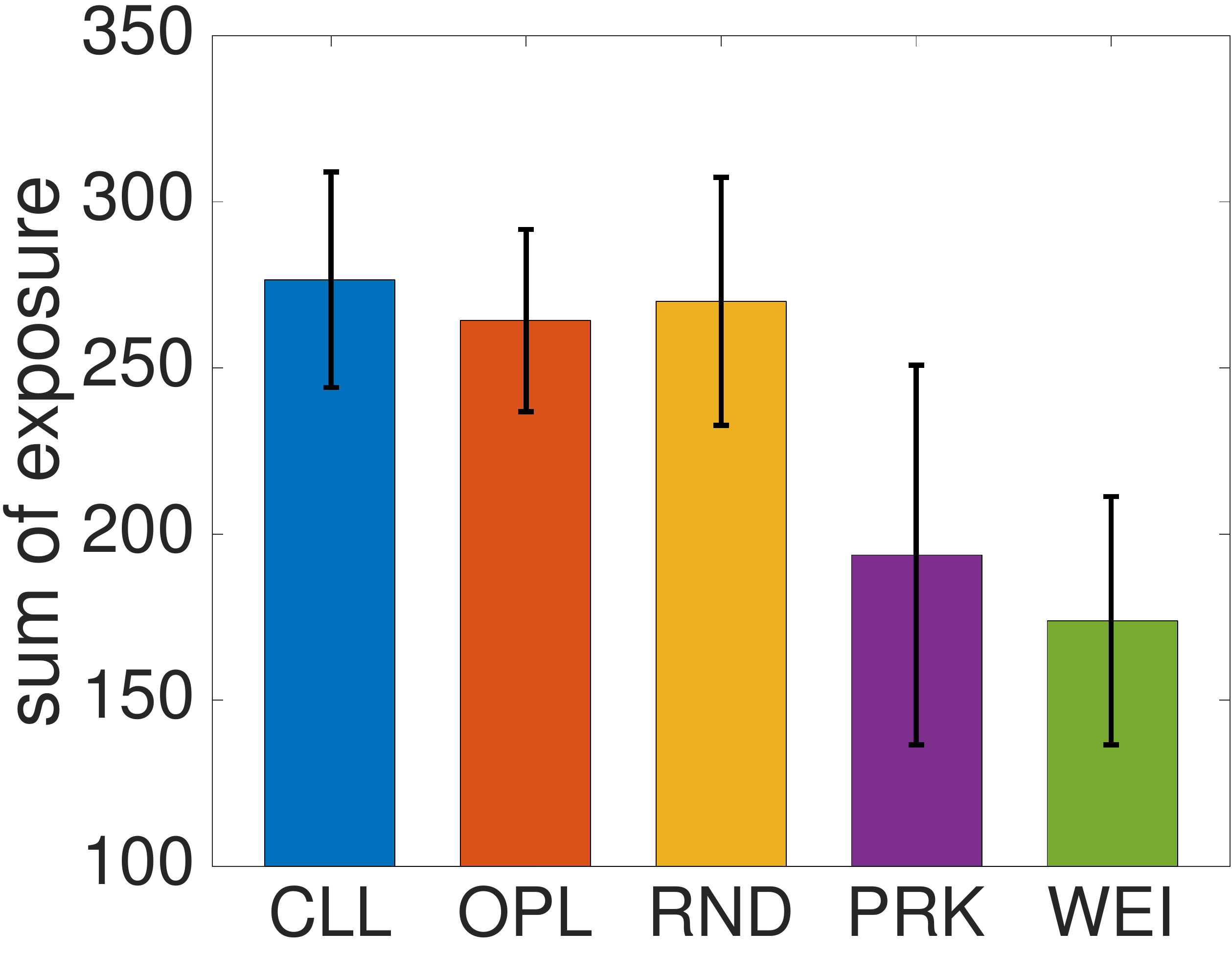}  \\
            i) $\Delta=5$ & j) $\Delta=6.66$ & k) $\Delta=8.33$ & l) $\Delta=10$\\
  \end{tabular}% \vspace{-3mm}
  \caption{Caped exposure maximization results on synthetic data; top row: $n$ varies, $M=6$, $T=40$; 
  middle row: $M$ varies, $T=40$, $n=200$; bottom row: $T$ varies, $n=200$, $M=6$}
  \label{fig:synth-cem-results}
\end{figure*}

\begin{figure*}[!t]
  % \vspace{-3mm}
  \centering
  \setlength{\tabcolsep}{6pt}
  \begin{tabular}{cccc}
          \hspace{-5mm}
          \includegraphics[width=0.25\textwidth]{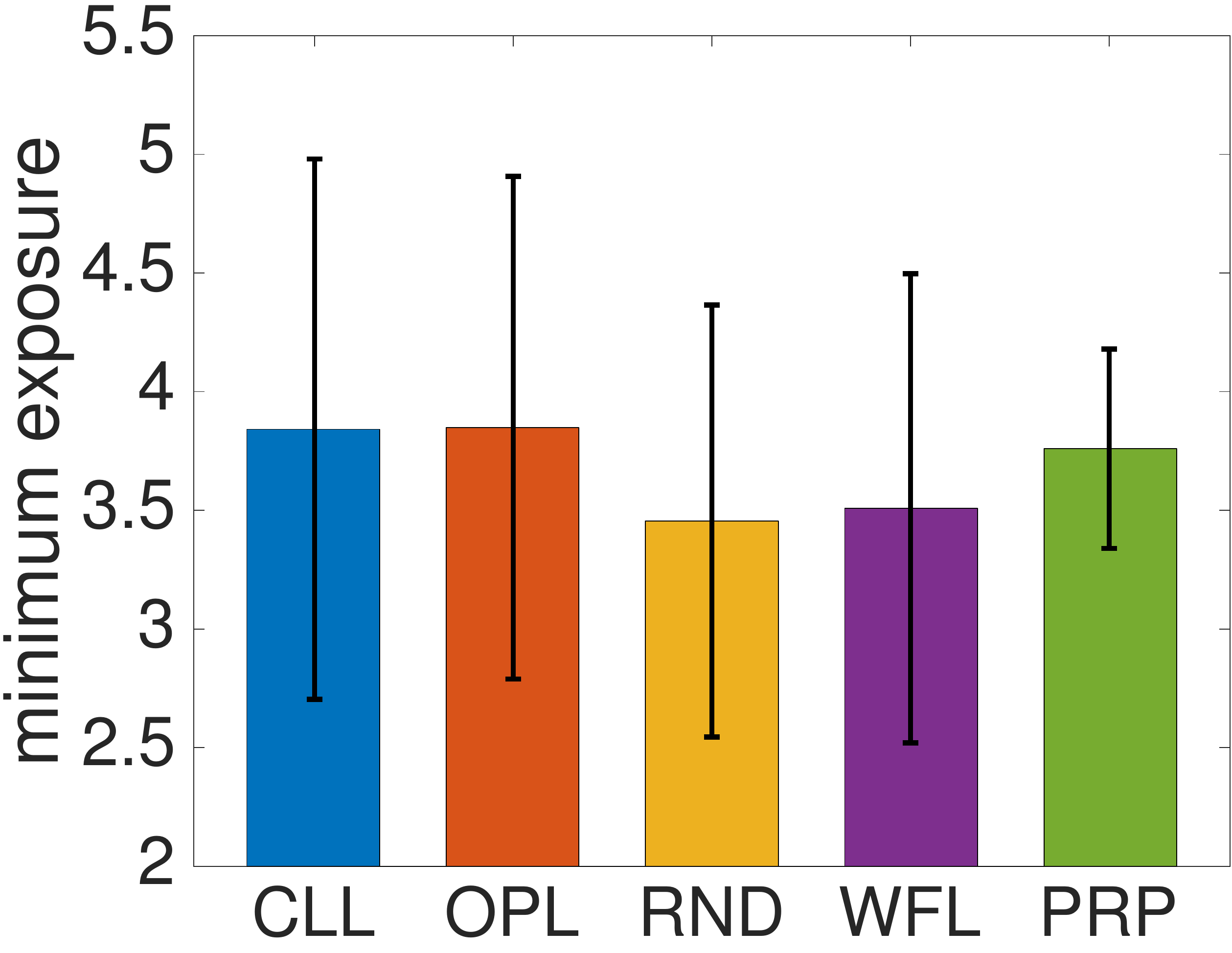} &
          \hspace{-5mm}
          \includegraphics[width=0.25\textwidth]{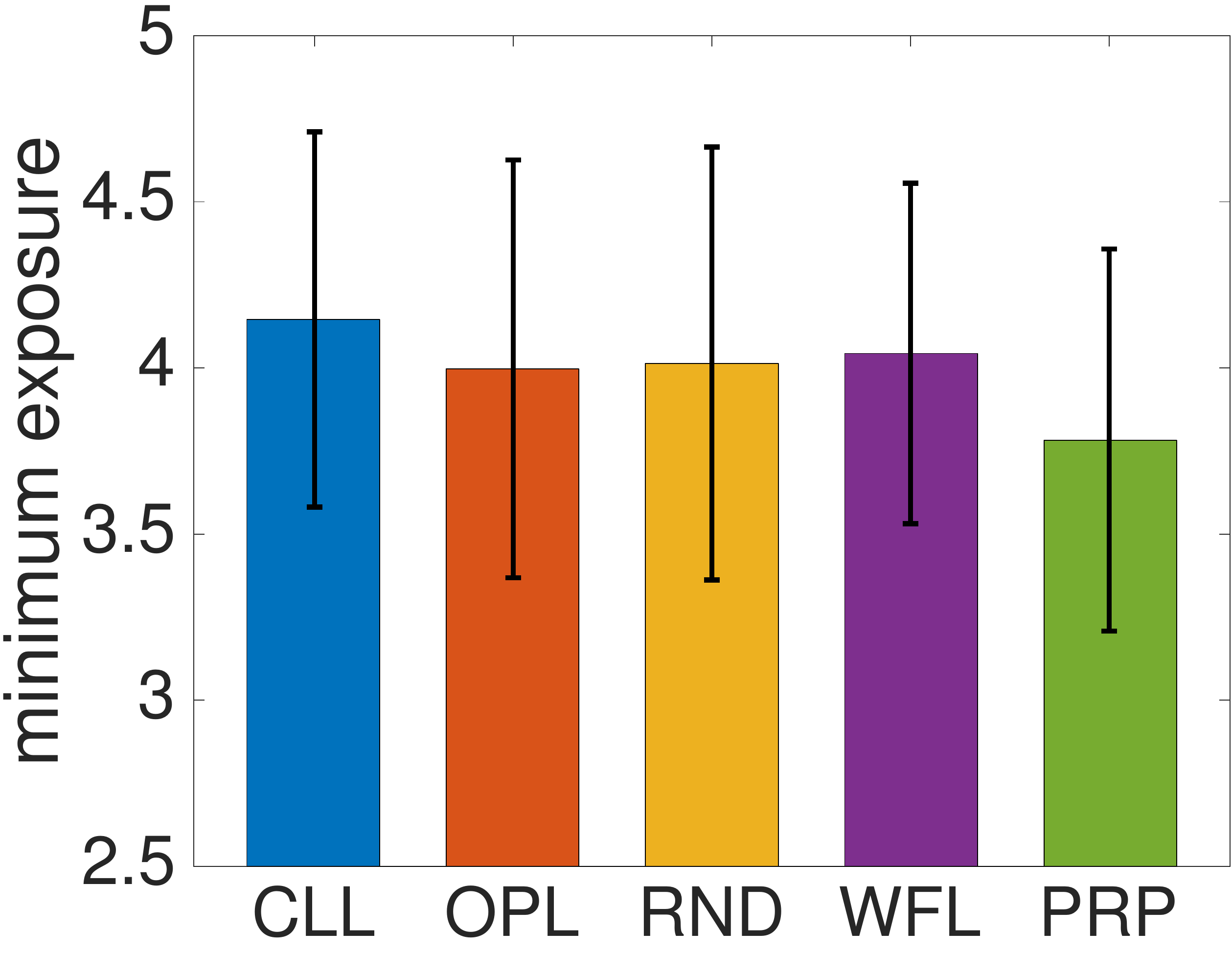} &
          \hspace{-5mm}
          \includegraphics[width=0.25\textwidth]{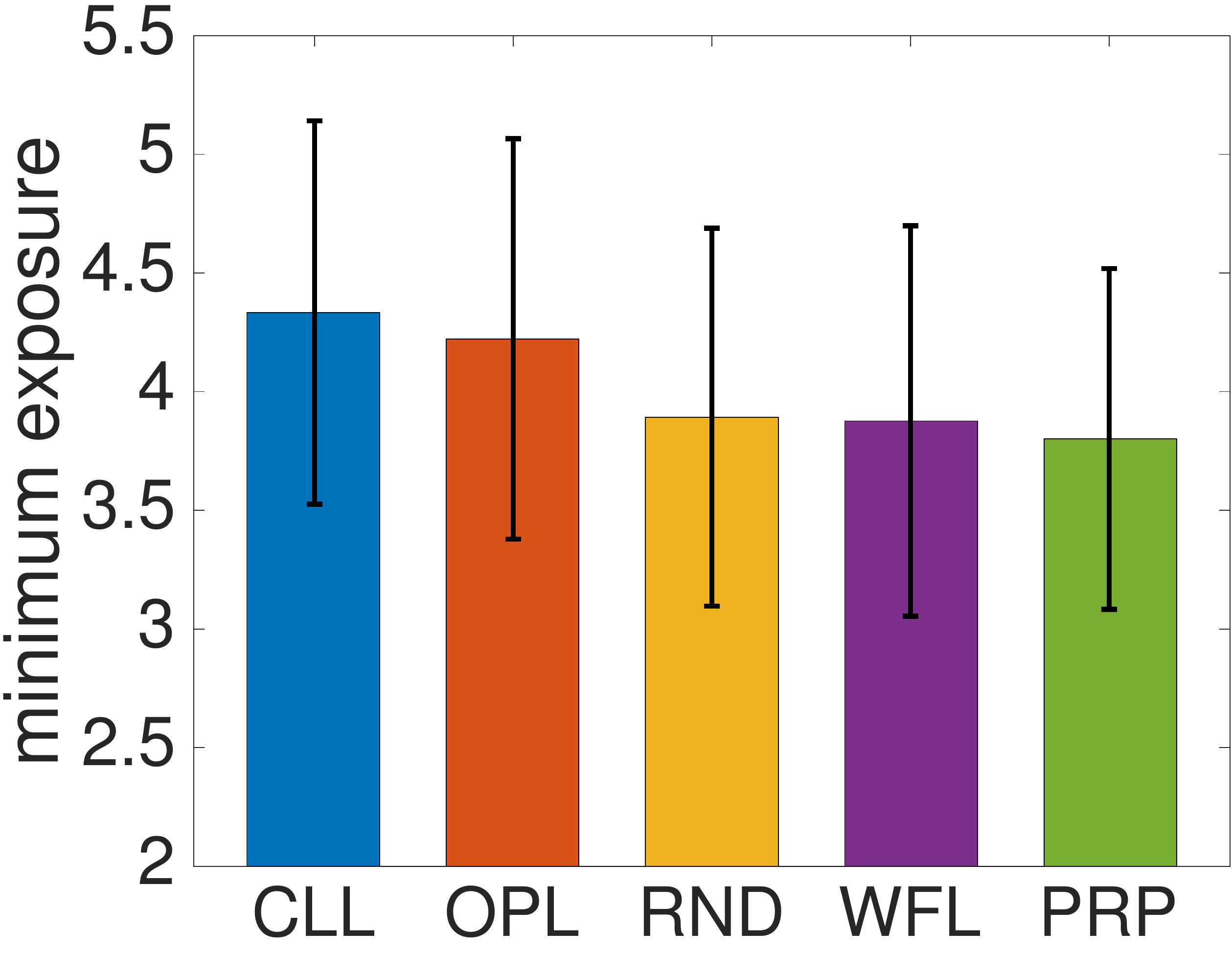} &         
          \hspace{-5mm}
          \includegraphics[width=0.25\textwidth]{fig-synth-MMESH-size-6} \\
          a) $n=150$ & b) $n=200$ & c) $n=250 $ & d) $n=300$ \\
          \hspace{-5mm}
          \includegraphics[width=0.25\textwidth]{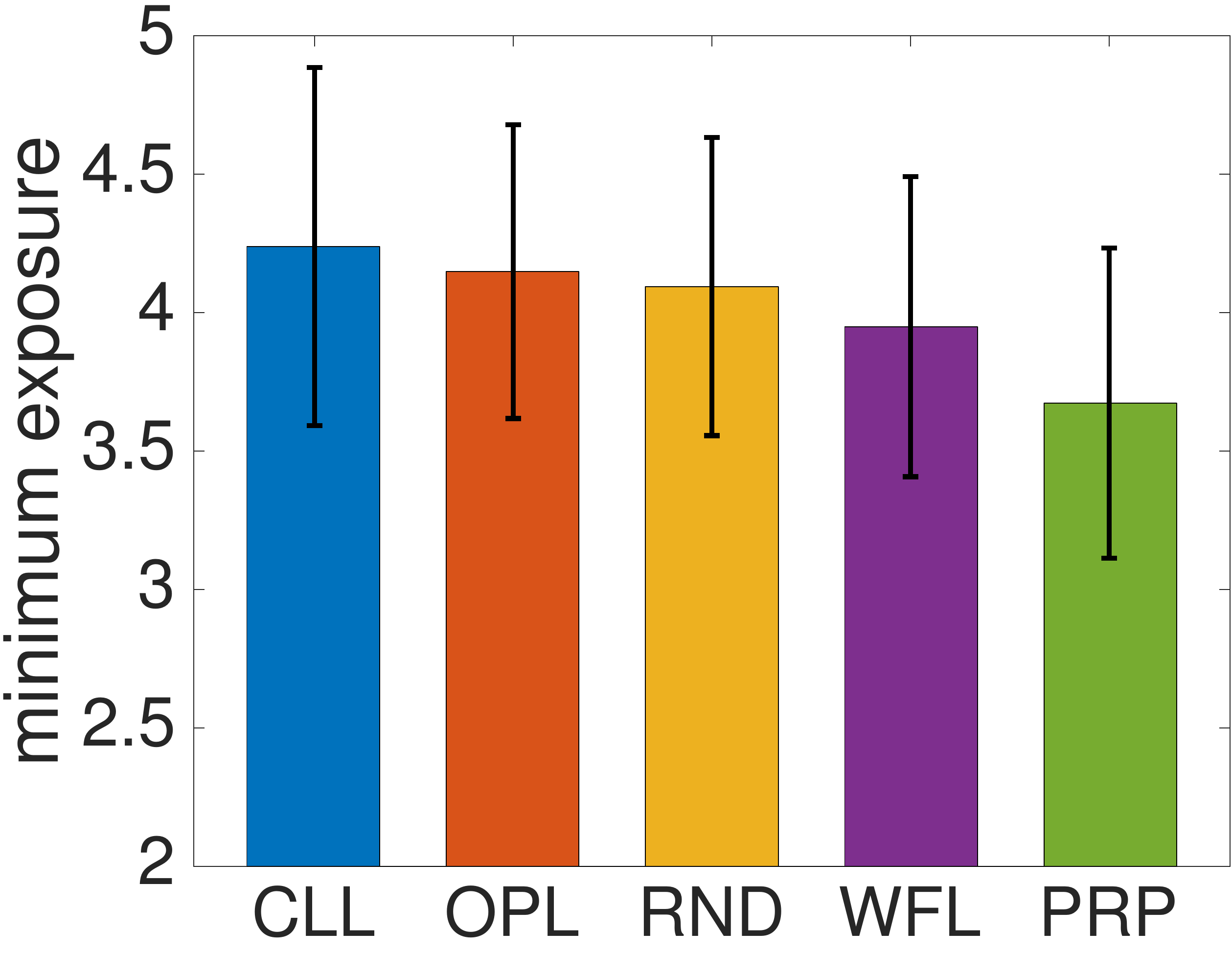} &
        \hspace{-5mm}
          \includegraphics[width=0.25\textwidth]{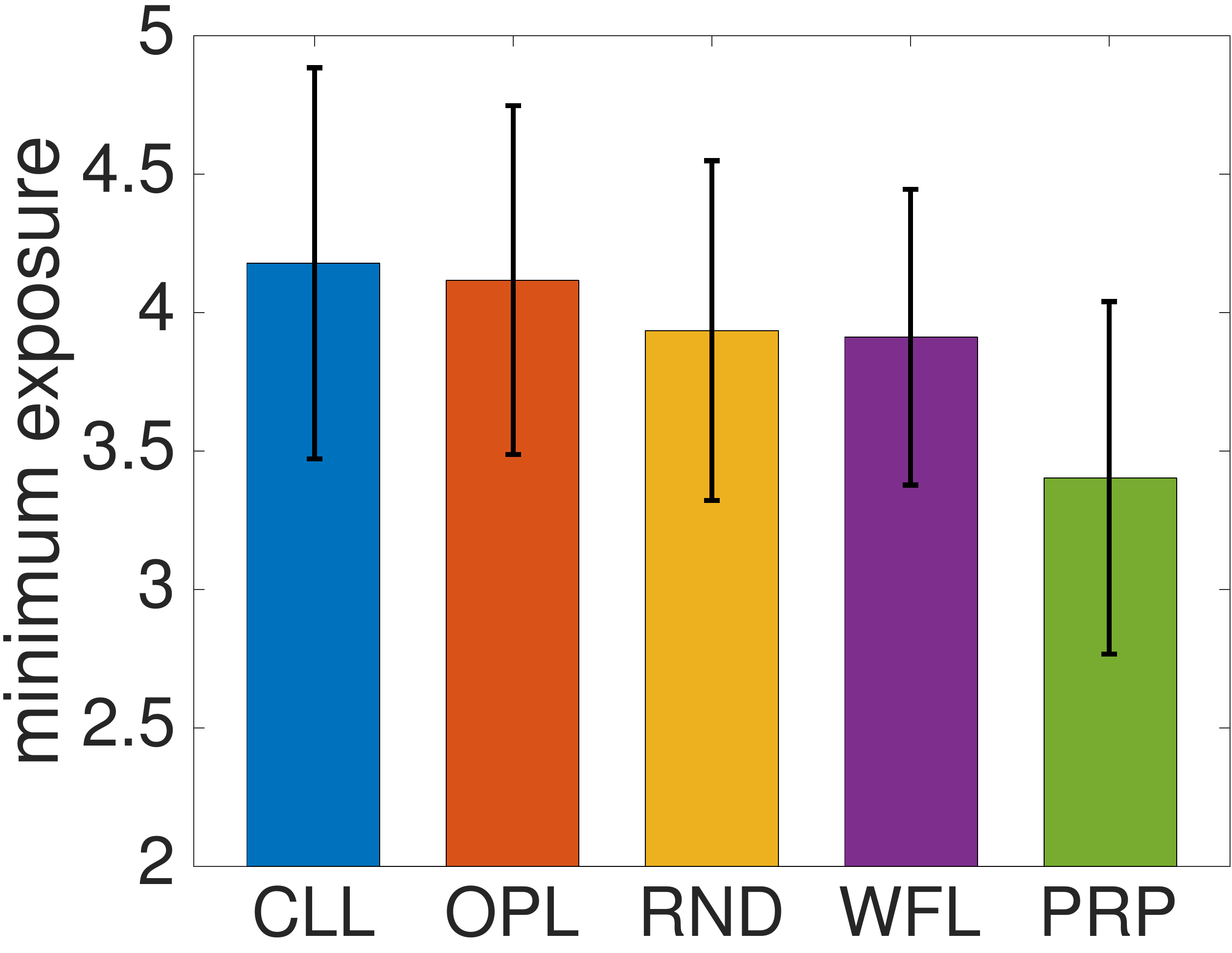} &
          \hspace{-5mm}
          \includegraphics[width=0.25\textwidth]{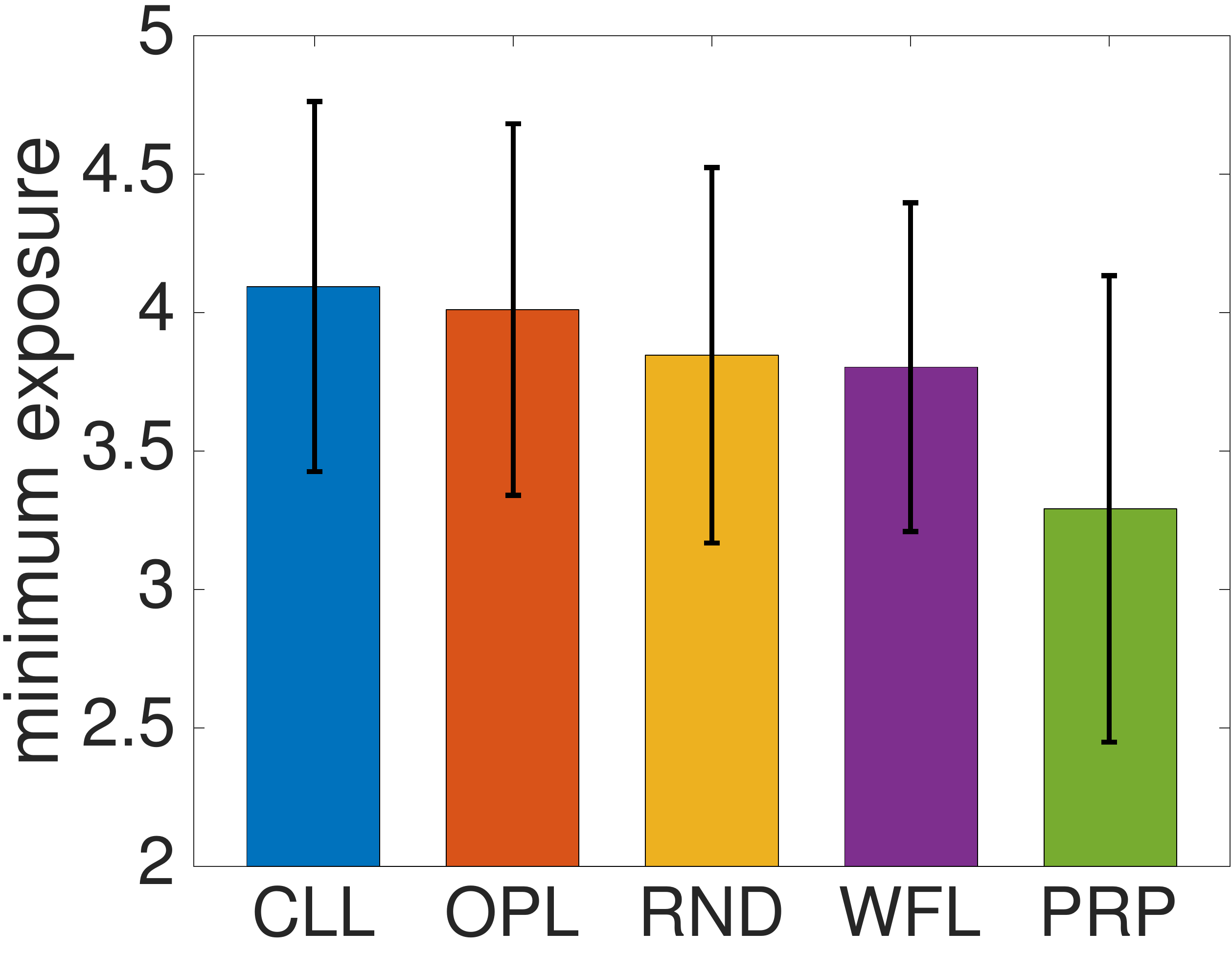} &
          \hspace{-5mm}
          \includegraphics[width=0.25\textwidth]{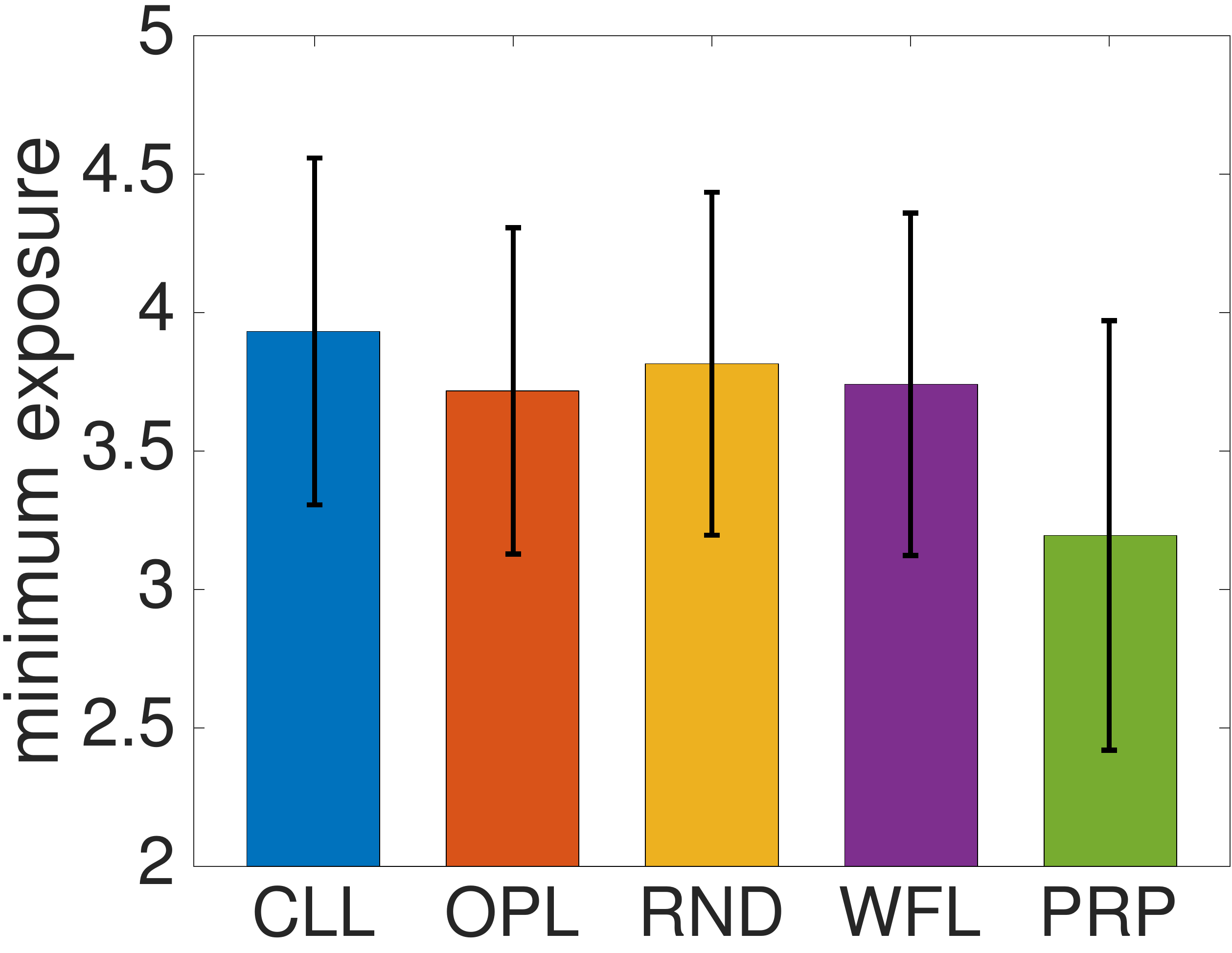}  \\
            e) $M=4$ & f) $M=6$ & g) $M=8$ & h) $M=10$ \\
          \hspace{-5mm}
          \includegraphics[width=0.25\textwidth]{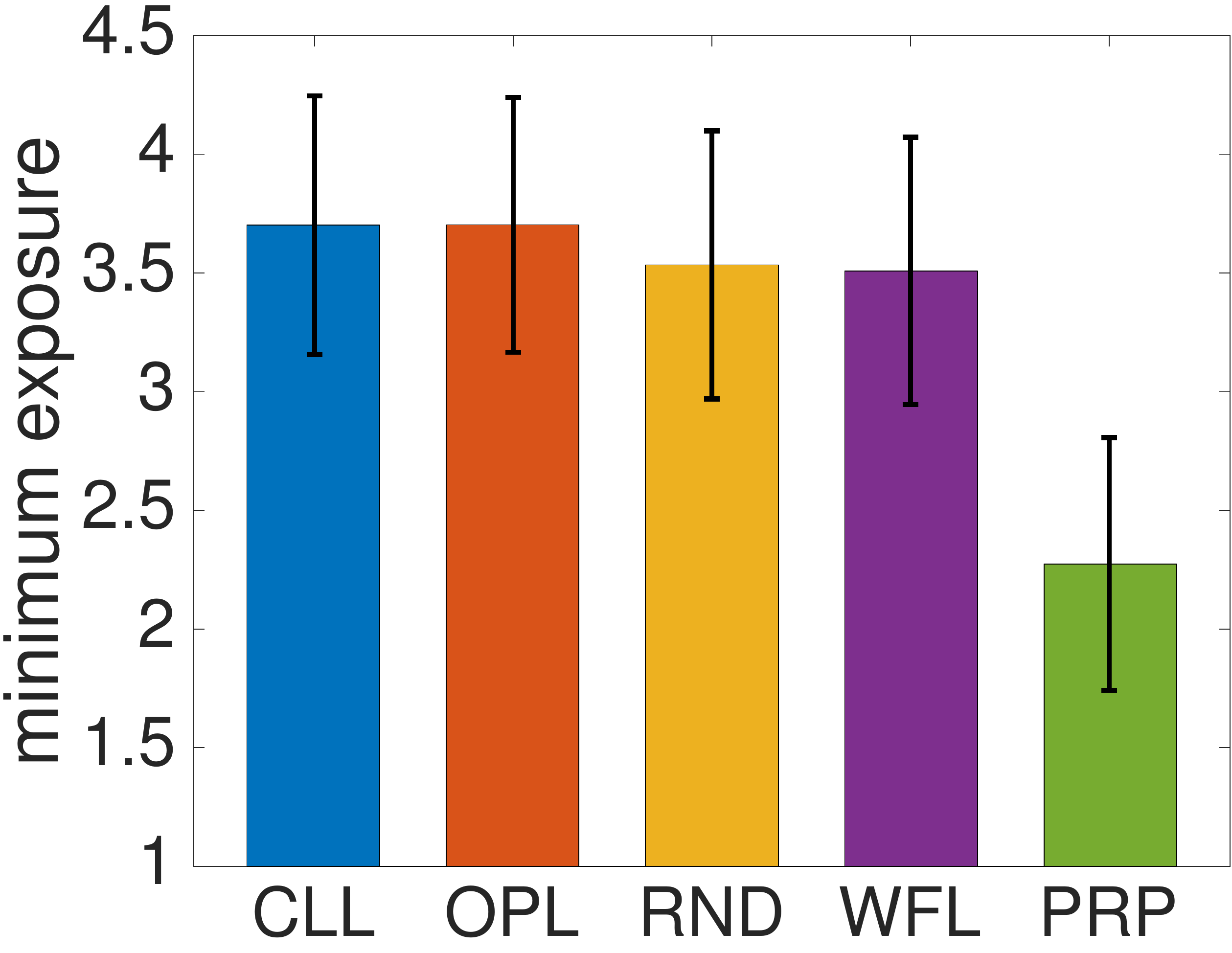} &
        \hspace{-5mm}
          \includegraphics[width=0.25\textwidth]{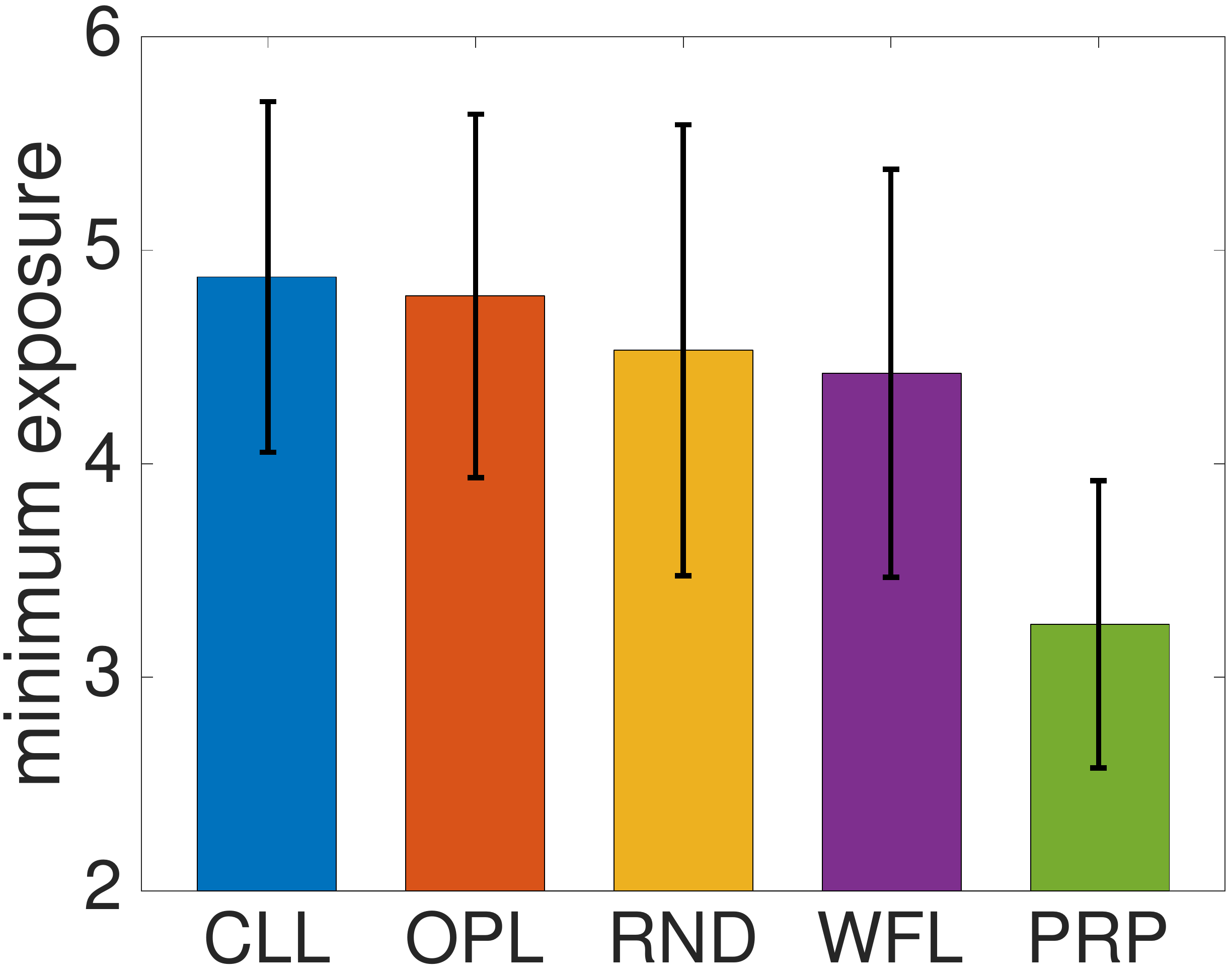} &
          \hspace{-5mm}
          \includegraphics[width=0.25\textwidth]{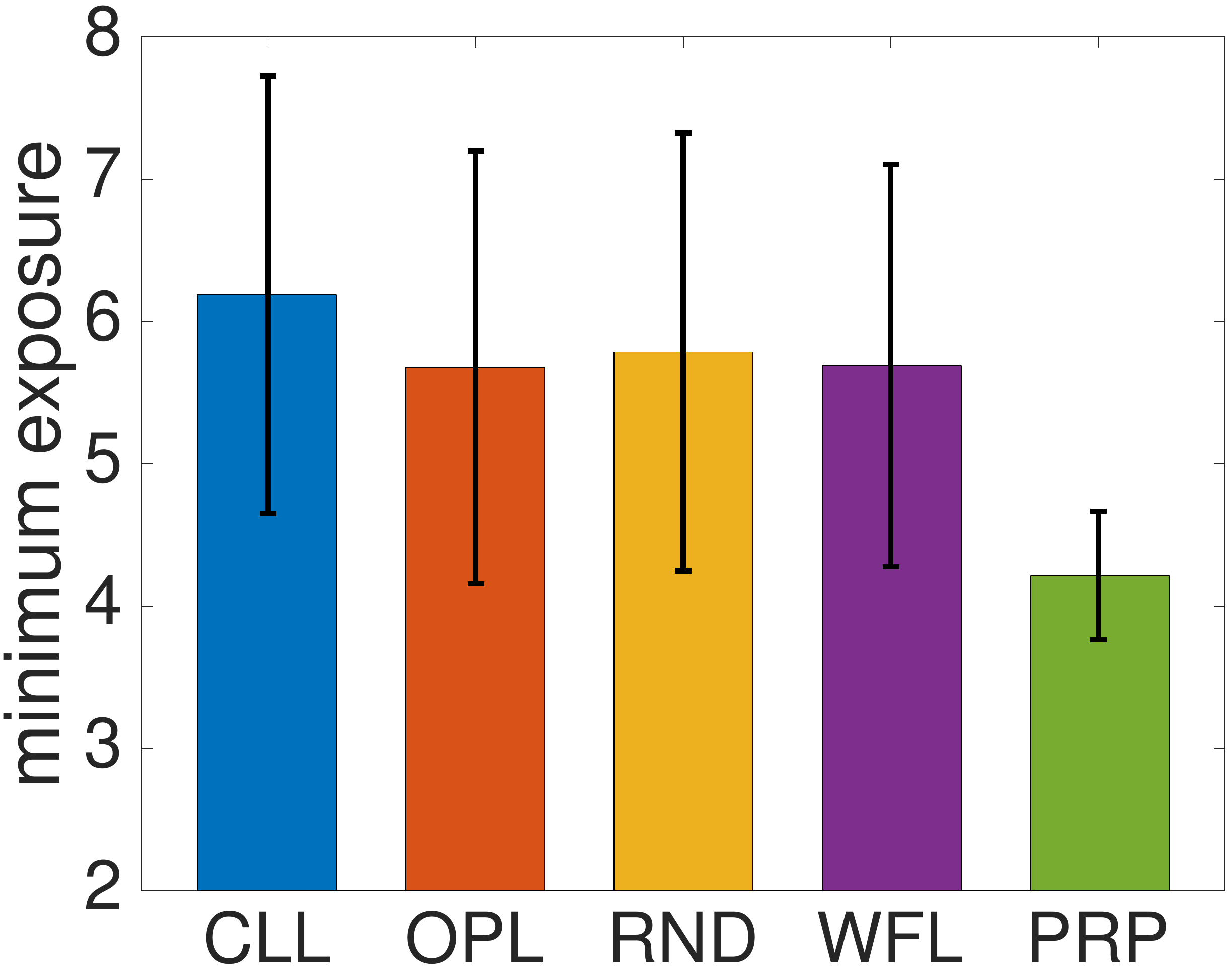} &
          \hspace{-5mm}
          \includegraphics[width=0.25\textwidth]{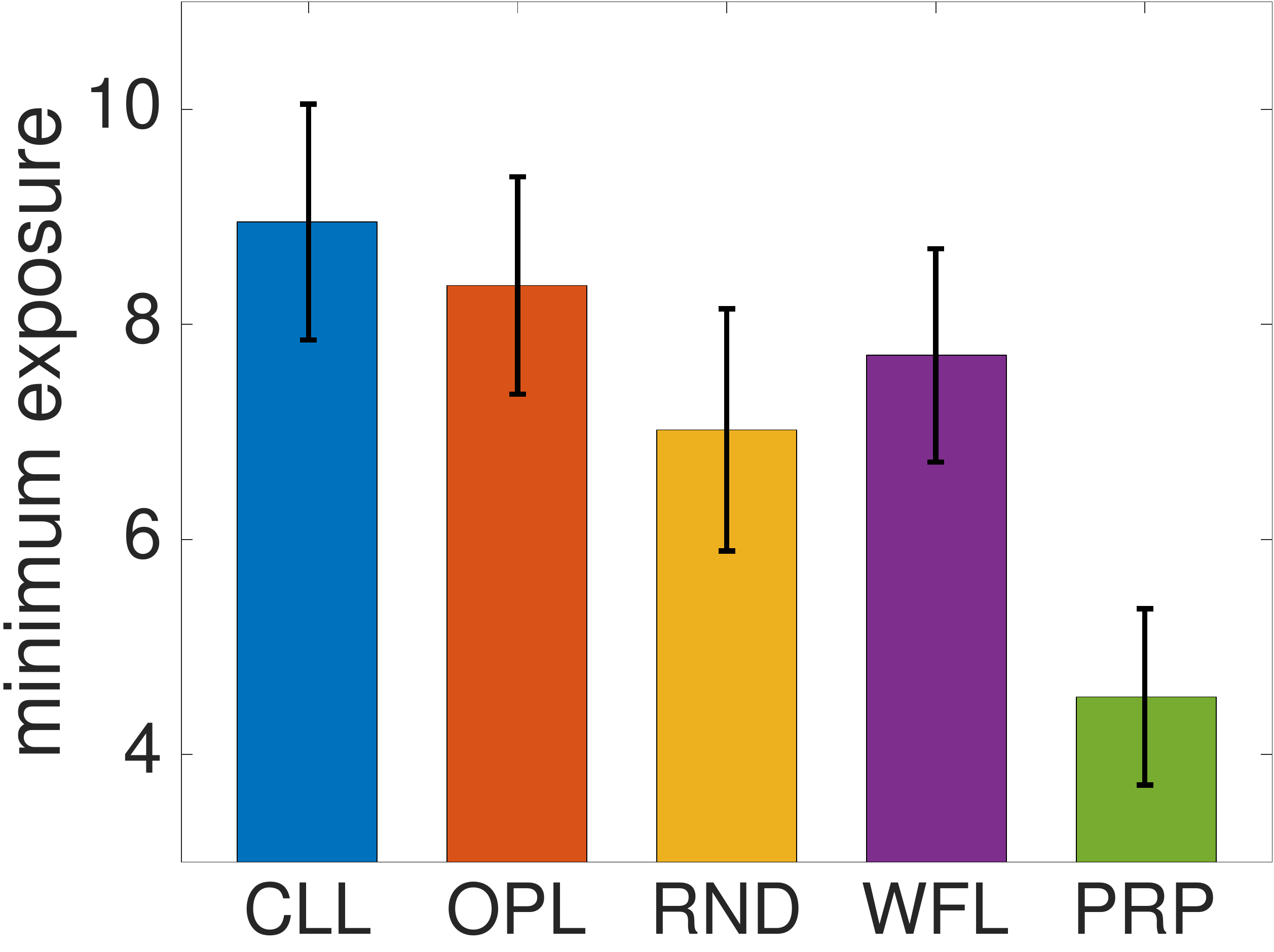}  \\
            i) $\Delta=5$ & j) $\Delta=6.66$ & k) $\Delta=8.33$ & l) $\Delta=10$\\
  \end{tabular}% \vspace{-3mm}
  \caption{Minimum exposure maximization results on synthetic data; top row: $n$ varies, $M=6$, $T=40$; 
  middle row: $M$ varies, $T=40$, $n=200$; bottom row: $T$ varies, $n=200$, $M=6$}
  \label{fig:synth-mmesh-results}
\end{figure*}

\begin{figure*}[!t]
  % \vspace{-3mm}
  \centering
  \setlength{\tabcolsep}{6pt}
  \begin{tabular}{cccc}
          \hspace{-5mm}
          \includegraphics[width=0.25\textwidth]{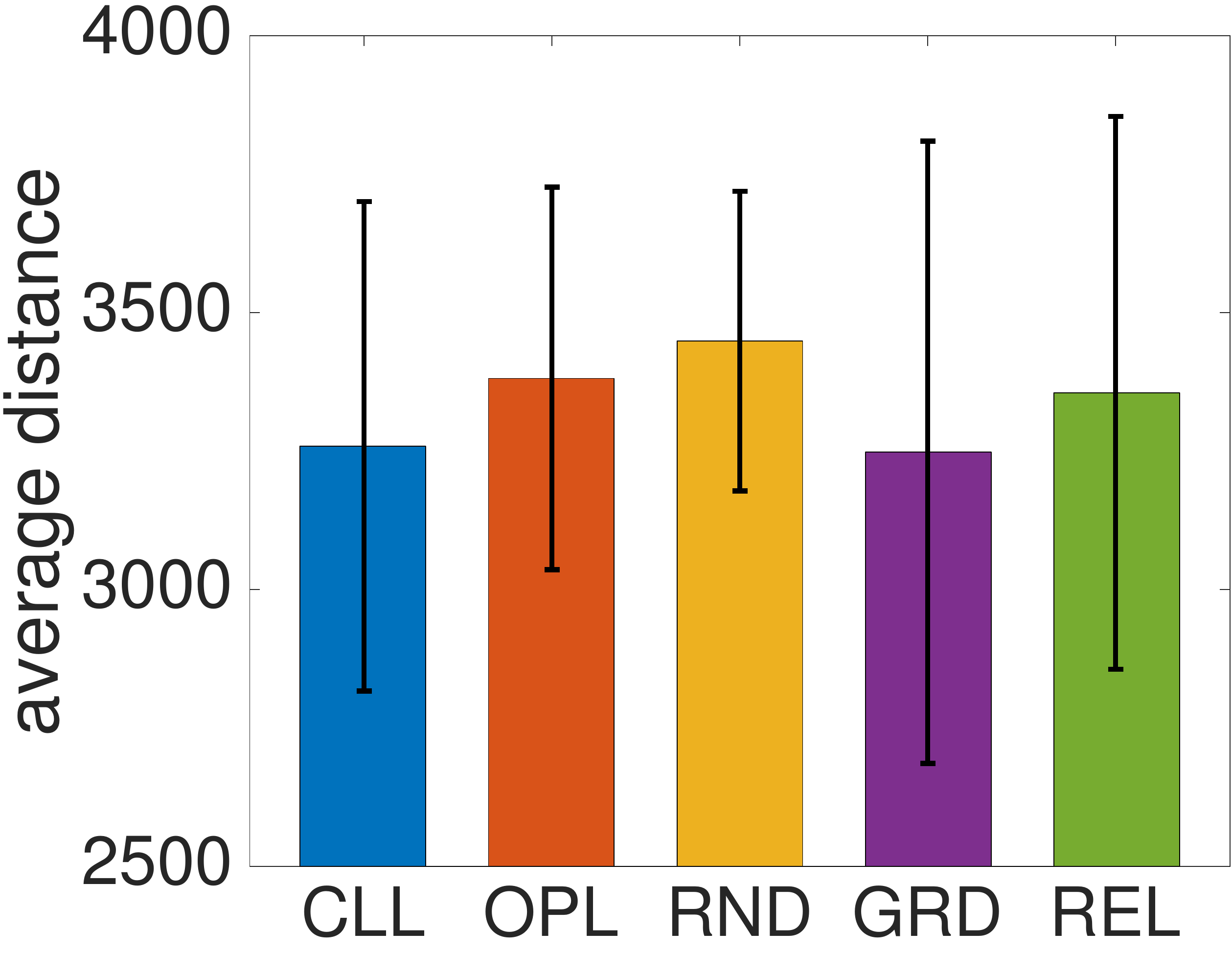} &
          \hspace{-5mm}
          \includegraphics[width=0.25\textwidth]{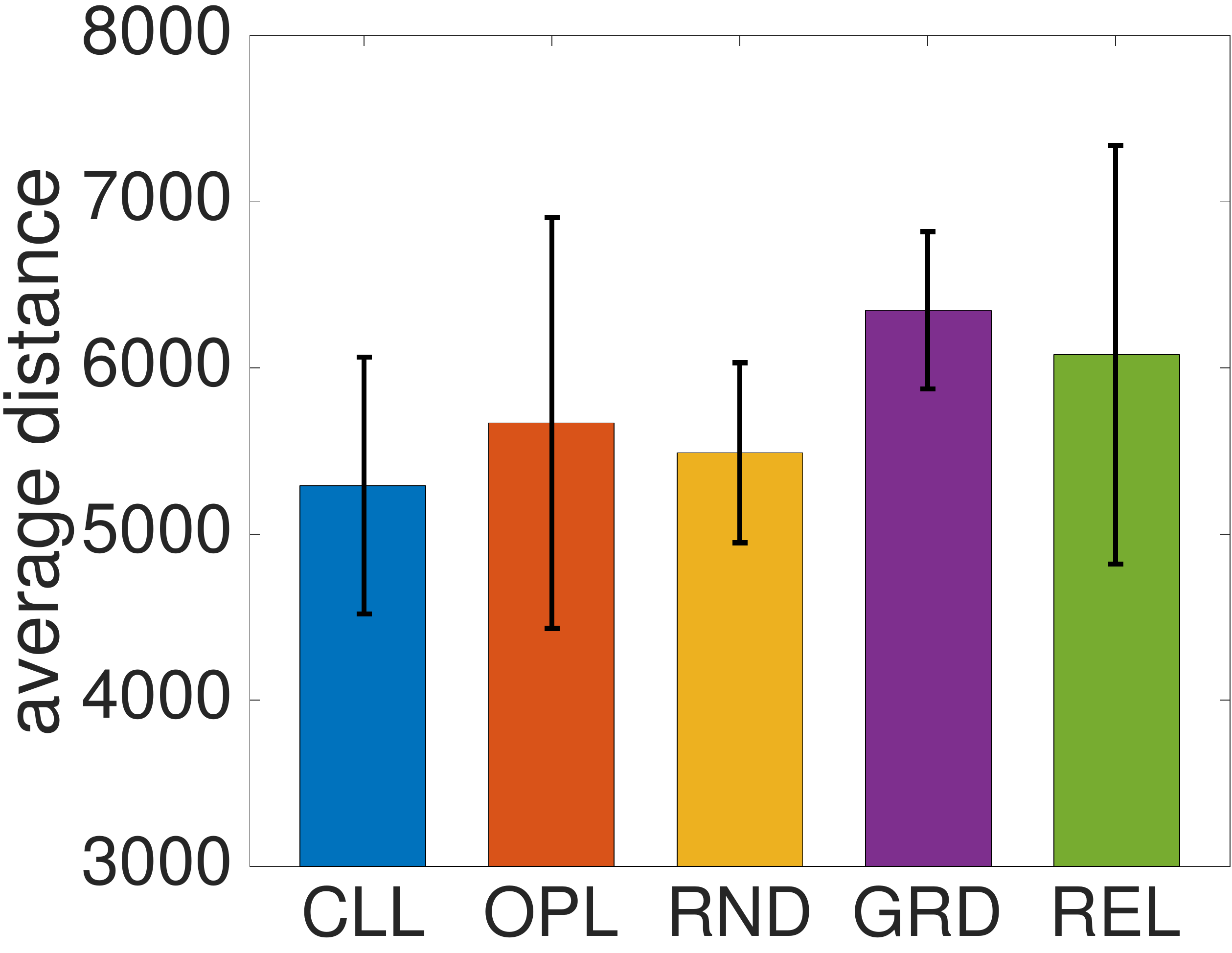} &
          \hspace{-5mm}
          \includegraphics[width=0.25\textwidth]{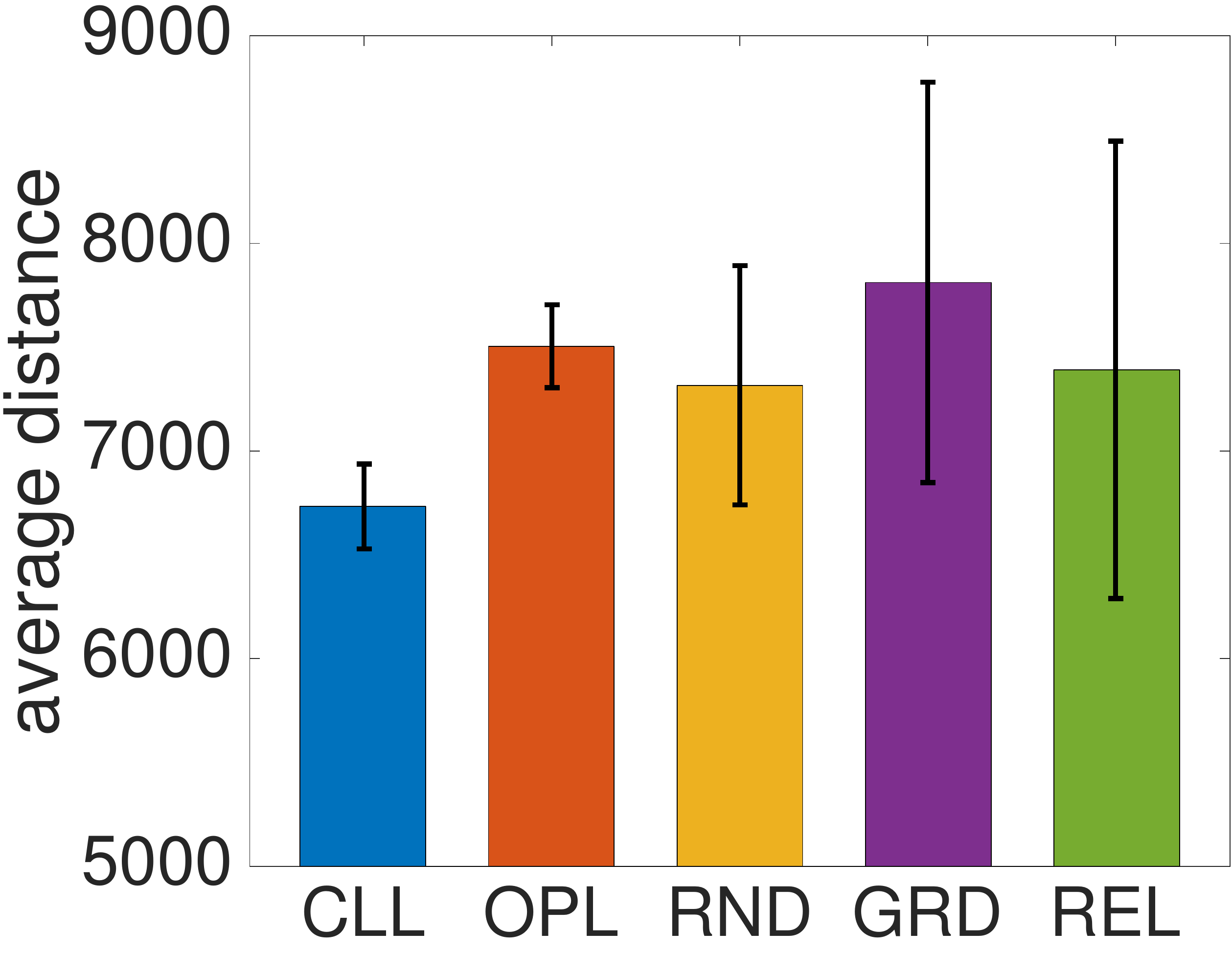} &         
          \hspace{-5mm}
          \includegraphics[width=0.25\textwidth]{fig-synth-LSESH-size-6} \\
          a) $n=150$ & b) $n=200$ & c) $n=250 $ & d) $n=300$ \\
          \hspace{-5mm}
          \includegraphics[width=0.25\textwidth]{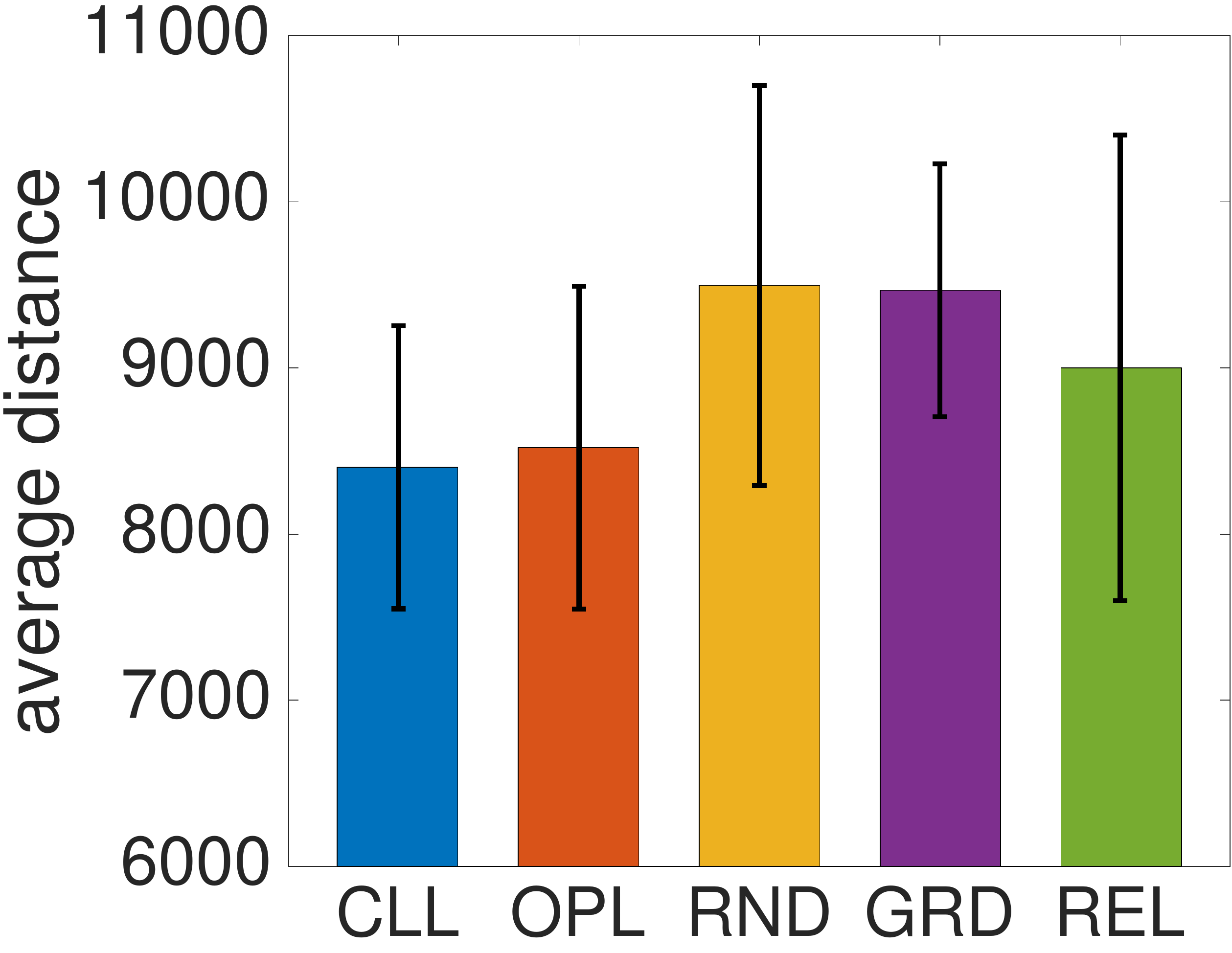} &
        \hspace{-5mm}
          \includegraphics[width=0.25\textwidth]{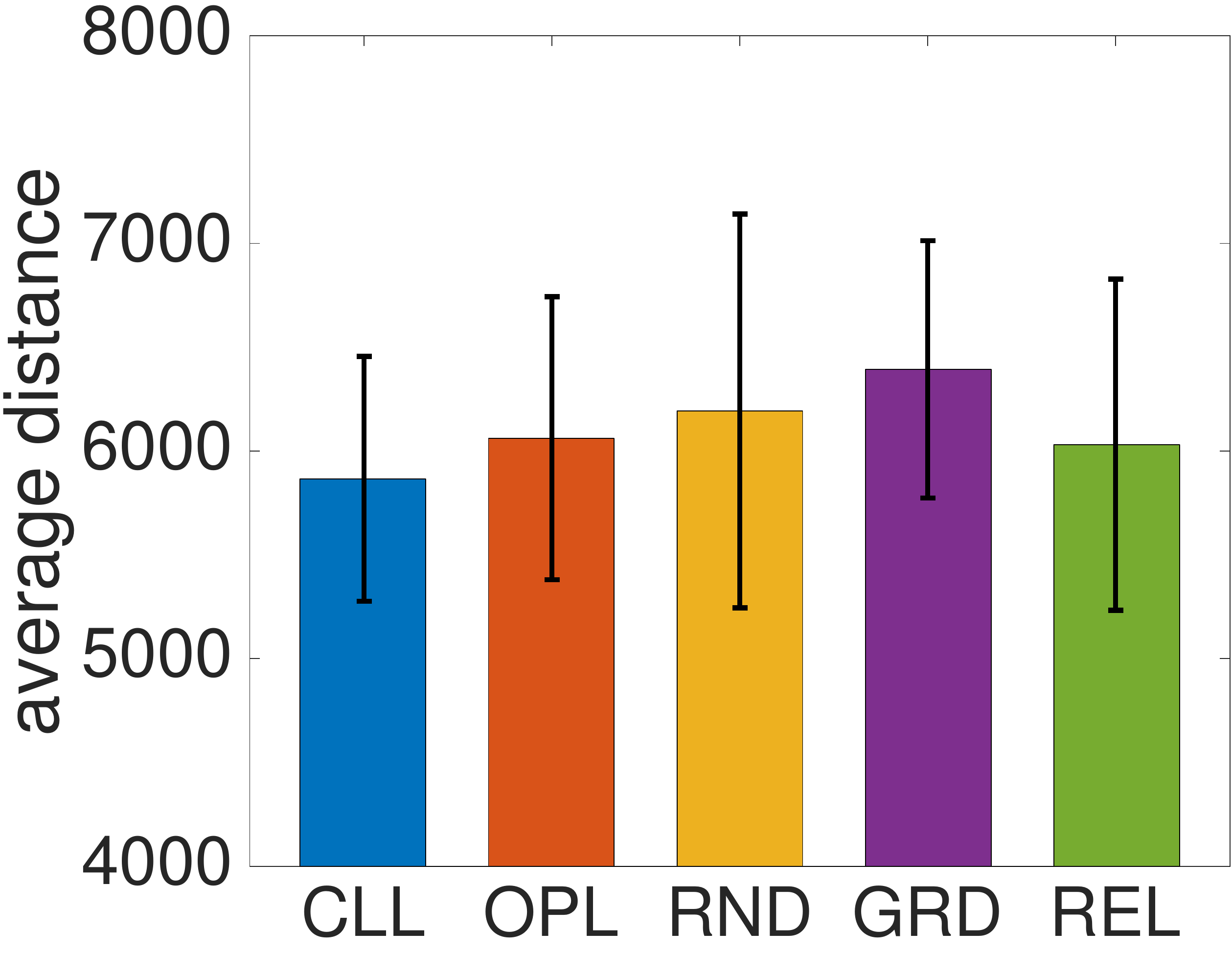} &
          \hspace{-5mm}
          \includegraphics[width=0.25\textwidth]{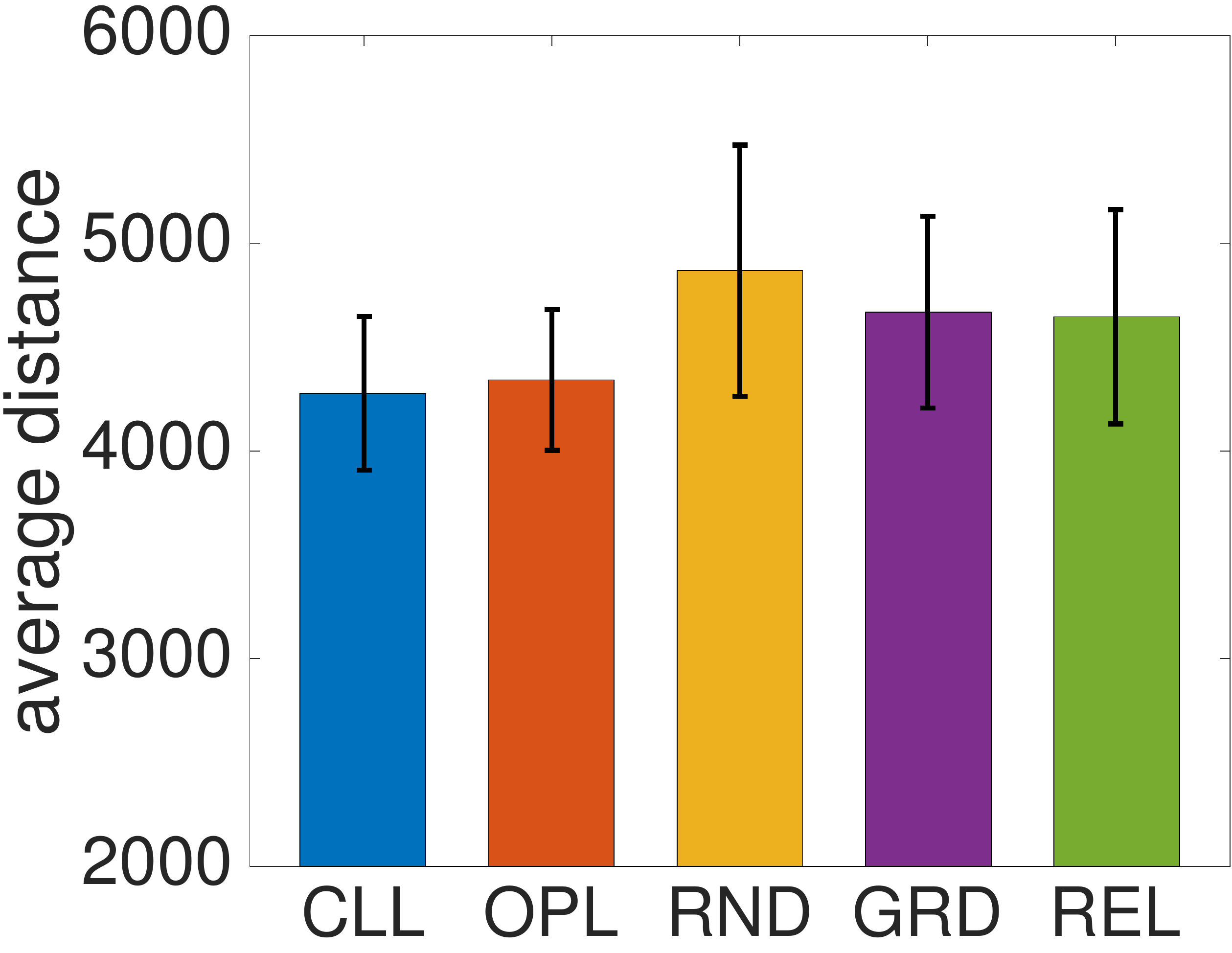} &
          \hspace{-5mm}
          \includegraphics[width=0.25\textwidth]{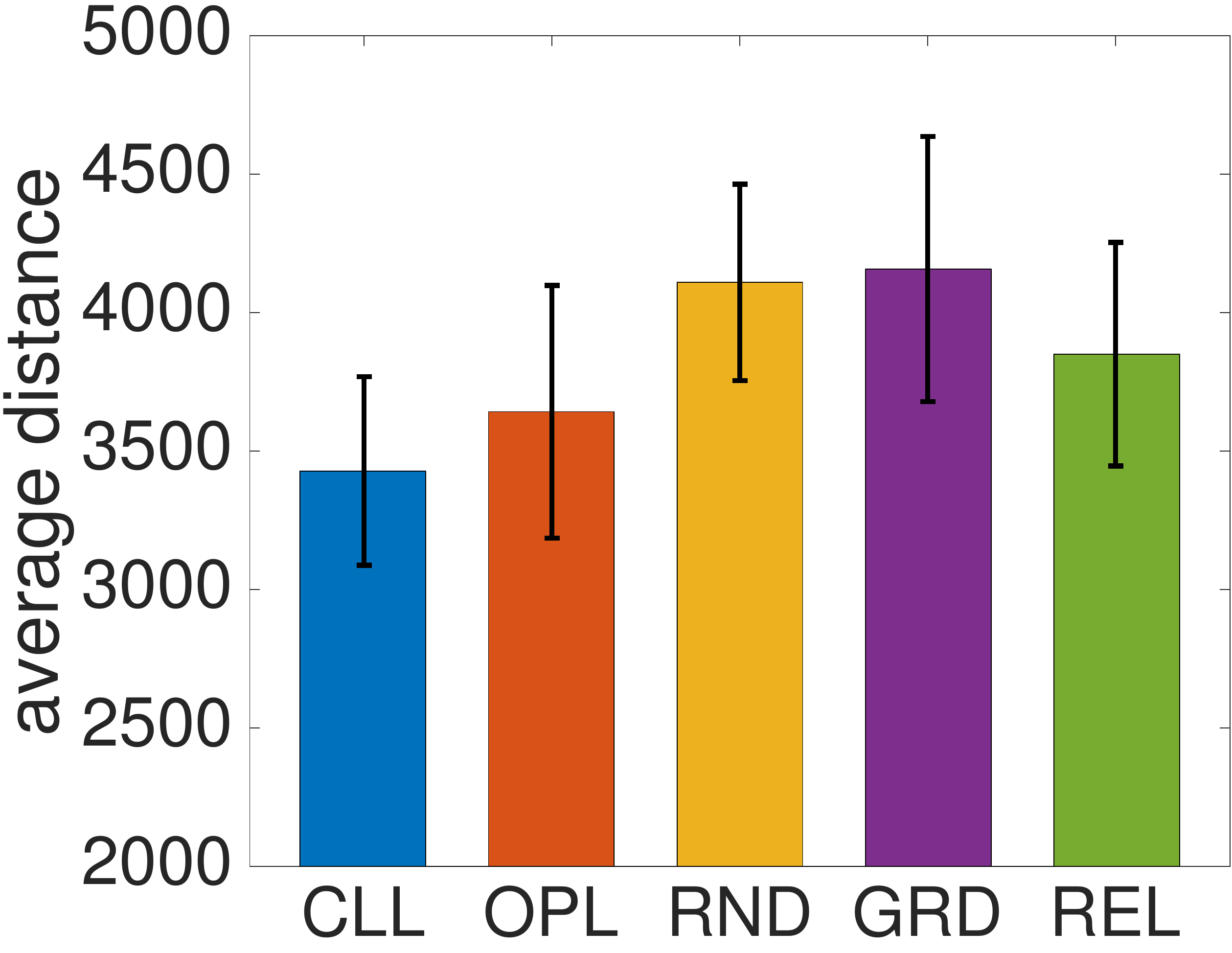}  \\
            e) $M=4$ & f) $M=6$ & g) $M=8$ & h) $M=10$ \\
          \hspace{-5mm}
          \includegraphics[width=0.25\textwidth]{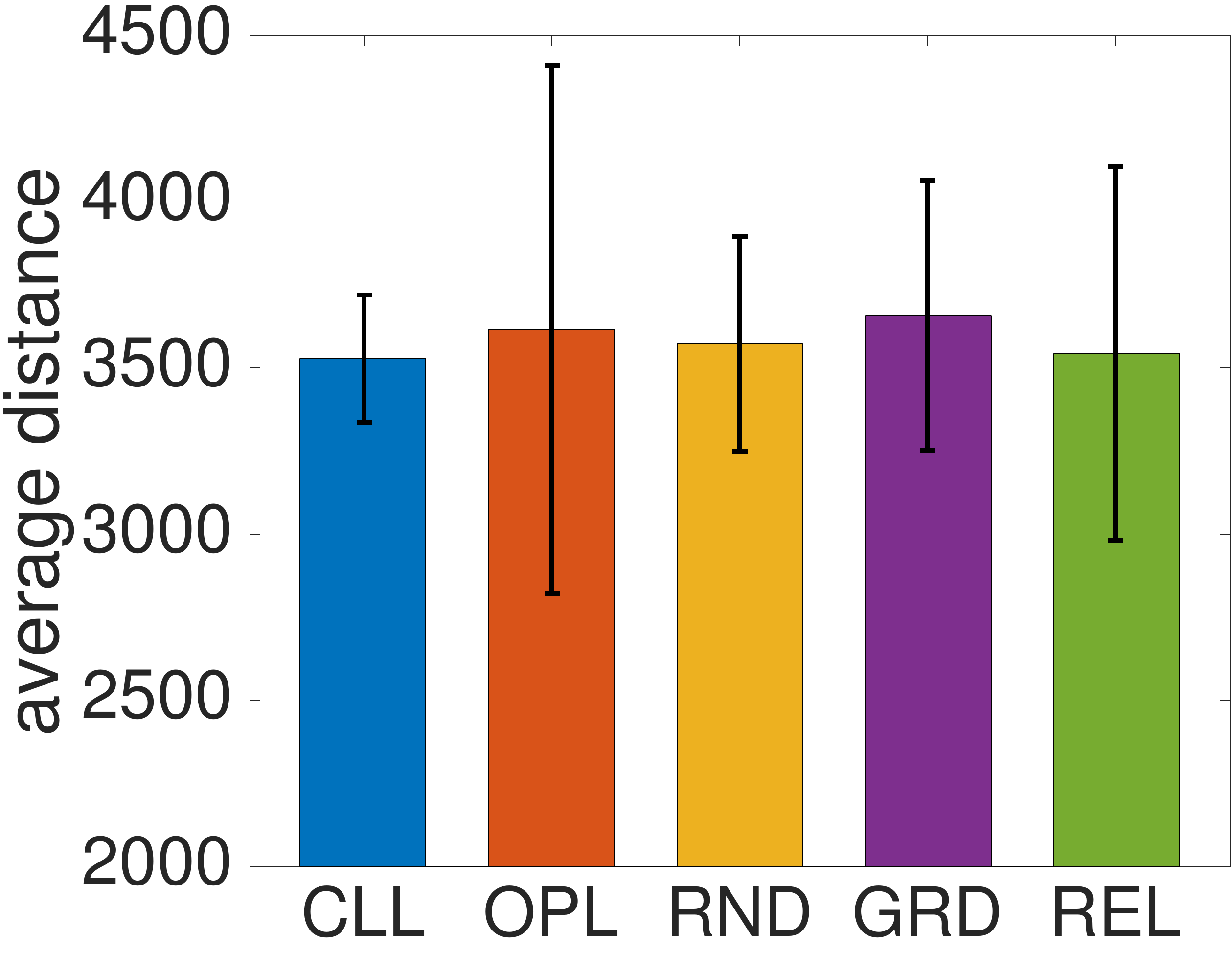} &
        \hspace{-5mm}
          \includegraphics[width=0.25\textwidth]{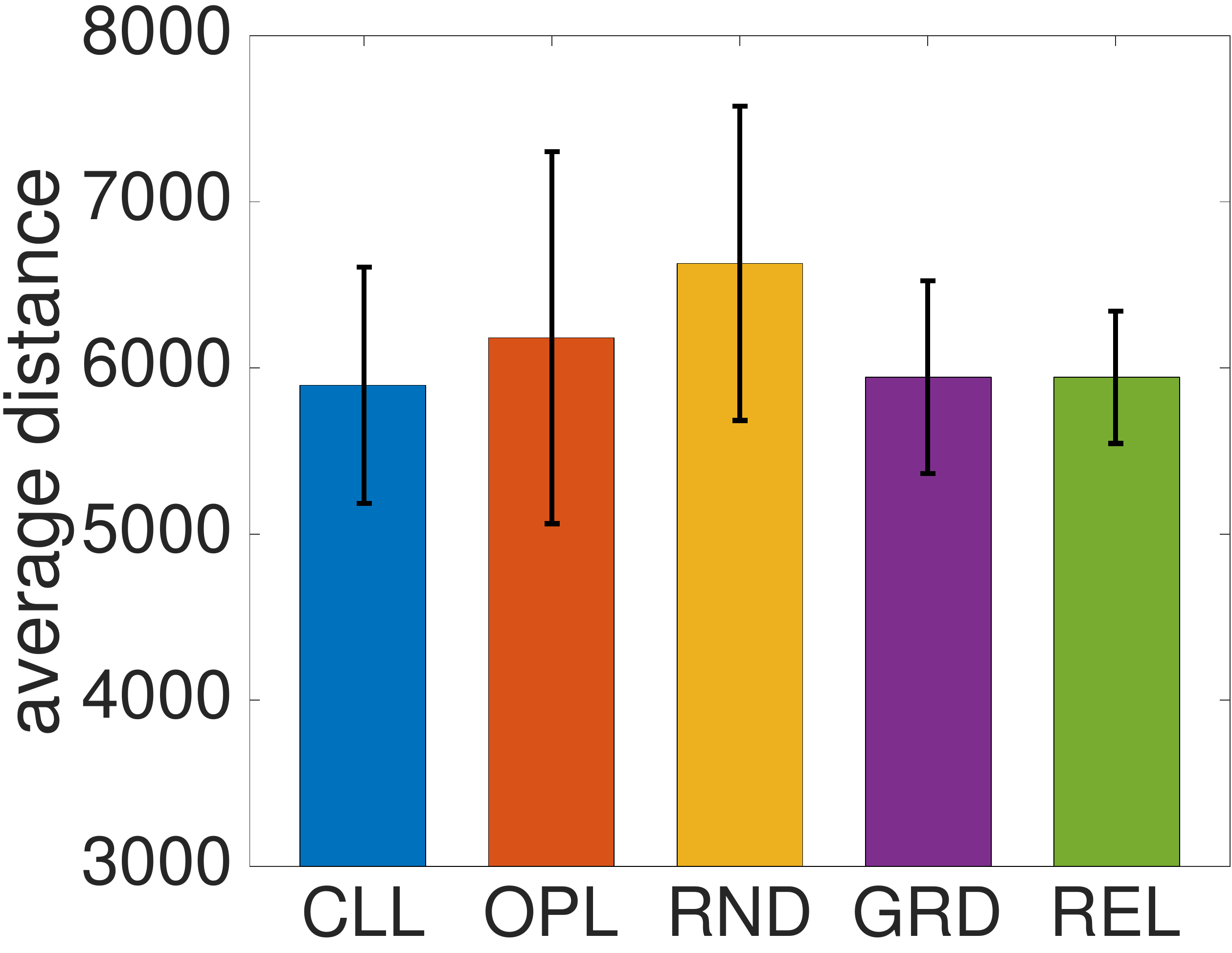} &
          \hspace{-5mm}
          \includegraphics[width=0.25\textwidth]{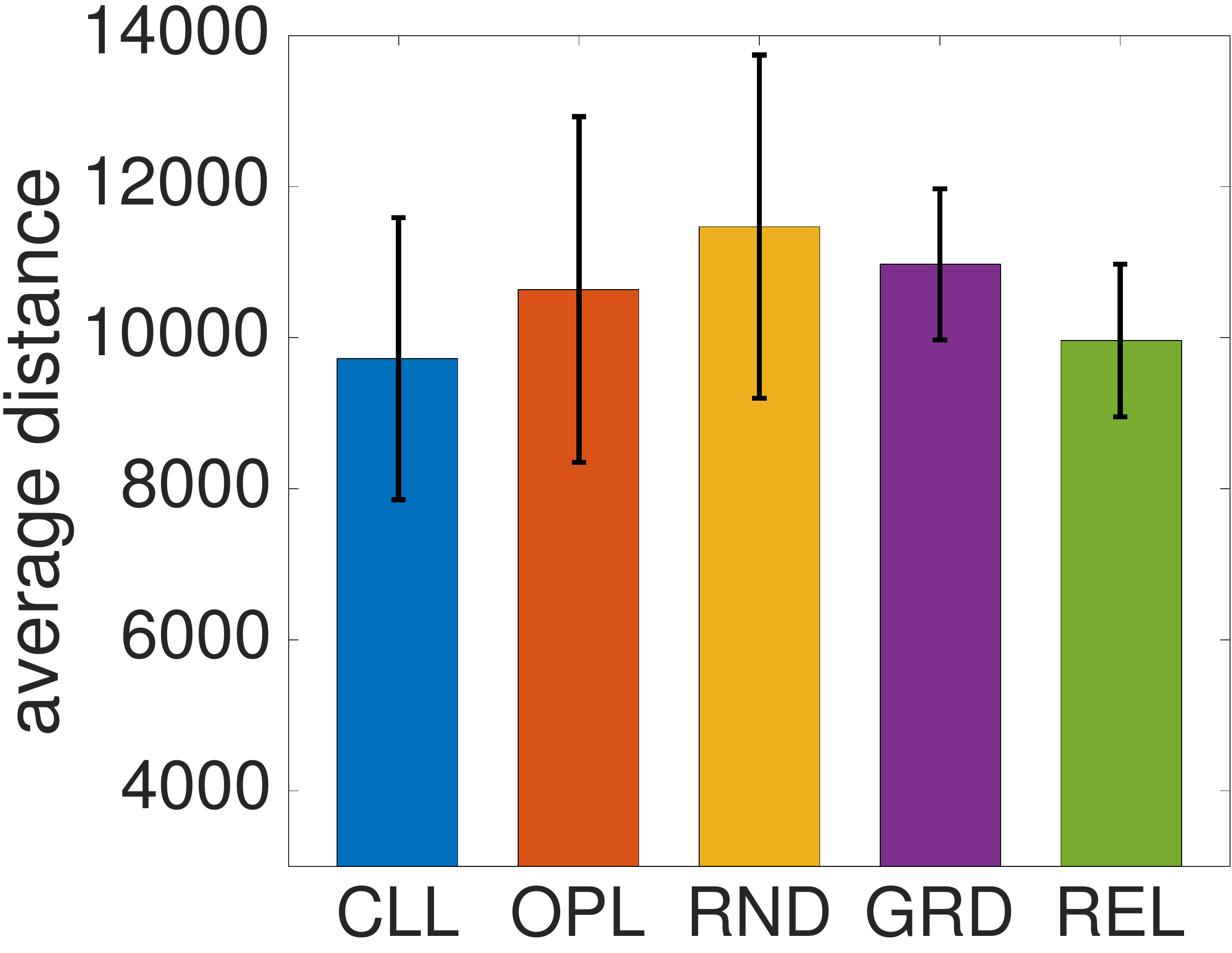} &
          \hspace{-5mm}
          \includegraphics[width=0.25\textwidth]{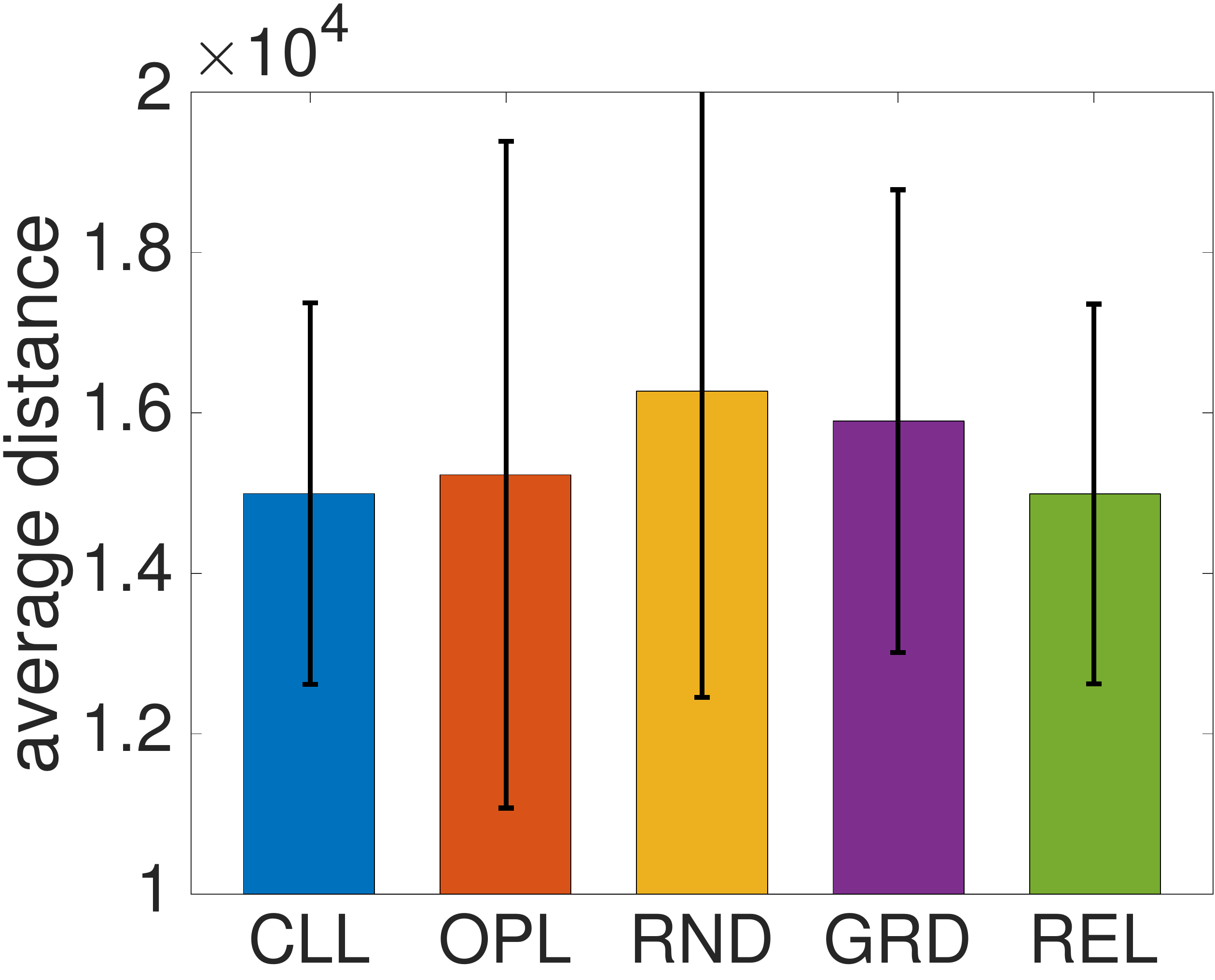}  \\
            i) $\Delta=5$ & j) $\Delta=6.66$ & k) $\Delta=8.33$ & l) $\Delta=10$\\
  \end{tabular}% \vspace{-3mm}
  \caption{Least-squares exposure shaping results on synthetic data; top row: $n$ varies, $M=6$, $T=40$; 
  middle row: $M$ varies, $T=40$, $n=200$; bottom row: $T$ varies, $n=200$, $M=6$}
  \label{fig:synth-lsesh-results}
\end{figure*}

\section{Extended synthetic results}
For the synthetic case we can freely evaluate the properties of the proposed algorithm under several conditions. We assess the performance of the algorithm and compare to the baselines in three settings:  i) increasing size of the network; ii) increasing number of intervention points; iii) increasing the time window (or equivalently the stage duration). The results are reported while keeping other parameters fixed. To compare it to the others we simulate the network with the prescribed intervention intensity and compute the objective function. The mean and standard deviation of the objective function out of 10 runs are reported.  

Figures \ref{fig:synth-cem-results}, \ref{fig:synth-mmesh-results}, and \ref{fig:synth-lsesh-results} shows the results for CEM, MEM, and LES respectively. In each figure, the first row is for varying number of nodes, the second row is for varying number of intervention points, and the third row is for varying duration of stages.
The proposed method is consistently better than the baselines. The trends and facts reported in the main paper are observed in this extended experiment. Additionally, we want to refer the high variance of baseline methods especially RND and OPL which is what we expect.

\end{document}